\documentclass[10pt,a4paper]{article}
\usepackage[left=3cm,right=3cm]{geometry}

\usepackage{xr-hyper}
\usepackage{authblk}
\usepackage[utf8]{inputenc} 
\usepackage[T1]{fontenc}    

\usepackage{wrapfig,lipsum,booktabs}

\usepackage{paralist}
\usepackage{booktabs}       
\usepackage{nicefrac}       
\usepackage{multirow}
 \usepackage{xargs}

\usepackage[tbtags]{amsmath}
\usepackage{amsthm}
\usepackage{bm}
\allowdisplaybreaks
\usepackage{amssymb,mathrsfs}
\usepackage{amsfonts}
\usepackage{mathtools}
\usepackage{upgreek}
\usepackage{graphicx}
\usepackage{wrapfig}
\usepackage[dvipsnames]{xcolor}
\usepackage{soul}
\usepackage{pifont}
\usepackage{bbm}
\usepackage[colorlinks = true, citecolor = blue]{hyperref}
\usepackage{algpseudocode,algorithm,algorithmicx}
\usepackage{stmaryrd}
\usepackage{array}
\usepackage[inline]{enumitem}

\usepackage{tikz}
\usepackage{pgfplots}
\usepackage{aliascnt}
\usepackage{todonotes}
\usepackage{cleveref}
\usepackage{autonum}
\usepackage{accents}
\usepackage{etaremune}
\DeclareMathAlphabet{\mathpzc}{OT1}{pzc}{m}{it}

\definecolor{lightred}{rgb}{1, 0.8, 0.8}

\makeatletter
\newtheorem{theorem}{Theorem}
\crefname{theorem}{theorem}{Theorems}
\Crefname{Theorem}{Theorem}{Theorems}

\newtheorem*{lemma_nonumber*}{Lemma}

\newtheorem{lemma}{Lemma}

\newaliascnt{corollary}{theorem}
\newtheorem{corollary}[corollary]{Corollary}
\aliascntresetthe{corollary}
\crefname{corollary}{corollary}{corollaries}
\Crefname{Corollary}{Corollary}{Corollaries}

\newaliascnt{proposition}{theorem}
\newtheorem{proposition}[proposition]{Proposition}
\aliascntresetthe{proposition}
\crefname{proposition}{proposition}{propositions}
\Crefname{Proposition}{Proposition}{Propositions}

\newaliascnt{definition}{theorem}
\newtheorem{definition}[definition]{Definition}
\aliascntresetthe{definition}
\crefname{definition}{definition}{definitions}
\Crefname{Definition}{Definition}{Definitions}

\newaliascnt{remark}{theorem}
\newtheorem{remark}[remark]{Remark}
\aliascntresetthe{remark}
\crefname{remark}{remark}{remarks}
\Crefname{Remark}{Remark}{Remarks}

\crefname{example}{example}{examples}
\Crefname{Example}{Example}{Examples}

\crefname{figure}{figure}{figures}
\Crefname{Figure}{Figure}{Figures}

\newtheorem{assumption}{\textbf{H}\hspace{-3pt}}
\crefformat{assumption}{{\textbf{H}}#2#1#3}

\crefformat{assumptionsup}{{\textbf{A}}#2#1#3}

\def\msf{\mathsf{G}}

\def\msw{\mathsf{W}}

\def\msu{\mathsf{U}}

\def\complementary{\mathrm{c}}
\def\msi{\mathsf{I}}
\def\msj{\mathsf{J}}
\def\msa{\mathsf{A}}

\def\msk{\mathsf{K}}

\def\msb{\mathsf{B}} 
\def\msc{\mathsf{C}}
\def\mse{\mathsf{E}}
\def\msf{\mathsf{F}}

\def\msh{\mathsf{H}}

\def\msm{\mathsf{M}}
\def\msu{\mathsf{U}}
\def\msv{\mathsf{V}}

\def\msq{\mathsf{Q}}
\def\msx{\mathsf{X}}
\def\msz{\mathsf{Z}}

\def\mcs{\mathcal{S}}

\def\mcbb{\mathcal{B}}  
\newcommand{\mcb}[1]{\mathcal{B}(#1)}

\def\rset{\mathbb{R}}

\def\zset{\mathbb{Z}}
\def\nset{\mathbb{N}}
\def\nsets{\mathbb{N}^*}

\def\Rset{\mathbb{R}}

\def\Zset{\mathbb{Z}}
\def\Nset{\mathbb{N}}

\def\rml{\mathrm{L}}

\def\rmC{\mathrm{C}}

\newcommand{\abs}[1]{\left\vert #1 \right\vert}

\newcommandx{\psr}[3][3=]{\left\langle#1,#2 \right\rangle_{#3}}
\newcommandx{\normr}[2][2=]{ \left\Vert#1 \right\Vert_{#2}}
\newcommandx{\psrLigne}[3][3=]{\langle#1,#2 \rangle_{#3}}
\newcommandx{\normrLigne}[2][2=]{ \Vert#1 \Vert_{#2}}

\newcommandx{\norm}[2][1=]{\ifthenelse{\equal{#1}{}}{\left\Vert #2 \right\Vert}{\left\Vert #2 \right\Vert^{#1}}}

\newcommand{\defEns}[1]{\left\lbrace #1 \right\rbrace }

\newcommand\probaMarkovTilde[2][2=]
{\ifthenelse{\equal{#2}{}}{{\widetilde{\mathbb{P}}_{#1}}}{\widetilde{\mathbb{P}}_{#1}\left[ #2\right]}}

\newcommand{\plusinfty}{+\infty}

\def\eqsp{\;}

\newcommand\sequence[3][2=,3=]
{\ifthenelse{\equal{#3}{}}{\ensuremath{\{ #1_{#2}\}}}{\ensuremath{\{ #1_{#2}, \eqsp #2 \in #3 \}}}}
\newcommand\sequenceD[3][2=,3=]
{\ifthenelse{\equal{#3}{}}{\ensuremath{\{ #1_{#2}\}}}{\ensuremath{( #1)_{ #2 \in #3} }}}

\newcommand\sequenceDouble[4][3=,4=]
{\ifthenelse{\equal{#3}{}}{\ensuremath{\{ (#1_{#3},#2_{#3}) \}}}{\ensuremath{\{  (#1_{#3},#2_{#3}), \eqsp #3 \in #4 \}}}}

\def\iid{i.i.d.}

\def\Idd{\mathrm{I}_d}

\newcommand{\1}{\mathbbm{1}}

\def\transpose{\top}

\newcommand{\beq}{\begin{equation}}
\newcommand{\eeq}{\end{equation}}

\def\Leb{\mathrm{Leb}}

\newcommand*{\dd}{\mathop{}\!\mathrm{d}}
\def\ee{~}

\def\orbit{\mathcal{O}}
\def\F{U}

\def\Psiverlet{\Psi}
\newcommandx{\gperthmc}[2][1=,2=]{\ifthenelse{\equal{#1}{}}{\Xi}{\ifthenelse{\equal{#2}{}}{\Xi_{h,#1}}{\Xi_{#2,#1}}}}
\newcommandx{\Phiverlet}[2][1=,2=]{\ifthenelse{\equal{#1}{}}{\Phi}{\Phi_{#1}^{\circ (#2)}}}
\newcommandx{\gpertub}[2][1=,2=]{\ifthenelse{\equal{#1}{}}{g}{g_{#1}^{#2}}}
\newcommandx{\Phiverletq}[2][1=,2=]{\ifthenelse{\equal{#1}{}}{\widetilde{\Phi}}{\widetilde{\Phi}_{#1}^{\circ (#2)}}}
\newcommandx{\Phiverletqi}[2][1=,2=]{\ifthenelse{\equal{#1}{}}{\bar{\Psi}}{\bar{\Psi}_{#1}^{(#2)}}}
\newcommandx{\Pkerhmc}[2][1=,2=]{\ifthenelse{\equal{#1}{}}{\mathrm{P}}{\mathrm{P}_{#1, #2}}}
\newcommandx{\tPkerhmc}[2][1=,2=]{\ifthenelse{\equal{#1}{}}{\tilde{\mathrm{P}}}{\tilde{\mathrm{P}}_{#1, #2}}}
\newcommandx{\PkerhmcD}[2][1=,2=]{\ifthenelse{\equal{#1}{}}{\mathrm{K}}{\mathrm{K}_{#1, #2}}}
\def\rmp{\mathrm{p}}
\def\rmpp{\mathrm{P}}
\def\rmq{\mathrm{q}}
\def\rmqq{\mathrm{Q}}
\def\Kmax{K_{\mathrm{m}}}
\def\lyapD{\mathpzc{V}}

\def\VlyapD{\mathpzc{V}}

\def\gauss{\mathrm{N}}
\def\Kker{\mathrm{K}}
\def\KkerU{\mathrm{K}^{\mathsf{U}}}
\def\pivot{\mathbf{r}}

\def\Khmc{T}
\def\KkerH{\mathrm{K}^{\mathsf{H}}_{h,\Khmc}}
\def\KkerUE{\tilde{\mathrm{K}}^{\mathsf{U}}}
\def\VFL{\VlyapD}

\def\argmax{\text{argmax}}
\def\GaussStandard{\rho_0}
\def\dist{\mathrm{dist}}

\def\scrU{\mathscr{U}}
\def\Sfun{S}
\def\funS{\Sfun}
\def\ltt{\mathtt{L}}

\def\tpi{\tilde{\pi}}

\def\tint{\mathrm{t}_{\mathrm{int}}}
\def\inv{\leftarrow}

\def\Fr{\mathrm{Fr}}
\def\bfA{\mathbf{A}}
\def\bfB{\mathbf{B}}

\newcommand{\fracD}[2]{(#1/#2)}

\def\symset{\mathbb{S}}

\title{On the convergence of dynamic implementations of Hamiltonian Monte Carlo and No U-Turn Samplers}

\author[1]{Alain Durmus \thanks{alain.durmus@polytechnique.ed}} 
\vspace{2mm}

  \author[2]{Samuel Gruffaz \thanks{samuel.gruffaz@ens-paris-saclay.fr}}

  \vspace{2mm}
  \author[3]{Miika Kailas \thanks{miika.p.kailas@jyu.fi}} 

  \vspace{2mm}
  \author[4]{Eero Saksman \thanks{eero.saksman@helsinki.fi}}

  \vspace{2mm}
  \author[3]{Matti Vihola \thanks{matti.s.vihola@jyu.fi}}

  \vspace{2mm}
  \affil[1]{CMAP, CNRS, Ecole Polytechnique, Institut Polytechnique
  de Paris, 91120 Palaiseau, France }
\affil[2]{Université Paris-Saclay, ENS Paris-Saclay, Centre Borelli, F-91190 Gif-sur-Yvette, France.}
\affil[3]{University of Jyväskylä, Department of Mathematics and Statistics.  }
\affil[4]{University of Helsinki, Department of Mathematics and Statistics. }

\begin{document}

\maketitle

\begin{abstract}
    There is substantial empirical evidence about the success of dynamic
    implementations of Hamiltonian Monte Carlo (HMC), such as the No
    U-Turn Sampler (NUTS), in many challenging inference problems but
    theoretical results about their behavior are scarce.  The aim of
    this paper is to fill this gap. More precisely, we consider a
    general class of MCMC algorithms we call dynamic HMC. We show that
    this general framework encompasses NUTS as a particular case, implying
    the invariance of the target distribution as a by-product.  Second,
    we establish conditions under which NUTS is irreducible and
    aperiodic and as a corrolary ergodic.
  Under conditions similar to the ones existing for HMC, we also show that NUTS is geometrically ergodic.
  Finally, we improve existing convergence results for HMC showing that this method is ergodic without any
  boundedness condition on the stepsize and the number of leapfrog
  steps, in the case where the target is a perturbation of a Gaussian
  distribution. 
\end{abstract}

%


\section{Introduction}


In this paper we consider No U-Turn Samplers (NUTS), a class of dynamic implementations of the Hamiltonian Monte Carlo algorithm (HMC).
 HMC is a Metropolis-Hastings algorithm designed to sample from a target probability density $\pi$ on $\Rset^d$.
 This method has a pretty long history beginning from computational physics in 1987 \cite{duane1987hybrid}
  before quickly gaining in popularity inside the statistics community in the early paper of \cite{neal1992bayesian};
   see, for example \cite[chapter 9]{liu2001monte}, \cite{neal-hmc} and \cite{girolami2011riemann}.
 This method is now the main inference engine of popular probabilistic programming languages such as Stan \cite{Stan_libcarpenter2017stan},
  PyMC3 \cite{Salvatier2016} and Turing \cite{ge2018t}.
    The HMC algorithm aims to remove the random-walk behavior that plagues most MCMC algorithms:
     the proposals -- obtained by integrating a system of Hamiltonian equations using the leapfrog integrator --
      are far away from the starting position while still having a high probability of being accepted.

   During the previous decade the challenge to avoid a drop in performance was the appropriate tuning of the  parameters of the leapfrog integrator:
    the stepsize $h>0$ and the number of leapfrog steps $T\in \Nset^*$.
   Indeed, the length of the time interval $hT$ along which the Hamiltonian equations are integrated \cite{beyn2014markov}
    controls the sampler's exploration/exploitation trade-off since it changes the distance between the current state and the proposal.
   One option is to fix $T$ and to estimate $h$ with an adaptive mechanism, see \cite{andrieu2008tutorial} for a review.
    Then, $T$ may be selected by cross validation or by using expert knowledge, depending on the context.

     As a major breakthrough, the first NUTS algorithm using slice sampling was introduced in \cite{hoffman2014no}
      as a variant of HMC which selects $T$ automatically by design
       and which finds $h$ using an adaptive mechanism called dual-averaging \cite{nesterov2009primal}.
     The algorithm implemented in the Stan library \cite{Stan_libcarpenter2017stan} has been further developed and improved, in particular by replacing the slicing procedure with a multinomial mechanism \cite{Betancourt}.
     The main idea is to integrate the Hamiltonian equations until the No-U-turn criterion fails,
     corresponding to the moment where the trajectory turns back to the area where it comes from, 
      with the objective of maximizing the distance of the ending point to the starting point.
     Then, a position is sampled on the resulting trajectory. 
     Moreover, this sampling is designed to leave $\pi$ invariant and to encourage the selection of points in areas of high density (relatively to $\pi$) far from the starting point.

      More generally, different dynamic and adaptive implementations of HMC selecting the integration time according to the context have been proposed,
       to cite a few: HMC with randomized integration time (rHMC, see e.g. \cite{livingstone2019geometric} and references therein) to reduce the effect of ``resonant" behavior,
        ChEES-HMC \cite{hoffman2021adaptive} to allow parallel computations on GPU
         or the recently suggested Apogee-to-Apogee Path Sampler \cite{sherlock2023apogee}, closely related to NUTS, to choose the integration time dynamically.
 There is substantial empirical evidence about their success in many challenging inference problems \cite{tang2014learning, schofield2016model, brauner2021inferring, van2021bayesian, harrer2021doing, yu2022beyond}
   but precise theoretical results about their behaviour (apart from rHMC) are scarce compared to the original HMC \cite{Durmus2017-tf, livingstone2019geometric, byrne2013geodesic, betancourt2017geometric, bou2017randomized, Bou-Rabee2018-ly, Chen2023hmc, Bou-Rabee2023, Gouraud2022}.

        The goal of the present paper is to derive primary theoretical guarantees for NUTS.
        As far as we know, our results are the first to imply the convergence of NUTS. 
         More precisely, our contributions are as follows.
         First, we present a general framework for dynamic HMC algorithms and prove a condition on their reversibility and invariance.
         The condition and its proof are transparent, yet general enough to encompass NUTS \cite{hoffman2014no, Betancourt}. 
         Second, as the primary contribution of our work, we prove the irreducibility of the current Stan implementation of NUTS.
         Classical results depending on the regularity of the transition kernel \cite{mengersen1996rates, roberts1996exponential} and recent results for basic HMC \cite{Durmus2017-tf, livingstone2019geometric} do not apply directly.
         In particular, establishing classical regularity conditions for the transition kernel using a nonzero probability of an one-step transition is ruled out by the construction of NUTS, necessitating the use of global information on HMC trajectories.
          Our irreducibility results (\Cref{thm:ergodicity_compile}) apply without any restrictions on the step size or maximum number of steps to the cases where $-\log \pi$ is real analytic with vanishing Hessian at infinity, 
        or alternatively with extra conditions on the step size under less stringent regularity assumptions.\footnote{
            Note that if $-\log \pi$ grows faster than quadratically in every direction, HMC and its dynamic variants fail to be geometrically ergodic due to the instability of the leapfrog integrator.
            See \Cref{sec:ergodicity-of-dynamic-hmc} for an extended discussion of the growth rates of potentials for which geometric ergodicity is possible.
          }
          The conditions that we consider highlight the regularity of the No-U-Turn stopping rule which is at the heart of the dynamic trajectory selection of NUTS (see \Cref{hyp:3} in \Cref{ergodicity_section}).
          The technical challenges are further discussed at the end of \Cref{sec:ergodicity-of-dynamic-hmc}.
          
         Third, we establish geometric ergodicity of NUTS under similar conditions as the ones considered in \cite{Durmus2017-tf} for HMC, without any additional smoothness assumptions on the potential $-\log \pi$.
         Finally, our considerations of HMC trajectories allow us to weaken the conditions on the stepsize proposed in \cite{Durmus2017-tf} to establish the ergodicity of basic HMC.
         More specifically, we remove any boundedness condition on the stepsize in the case where $\pi$ has the same tail behavior as a Gaussian.

        \paragraph*{Outline}
        The paper is organized as follows.
        In \Cref{sec:dynamic-hmc},
        we present a class of MCMC methods we call Dynamic HMC which  encompasses NUTS and HMC as particular cases. In addition, we provide conditions ensuring that the target distribution is invariant for the resulting Markov kernel. 
        In \Cref{sec:nuts-its-invariance}, we verify that these conditions are met for the NUTS implentation in Stan as an illustrative and comprehensive example.
        Conditions under which the NUTS implementation in Stan is ergodic and $\lyapD$-uniformly geometrically ergodic are presented in \Cref{ergodicity_section,section:ergo_geo}
         respectively.
         Finally, some properties related to the irreducibility of the HMC algorithm, which are of independent interest, are stated in \Cref{section:general_HMC}. 

             All results are more oriented toward a qualitative understanding than a quantitative analysis since the constants are only sometimes tractable.
     Nevertheless, this work can be a starting point for more quantitative analysis. 

\subsection{Notation}

We denote by $\mathcal{P}(\msx)$ the power set of a set $\msx$,
 integer ranges by $[k:l]=\{k,\ldots,l\}\subset \mathcal{P}(\Zset)$ and $ [l]=[1:l]$ with $k,l\in \Nset$, and
  the sets of non-negative and positive real numbers by
 $\mathbb{R}_{\geq0}$and $\mathbb{R}_{>0}$, respectively.
The set  $\rset^d$ is endowed with the Euclidean scalar product $\psr{\cdot}{\cdot}$, the corresponding norm $|\cdot|$ and Borel $\sigma$-field $\mathcal{B}(\mathbb{R}^d)$.    
Denote by $\mathbb{F}(\mathbb{R}^d)$ the set of Borel measurable functions on $\mathbb{R}^d$ and for $f \in \mathbb{F}(\mathbb{R}^d),\|f\|_{\infty}=\sup _{x \in \mathbb{R}^d}|f(x)|$. 
    The Lebesgue measure is denoted by $\Leb$.
    For $\mu$ a probability measure on $(\mathbb{R}^d, \mathcal{B}(\mathbb{R}^d))$ and $f \in \mathbb{F}(\mathbb{R}^d)$ a $\mu$-integrable function, 
    denote by $\mu(f)$ the integral of $f$ with respect to $\mu$. 
    Let $\VlyapD: \mathbb{R}^d \to [1, \infty)$ be a measurable function. 
    For $f \in \mathbb{F}(\mathbb{R}^d)$, the $\VlyapD$-norm of $f$ is given by $\|f\|_{\VlyapD}=\|f / \VlyapD\|_{\infty}$. 
    For two probability measures $\mu$ and $\nu$ on $(\mathbb{R}^d, \mathcal{B}(\mathbb{R}^d))$, the $\VlyapD$-total variation distance of $\mu$ and $\nu$ is defined as
$$
\|\mu-\nu\|_\VlyapD=\sup _{f \in \mathbb{F}(\mathbb{R}^d),\|f\|_\VlyapD \leq 1} \abs{\int_{\mathbb{R}^d} f(x) \mathrm{d} \mu(x)-\int_{\mathbb{R}^d} f(x) \mathrm{d} \nu(x)} \eqsp.
$$
If $\VlyapD \equiv 1$, 
then $\|\cdot\|_{\VlyapD}$ is the total variation denoted by $\|\cdot\|_{\mathrm{TV}}$. 
For any $x \in \mathbb{R}^d$ and $M>0$ we denote by $\mathrm{B}(x, M)$,
the Euclidean ball centered at $x$ with radius $M$.
Denote by $\mathrm{I}_n$ the identity matrix.
    Let $k \geq 1$. Denote by $(\mathbb{R}^d)^{\otimes k}$ the $k^{\text {th }}$ tensor power of $\mathbb{R}^d$, 
    for any $x \in \mathbb{R}^d, y \in \mathbb{R}^{\ell}$, $x \otimes y \in(\mathbb{R}^d)^{\otimes 2}$ the tensor product of $x$ and $y$, and $x^{\otimes k} \in(\mathbb{R}^d)^{\otimes k}$ the $k^{\text {th }}$ tensor power of $x$.
     We equip product spaces with the norm $\|x_1 \otimes \cdots \otimes x_k\|=\sup _{i \in\{1, \ldots, k\}}|x_i|$,
     where $x_1, \ldots, x_k \in \mathbb{R}^d$.
      We let $\mathcal{L}((\mathbb{R}^d)^{\otimes k}, \mathbb{R}^{\ell})$ stand for the set of linear maps from $(\mathbb{R}^n)^{\otimes k}$ to $\mathbb{R}^{\ell}$ and for $\mathrm{L} \in \mathcal{L}((\mathbb{R}^d)^{\otimes k}, \mathbb{R}^{\ell})$,
       we denote by $\|\mathrm{L}\|$ the operator norm of $\rml$.
       Let $f: \mathbb{R}^d \rightarrow \mathbb{R}^{d}$ be a Lipschitz function, namely there exists $C \geq 0$ such that for any $x, y \in \mathbb{R}^d,|f(x)-f(y)| \leq C|x-y|$.
        Then we denote by
       $\|f\|_{\text {Lip }}=\inf \left\{|f(x)-f(y)| /|x-y| \mid x, y \in \mathbb{R}^d, x \neq y\right\}$. Let $k \geq 0$ and $\mathrm{U}$ 
       be an open subset of $\mathbb{R}^d$.
        Denote by $\mathrm{C}^k(\mathsf{U}, \mathbb{R}^{d})$ the set of all $k$ times continuously differentiable funtions from $\mathsf{U}$ to $\mathbb{R}^{d}$.
        Let $\Phi \in \mathrm{C}^k(\mathsf{U}, \mathbb{R}^{d})$.
         Write $\dd^k \Phi: \msu \to \mathcal{L}((\mathbb{R}^d)^{\otimes k}, \mathbb{R}^{\ell})$ for the $k^{\text {th }}$ differential of $\Phi \in \mathrm{C}^k(\mathbb{R}^d, \mathbb{R}^{\ell})$. For $x \in \rset^d$, denote by $\dd^k \Phi(x)$ the $k$-th differential of $\Phi$ at $x$. 
         For smooth enough functions $f: \mathbb{R}^d \to \mathbb{R}$, denote by $\nabla f$ and $\nabla^2 f$ the gradient and the Hessian of $f$ respectively.
         Let $\msa \subset \mathbb{R}^d$.
          We write $\overline{\msa}, \msa^{\circ}$ and $\partial \msa$ for the closure, the interior and the boundary of $\msa$, respectively.
        For any $n_1, n_2 \in \mathbb{N}, n_1>n_2$, we take the convention that $\sum_{k=n_2}^{n_1}=0$.
       We denote  for any not empty sets $\msa,\msc \subset (\Rset^d)^2 $, $\dist((q,p),\msa)=\inf_{(q',p')\in \msa} \dist((q,p),(q',p')) $ and $\dist(\msc,\msa)=\inf_{(q,p)\in \msc} \dist((q,p),\msa)$.
       The space of real matrices with $d\in \Nset^*$ rows and $c\in \Nset^*$ columns is identified with $\rset^{d \times c}$ and the space of square symmetric matrices is denoted by $\symset_d(\Rset)=\{\mathbf{A}\in \rset^{d\times d} \, : \, \bfA^{\top} = \bfA\}$.

\section{Dynamic HMC algorithms}
\label{sec:dynamic-hmc}


The aim of this section is to present a general framework for discussing HMC and its many variants.
We introduce a general HMC scheme that includes the basic HMC algorithm \cite{duane1987hybrid, neal1992bayesian}, its randomized version and the dynamic No U-Turn Sampler \cite{hoffman2014no} as special cases.
Despite its generality, our scheme admits a simple sufficient condition on its constituents for the invariance of the target distribution $\pi$, unifying and simplifying the existing case-by-case analysis of invariance of HMC-type algorithms.
For ease of presentation we start by introducing HMC and the related concepts and objects that are necessary for our analysis.
However, an exhaustive introduction is out of the scope of this work and we refer to \cite{Geo_integratorsbou2018geometric,Betancourt} for more detail and motivation.

\subsection{Hamiltonian Monte Carlo}

We assume that the target distribution $\pi$ admits a positive density (still denoted by $\pi$) with respect to the Lebesgue measure of the form
$\pi(x) \propto \exp(-U(x))$ with a twice differentiable potential function $U: \Rset^d \to \Rset$.
We define the extended target distribution $\tpi = \pi\otimes \gauss(0,\Idd)$, with density with respect to the Lebesgue (still denoted by $\tpi$), $\tpi(q,p)\propto \exp(-H(q,p))$ with the Hamiltonian function $H: (\rset^d)^2 \to \rset$ given by\begin{equation}
    \label{eq:def_ham}
    H(q, p) = U(q) + \tfrac{1}{2} p^\top p \eqsp.
\end{equation}
We assume that the potential $U$ satisfies the following.

\begin{assumption}
\label{hyp:regularity}
    $U$ is continuously twice differentiable on $\Rset^d$ and the map $q \mapsto \nabla U(q)$ is $\ltt_1$-Lipschitz: for any $q,q'\in\rset^d$,
\begin{equation}
  \label{eq:3}
    | \nabla U(q) -  \nabla U(q')| \leq \ltt_1 |q-q'| \eqsp.
\end{equation}
\end{assumption}
The HMC algorithm and its extensions rely on the Hamiltonian dynamics associated with $U$, defined by Hamilton's equations
\begin{equation}
    \label{eq:hamiltonian_system}
    \frac{\dd q_t}{\dd t} = \frac{\partial H}{\partial p} (q_t,p_t) =  p_t\eqsp,
    \quad
    \frac{\dd p_t}{\dd t}=-\frac{\partial H}{\partial q} (q_t,p_t)= -\nabla U(q_t)\eqsp.
\end{equation}
Under \Cref{hyp:regularity} any initial condition $(q_0, p_0)$ gives rise to an unique solution to \eqref{eq:hamiltonian_system} and moreover it is well-known (see e.g., \cite{Geo_integratorsbou2018geometric}) that the Hamiltonian dynamics preserves the extended target distribution $\tpi$ in the sense that $\tpi(q_t, p_t) = \tpi(q_0, p_0)$ for any initial condition $(q_0, p_0)$ and the associated solution $(q_t, p_t)_{t \geq 0}$.
Since Hamiltonian dynamics also preserves the volume of the phase space $(\Rset^d)^2$, it follows that if the initial condition is drawn randomly as $(q_0, p_0) \sim \tpi$ then also $(q_t, p_t) \sim \tpi$ for all $t \geq 0$.
Given a fixed time horizon $t_{f}>0$ and a sequence of \iid~$\gauss(0,\Idd)$ random variables $(G_k)_{k\in\nset}$, then the \emph{ideal} HMC algorithm consists in defining a Markov chain $(Q_k^{\mathrm{ideal}},P_k^{\mathrm{ideal}})_{k \in\nset}$ such that for $k=1, 2, \dots$, $(Q_k^{\mathrm{ideal}},P_k^{\mathrm{ideal}})$ is the solution of \eqref{eq:hamiltonian_system} at $t_f$ starting from $(Q_{k-1}^{\mathrm{ideal}},G_{k-1})$.
The marginal chain $(Q_k^{\mathrm{ideal}})_{k \in\nset}$ then targets the desired probability distribution $\pi$.

Simulating the ideal Markov chain $(Q_k^{\mathrm{ideal}}, P_k^{\mathrm{ideal}})$ exactly is computationally infeasible since in all but few special cases Hamilton's equations \eqref{eq:hamiltonian_system} need to be solved numerically.
The most common implementation of HMC uses the leapfrog integrator associated with \eqref{eq:hamiltonian_system}. 
Given a time step $h > 0$ and a current point $(q_0,p_0) \in (\rset^d)^2$, one leapfrog step is defined as
\begin{equation}
\label{eq:iteration_verlet}
  (q_{1},p_{1})  = \Phiverlet^{(1)}_h(q_0,p_0) \eqsp, \quad \Phiverlet^{(1)}_h =  \Psiverlet^{(1)}_{h/2} \circ \Psiverlet^{(2)}_h \circ \Psiverlet^{(1)}_{h/2} \eqsp,
\end{equation}
where for each $t \in \Rset_{\geq0}$, the momentum and position update maps $\Psiverlet^{(1)}_t, \Psiverlet^{(2)}_t :(\Rset^d)^2 \to (\Rset^d)^2$ are given by
\begin{equation}
    \label{eq:def_Psiverlet_0}
    \Psiverlet^{(1)}_t(q,p) = (q, p-t\nabla \F(q)) \eqsp,
    \quad
    \Psiverlet^{(2)}_t(q,p) = (q+tp, p) \eqsp
\end{equation}
for any $(q,p) \in (\Rset^d)^2$.
Note that $\Phiverlet[h][1]$ is a volume-preserving bijection $(\Rset^d)^2 \to (\Rset^d)^2$, which implies that likewise its inverse $\Phiverlet[h][-1]$ and more generally any iterate $\Phiverlet[h][\ell]$, $\ell \in \Zset$, is also a volume-preserving bijection $(\Rset^d)^2 \to (\Rset^d)^2$.

In contrast to ideal Hamiltonian dynamics the leapfrog integrator $\Phiverlet[h][\ell]$, for any $\ell \in\nsets$, does not preserve the extended target $\tpi$ and to ensure that HMC is invariant for $\tpi$ a Metropolis accept-reject step has to be added as follows.
Given a number of leapfrog steps $T\in\nsets$ and a sequence of \iid~$\gauss(0,\Idd)$ random variables $(G_k)_{k\in\nset}$, then the HMC algorithm consists in defining a Markov chain $(Q_k,P_k)_{k \in\nset}$ such that for $k=1, 2, \dots$, (1) a proposal $(\tilde{Q}_k, \tilde{P}_k) = \Phiverlet[h][T](Q_{k-1},G_k)$ is first generated, (2) which is accepted, i.e., we set $(Q_k,P_k) = (\tilde{Q}_k,\tilde{P}_k)$, with probability $1\wedge \exp[H(Q_{k-1},P_{k-1}) -H(\tilde{Q}_k,\tilde{P}_k)]$ and rejected, i.e., set $(Q_k,P_k) = (Q_{k-1},P_{k-1})$, otherwise.
Again, the marginal chain $(Q_k)_{k \in \nset}$ correctly targets $\pi$.

The choice of the number of leapfrog steps $T$ is crucial for the efficiency of the algorithm.
If the integration time $T h$ is small, the algorithm reduces to a random walk and diffusive exploration due to the resampling of the momentum at every iteration.
On the other hand if $T h$ is large, the approximate Hamiltonian trajectories will loop back to previously explored neighborhoods and the increased computation time yields little benefit.
To fully realize the algorithm's potential of making long moves in the state space while maintaining computational efficiency, it is essential to strike the right balance between these extremes.
These observations highlight the critical role of $T$ in the expected performance of the algorithm.
In particular, as observed for example in \cite[Section 4.3]{Betancourt}, even for simple models it turns out that the optimal integration time at iteration $k$ for ideal HMC depends on the current point of the algorithm.
Adjusting and choosing the integration time $T h$ dynamically based on the current state is the main achievement of the NUTS algorithm \cite{hoffman2014no} that has proven to be remarkably robust and efficient across a wide range of statistical applications.
We defer the detailed presentation of NUTS to \Cref{sec:nuts-its-invariance} and continue here by presenting a generalization that, hopefully, makes the details easier to understand.

\subsection{General framework for dynamic HMC algorithms}
\label{sec:general-dynamic-hmc}

For the scheme in \Cref{def:dynamic-hmc} we need the following concepts and notation.
Let $h > 0$, $\Kmax \in \nset$ and let an \emph{orbit selection kernel}
\begin{equation}
    \rmpp_h = \{ \rmpp_h(\cdot \mid q_0, p_0) : (q_0, p_0) \in (\Rset^d)^2 \} \eqsp
\end{equation}
be a family of probability distributions on $\mathcal{P}([-2^{\Kmax}: 2^{\Kmax}])$.
In the dynamic HMC scheme below, an orbit selection kernel defines the probabilities $\rmpp_h(\msj \mid q_0, p_0)$ of considering samples from the orbit $\orbit_\msj(q_0, p_0) = \{ \Phiverlet[h][j](q_0, p_0): j \in \msj\}$, where the size of the index set $\msj \subset [-2^{\Kmax}: 2^{\Kmax}] \subset \Zset$ is bounded above by a constant.\footnote{
    While the restriction to a constant upper bound on the lengths of orbits is somewhat artificial from the theoretical point of view, and technically excludes algorithms such as HMC with randomised integration time where the number of steps is taken to have an unbounded distribution (exponential distribution is a common choice) from the class of dynamic HMC algorithms as defined here, all practical algorithms have some effective limitation on the number of leapfrog steps taken during a single iteration.
    In particular, the NUTS algorithm explicitly incorporates the bound $2^{\Kmax}$ for the number of leapfrog steps so we take the convenient opportunity here to introduce the notation.

    Some of our results include an assumption that bounds the total allowed integration time $h 2^{\Kmax}$.
    In the cases where such a bound is assumed, we usually write $\Kmax = \Kmax(h)$ to emphasize the interdependence of the admissible values of $h$ and $\Kmax$.
}
Let an \emph{index selection kernel}
\begin{equation}
    \rmqq_h = \{ \rmqq_h(\cdot \mid \msj, q_0, p_0) : \msj \subset [-2^{\Kmax}: 2^{\Kmax}], (q_0, p_0) \in (\Rset^d)^2 \}
\end{equation}
be a family of probability distributions on the index sets $\msj \subset [-2^{\Kmax}: 2^{\Kmax}]$, indexed by the orbit index sets $\msj \subset [-2^{\Kmax}: 2^{\Kmax}]$ selected by $\rmpp_h$ and the associated initial points $(q_0, p_0) \in (\Rset^d)^2$ in the phase space.
An index selection kernel defines the probability $\rmqq_h(j \mid \msj, q_0, p_0)$ of choosing the leapfrog iterate $\Phiverlet[h][j](q_0, p_0) \in \orbit_\msj(q_0, p_0)$ as the next state of the Markov chain when the current state is $(q_0, p_0)$ and the orbit has been selected.

\begin{definition}
    \label{def:dynamic-hmc}
    We define the dynamic HMC scheme associated to an orbit selection kernel $\rmpp_h$ and index selection kernel $\rmqq_h$ as the Markov chain $(Q_k)_{k \in \nset}$ defined by the following steps that define $Q_{k+1}$ given $Q_k$:
    \begin{enumerate}[wide, labelwidth=!, labelindent=0pt,label=(\arabic*)]
        \item Sample $P_{k+1}$ with distribution $\gauss(0, \Idd)$.
        \item Sample $\msi_{k+1}$ with distribution $\rmpp_h(\cdot \mid Q_k, P_{k+1})$.
        \item Sample $J_{k+1}$ with distribution $\rmqq_h(\cdot \mid \msi_{k+1}, Q_k, P_{k+1})$.
        \item Set $Q_{k+1} = \operatorname{proj}_1 \{\Phiverlet[h][J_{k+1}](Q_k,P_{k+1})\}$, where $\operatorname{proj}_1: (\Rset^d)^2 \to \mathbb{R}^d$ is the projection onto the first $d$ coordinates, i.e., from the phase space to the position coordinates.
    \end{enumerate}
\end{definition}
To specialize the general scheme to a specific algorithm, the orbit selection kernel $\rmpp_h$ and index selection kernel $\rmqq_h$ should be chosen so that the procedure leaves the desired target distribution $\pi$ invariant.
For example, the basic HMC algoritm with number of steps $T \in \nset$ may be realized as the dynamic HMC scheme with the orbit selection defined deterministically by
\begin{equation}
    \label{eq:hmc-orbit-selection}
    \rmpp_h^{\mathrm{HMC}}(\{0, T\} \mid q_0, p_0) = 1 \eqsp
\end{equation}
and the index selection probabilities via the Metropolis acceptance rate as
\begin{multline}
    \label{eq:hmc-index-selection}
    \rmqq_h^{\mathrm{HMC}}(\cdot \mid \{0,T\},q_0,p_0)
    =
    \left( 1 \wedge \frac{\tpi(\Phiverlet[h][T](q_0,p_0))}{\tpi(q_0,p_0)} \right) \updelta_T(\cdot)
\\    + \left(1-1\wedge \frac{\tpi(\Phiverlet[h][T](q_0,p_0))}{\tpi(q_0,p_0)}\right) \updelta_0(\cdot)
    \eqsp .
\end{multline}
More generally, \Cref{prop:trajectory-invariance} below gives a simple sufficient condition for the dynamic HMC algorithm associated to a particular choice of $\rmpp_h$ and $\rmqq_h$ to be invariant with respect to $\pi$.

Before stating the invariance result we make a few further comments on the general dynamic HMC scheme.
The transition kernel of the dynamic HMC algorithm associated to $\rmpp_h$ and $\rmqq_h$ has the form
\begin{gather}
    \Kker_{h}(q_0, \msa)
    =
    \int \dd p_0 \, \rho_{0}(p_0) \tilde{\Kker}_h((q_0, p_0), \msa)\eqsp,
    \qquad \text{where}
    \label{eq:transition-kernel}
    \\
    \tilde{\Kker}_h((q_0, p_0), \msa) = \sum_{\msj \subset \zset} \sum_{j \in \msj}  \rmpp_h(\msj \mid q_0, p_0) \rmqq_h(j \mid \msj, q_0, p_0) \updelta_{\operatorname{proj}_1(\Phiverlet[h][j](q_0, p_0))}(\msa) \eqsp
\end{gather}
for $q_0 \in \Rset^d$ and $\msa \in \mcbb(\Rset^d)$, and where we use the natural convention that $\rmpp_h(\msj \mid q_0, p_0) = 0$ for $\msj \not \subset [-2^{\Kmax}:2^{\Kmax}\}$ and where $\rho_{0}(\cdot)$ denotes the density of $\gauss(0, \Idd)$ on $\Rset^d$.
We refer to as $\tilde{\Kker}_h$ as an extended deterministic dynamic HMC kernel. 
Note that because of the dependence of $\rmqq_h$ and $\rmpp_h$ on the momentum $p_0$, the kernel $\Kker_h$ cannot generally be expressed as a position-dependent mixture of basic HMC kernels, i.e., in the form
$
    \Kker_h'(q_0,\msa)
    =
    \sum_{j\in \mathbb{Z}} \varpi_j(q_0) \int \dd p_0 \ \rho_{0}(p_0)
    \updelta_{\operatorname{proj}_1(\Phiverlet[h][j](q_0, p_0))}(\msa)
$
where $\{\varpi_j(q_0)\}_{j\in\zset}$ is a sequence of non-negative weights which sum to $1$ for all $q_0$.
Thus we emphasize that the dynamic HMC scheme presented here is significantly more general than mixtures of HMC kernels with different numbers of leapfrog steps.
In particular the scheme is general enough to encompass e.g., NUTS \cite{hoffman2014no}, the Apogee-to-Apogee Path Sampler \cite{sherlock2023apogee} and other algorithms where the orbit selection is defined dynamically via a stopping time.

The following proposition enables us to choose appropriate distributions $\rmpp_h$ and $\rmqq_h$ so that $\pi$ is invariant for $\Kker_h$.
The proposition and its proof unifies and generalizes existing invariance proofs of HMC algorithms in the literature, particularly that in \cite[Appendix A]{Betancourt}.
\begin{proposition}\label{prop:trajectory-invariance}
    If the orbit selection kernel $\rmpp_h$ and index selection kernel $\rmqq_h$ satisfy
    \begin{equation}
        \begin{aligned}
            &\tpi(q_0, p_0) \rmpp_h(\msj \mid q_0, p_0)
            \\ & \qquad =
            \sum_{j \in \mathbb{Z}} \mathbbm{1}_{\msj}(0) \tpi(\Phiverlet[h][-j](q_0, p_0)) \rmpp_h(\msj+j \mid \Phiverlet[h][-j](q_0, p_0)) \rmqq_h(j \mid \msj+j, \Phiverlet[h][-j](q_0, p_0)) \eqsp,
        \end{aligned}
        \label{eq:trajectory-invariance}
    \end{equation}
    for all $q_0,p_0 \in \Rset^d$ and $\msj \subset \Zset$, the transition kernel $\Kker_h$ defined in \eqref{eq:transition-kernel} leaves the target measure $\pi$ invariant.
\end{proposition}
\begin{proof}
    The proof is a straightforward computation, presented in Section 1.1 of the Supplementary Material A.
\end{proof}

The orbit selection probabilities $\rmpp_h(\msj \mid q_0, p_0)$ and $\rmpp_h(\msj+j \mid \Phiverlet[h][-j](q_0, p_0))$ in \eqref{eq:trajectory-invariance} refer to the same orbit in phase space, as
\begin{align}
    \label{eq:6}
    \orbit_{\msj+j}(\Phiverlet[h][-j](q_0, p_0))
    & =
    \{ \Phiverlet[h][\ell](\Phiverlet[h][-j](q_0, p_0)) \mid \ell \in \msj+j \}
    \\
    & =
    \{ \Phiverlet[h][\ell](q_0, p_0) \mid \ell \in \msj \}
    =
    \orbit_\msj(q_0, p_0)\eqsp,
\end{align}
and the index selection probability $\rmqq_h(j \mid \msj+j, \Phiverlet[h][-j](q_0, p_0))$ in \eqref{eq:trajectory-invariance} is the probability of choosing $\Phiverlet[h][j](\Phiverlet[h][-j](q_0, p_0)) = (q_0, p_0)$ from the orbit started at $\Phiverlet[h][-j](q_0, p_0)$.
Thus the meaning of the condition \eqref{eq:trajectory-invariance} is that for the invariance of the dynamic HMC kernel to hold it is sufficient that \emph{on every fixed phase space trajectory separately}\footnote{
    To be precise, by phase space trajectory we mean an orbit $\orbit_\msj(q_0, p_0)$ together with the indexing information.
    The concept could be formally defined as an equivalence class of pairs of index sets and initial points $(\msj, (q_0, p_0))$ with pairs considered equivalent, $(\msj', (q', p')) \equiv (\msj, (q, p))$, if $(q', p') = \Phiverlet[h][-j](q, p)$ and $\msj' = \msj + j$ for some $j \in \msj$.
    We will, however, not explicitly work with this formulation.
}
the index selection kernel $\rmqq_h$ leaves invariant the finitely supported measure on $\msj$ defined by the induced weights $\tpi(\Phiverlet[h][-j](q_0, p_0)) \rmpp_h(\msj+j \mid \Phiverlet[h][-j](q_0, p_0))$.
In particular, for a given $\rmpp_h$, the choice of $\rmqq_h$ to guarantee the invariance of $\pi$ for $\Kker_h$ according to \Cref{prop:trajectory-invariance} reduces to a problem of designing invariant Markov kernels on \emph{finite} state spaces.

As far as we are aware, in all dynamic HMC algorithms currently in use the orbit selection kernel $\rmpp_h$ is symmetric in the sense of the following corollary, which gives the invariance condition \eqref{eq:trajectory-invariance} a particularly simple form.

\begin{corollary}
    \label{cor:trajectory-invariance-symmetric}
Suppose the orbit selection kernel $\rmpp_h$ satisfies the symmetry condition
    \begin{align}
        \label{eq:trajectory-condition2}
        \rmpp_h(\msj + j \mid \Phiverlet[h][-j](q_0, p_0)) = \rmpp_h(\msj \mid q_0, p_0) \eqsp
    \end{align}
    for all $(q_0, p_0) \in (\Rset^d)^2$, $\msj \subset \Zset$ and $-j \in \msj$. Then, the invariance condition \eqref{eq:trajectory-invariance} is equivalent to
    \begin{equation}
        \tpi(q_0, p_0)
        =
        \sum_{j \in \mathbb{Z}} \mathbbm{1}_{\msj}(0) \tpi(\Phiverlet[h][-j](q_0, p_0)) \rmqq_h(j \mid \msj+j, \Phiverlet[h][-j](q_0, p_0)) \eqsp .
    \end{equation}
\end{corollary}

\begin{proof}
    Immediate from \eqref{eq:trajectory-invariance} and \eqref{eq:trajectory-condition2}.
\end{proof}


In \Cref{sec:nuts-its-invariance} we present the NUTS algorithm as a particular instance of the dynamic HMC scheme and show that it defines a Markov kernel which admits $\tpi$ as invariant probability measure using \Cref{prop:trajectory-invariance}.

\subsection{Ergodicity of dynamic HMC}
\label{sec:ergodicity-of-dynamic-hmc}
The main focus of this work is to study the ergodic properties of dynamic HMC algorithms, in particular the NUTS algorithm of Hoffman and Gelman \cite{hoffman2014no} and its more recent developments \cite{Betancourt}.
As background we review some existing results in the literature, though a comprehensive survey is out of the scope of this work.
We focus here on results that contextualize our main results for NUTS, Theorems \ref{thm:ergodicity_compile} and \ref{thm:ergo_geo}, in terms of earlier contributions towards similar theoretical guarantees for HMC-type algorithms and also the fundamental limits of such algorithms.

We say that a Markov kernel $\Kker$ on $\Rset^d$ with the invariant
measure $\pi$ is $\pi$-ergodic if for  $\pi$-almost every $x\in\rset^d$,
\begin{equation}
  \label{eq:pi_ergodic}
  \lim_{k \to \plusinfty} \|\Kker^k(x, \cdot) - \pi\|_{\mathrm{TV}} = 0 \eqsp.
\end{equation}
Note that this property implies a strong law of large numbers for $\pi$-integrable functions.
In addition, recall that, 
for a measurable function $\VFL: \Rset^d \to [1, \infty)$, a Markov kernel $\Kker$ on $\Rset^d$ with the invariant measure $\pi$ is said to be $\VFL$-uniformly ergodic if there exist  $C > 0$ and $\gamma \in (0, 1)$ for which for any $x \in \Rset^d$ and $k \in \Nset$,
\begin{equation}
    \label{eq:v-uniform-ergodicity}
    \|\Kker^k(x, \cdot) - \pi\|_{\VFL} \leq C \gamma^k \VFL(x) \eqsp.
  \end{equation}

We first discuss the fundamental limitation on the range of potentials for which ergodicity and $\VFL$-uniform ergodicity of HMC algorithms may hold.
It is apparent from the definition \eqref{eq:iteration_verlet} of the leapfrog integrator that if $|\nabla U(q_0)|$ is large compared to $|q_0|$ and $h |p_0|$, the discretized dynamics is unable to accurately track the continuous dynamics described in \eqref{eq:hamiltonian_system}.
The loss of stability of the leapfrog integrator is well-known in the case of rapidly growing potentials, particularly in the tails, as discussed in more detail in \cite{Geo_integratorsbou2018geometric, Livingstone2017-ub}.
For the usual quadratic kinetic energy as in \eqref{eq:hamiltonian_system}, the limit for stability is at quadratic potentials.
However, since the definitions \eqref{eq:pi_ergodic} and \eqref{eq:v-uniform-ergodicity} of $\pi$-ergodicity and $\VFL$-uniform ergodicity guarantee convergence from ($\pi$-almost) any starting point $x \in \Rset^d$, the instability of the leapfrog integrator in the tails indicates that geometric ergodicity may not be expected to hold for potentials exhibiting growth faster than quadratic.
We emphasize that this limitation concerning the tails of the target distribution reflects the well-known instability of the leapfrog integrator, which has practical implications for tuning HMC algorithms.


The Metropolis-adjusted Langevin algorithm (MALA) can be seen as a variant of the basic HMC algorithm with a single leapfrog step. Establishing ergodicity for MALA is relatively straightforward under mild regularity conditions on the potential function $U$, as discussed in \cite{roberts1996exponential}. Specifically, the MALA kernel $\Kker_h^{\mathrm{MALA}} = \Kker_{h, 1}^{\mathrm{HMC}}$ corresponds to a Metropolis-Hastings kernel with a position-dependent Gaussian proposal. In this case, it is evident that any open set can be reached from any point with a positive probability in a single iteration.
However, achieving $\VFL$-uniform ergodicity requires additional conditions on $U$,  especially with regards to its growth at infinity. Specifically, it requires that the potential exhibits at most quadratic growth as explained in more detail in \cite{roberts1996exponential} and other references. The proof strategies for establishing $\VFL$-uniform ergodicity are relatively more involved \cite{roberts1996exponential,durmus2022geometric}.

Livingstone et al. \cite{livingstone2019geometric} investigate HMC kernels of the form
\begin{equation}
    \label{eq:rhmc-transition-kernel}
    \Kker_h^{\mathrm{rHMC}} = \sum_{j=1}^J \varpi_j \Kker_{h,j}^{\mathrm{HMC}},
\end{equation}
In this form, $\Kker_h^{\mathrm{rHMC}}$ is a sum of HMC kernels, each with a different number of leapfrog steps. The weights $(\varpi_j)_j$ are nonnegative and sum up to 1, with $\varpi_1$ being positive. Additionally, a upper bound $J \in \Nset$ is specified for the non-zero weights.
Therefore, irreducibility and as result ergodicity of HMC established in \cite[Section 5.1]{livingstone2019geometric} are direct consequences of the ones of MALA since $\varpi_1 >0$. 
Durmus et al. \cite{Durmus2017-tf} relax the restriction $\varpi_1 > 0$ for achieving ergodicity of HMC kernels in the form of equation \eqref{eq:rhmc-transition-kernel}. Notably, their results cover cases with a deterministic number of leapfrog steps $T \geq 2$ and no restriction on the step size $h > 0$.
Our main results, as presented in \Cref{thm:ergodicity_compile} and \Cref{thm:ergo_geo}, establish both ergodicity and $\VFL$-uniform ergodicity for the NUTS algorithm described in detail in \Cref{sec:nuts-its-invariance}. These results have broad applicability for potentials that exhibit growth slower than quadratic at infinity, with additional assumptions in the case of quadratic growth.

A major challenge in establishing our results arises from the fact that while some NUTS variants include a MALA component in their transition kernel the specific variant we deal with, which is used in recent versions of Stan, does not.
As a result, conventional irreducibility arguments based on accessibility in one step, as seen in \cite{roberts1996exponential, livingstone2019geometric}, are insufficient.
Further, existing results on basic HMC establish strong control of leapfrog orbits only for short integration times $Th > 0$ (see e.g. \cite{Chen2023hmc, Bou-Rabee2023, Gouraud2022} and our \Cref{thm:trajectory-transitivity}) and on the other hand the use of degree theory to control the orbits for longer integration times as in \cite{Durmus2017-tf} seems difficult to adapt to the analysis of the dynamically defined stopping time.
Thus we need to employ a new proof strategy incorporating global information about leapfrog orbits.

\section{NUTS and its invariance}
\label{sec:nuts-its-invariance}

The original NUTS (No-U-Turn Sampler) algorithm developed by Hoffman and Gelman \cite{hoffman2014no} has undergone further development, and the current variant implemented in recent versions of Stan \cite{Stan_libcarpenter2017stan} and other probabilistic programming frameworks (such as PyMC3 \cite{Salvatier2016} and Turing.jl \cite{ge2018t}) has some differences from the original algorithm.

In this section, we provide a precise description of the algorithm that we analyze along with a comprehensive and detailed proof of its invariance. We have made efforts to align our algorithm with the current version of Stan (2.32), but since a complete description of Stan's algorithm is not readily available outside of the program code, there may be some differences, and certain minor differences are intentional.
In summary, our algorithm implements the original NUTS stopping rule for orbit selection \cite{hoffman2014no} while excluding the energy check (see Section 8 of Supplementary Material A), and what is called biased progressive sampling in \cite{Betancourt} for index selection.

We will now proceed with the detailed presentation of our algorithm, outlining its key components and steps.

\subsection{The NUTS algorithm}
Implementation of one iteration of the NUTS algorithm is given as pseudo-code in \Cref{alg:nuts-doubling}.
Given an initial position and momentum $q_0,p_0 \in \rset^d$ and sequences $(V_k)_{k =0}^{\Kmax}$, $(\bar{U}_{k})_{k=0}^{\Kmax}$ and $(\tilde{U}_{k})_{k=0}^{\Kmax}$  of \iid~random variables with distribution
$\mathrm{Ber}(1/2)$, $\mathrm{Unif}([0,1])$ and $\mathrm{Unif}([0,1])$, \Cref{alg:nuts-doubling} gives as output  
$(q_{j_f},I_f,j_f)$ satisfying $q_{j_f} = \operatorname{proj}_1 \{\Phiverlet[h][j_f](q_0,p_0)\}$ by definition and
$I_f,j_f$ are measurable transformations of $(V_k,\bar{U}_k,\tilde{U}_k)_{k =0}^{\Kmax}$ and $(q_0,p_0)$ which satisfy $j_f \in I_f$. 
As a result, \Cref{alg:nuts-doubling} defines an extended dynamic HMC kernel and its related dynamic HMC kernel in the sense of \eqref{eq:transition-kernel} that we denote by $\KkerUE_h$ and $\KkerU_h$, respectively. 
The orbit selection kernel associated to $\KkerU_h$, i.e., the distribution of $I_f$ given $(q_0, p_0)$, will be denoted by $\rmp_h$ and the index selection kernel associated to $\KkerU_h$, i.e., the distribution of $j_f$ given $I_f$ and $(q_0, p_0)$, will be denoted by $\rmq_h$.

We briefly describe the construction of the random variables $I_f, j_f$.
The random interval $I_f \subset \zset$ is defined recursively through the sequence of random intervals $(I_k)_{k=0}^{\Kmax}$ starting from $I_0 = \{0\}$.
The interval $I_{k+1}$ is the union of $I_k$ and an interval $I^{\text{new}}_k$ with $|I^k| = |I^{\text{new}}_k|$, which is to the left of $I_k$ if $V_k=0$ and to the right of $I_k$ otherwise; see \Cref{I_construction}.
Given the intervals $(I_k)_{k=0}^{\Kmax}$, the final interval is defined as $I_f = I_{K_f}$ where $K_f+1$ is a stopping time that indicates that an U-turn has occurred in the trajectory associated to $I_k$ as described in the sub-routine \Cref{alg:nuts-uturn}.
Note that since the intervals $(I_k)_{k=0}^{\Kmax}$ have lengths $2^k$ they may be naturally indexed by an increasing sequence of complete binary trees of depths $k$.
Once $I_f$ has been constructed the index $j_f \in I_f$ is chosen so that the transition $0 \to j_f$ leaves the induced target measure on $I_f$ (i.e., the measure on $I_f$ with weights $\tpi(\Phiverlet[h][j](q_0, p_0))$ for $j \in I_f$) invariant and so that $j_f$ is as far as possible from $0$ on $I_f$ in terms of the binary tree induced by the construction; see \Cref{scheme_mul_prob}.
We refer to \cite{Betancourt} for more intuition and an alternative presentation of this construction.
In addition, more detail on these random variables is given in \Cref{sectionrmp} and \Cref{sectionrmq} where explicit expressions for $\rmp_h$ and $\rmq_h$ are given and where we define the notation required for our theoretical analysis.

Our pseudo-code implementation for the simulation of $\KkerUE_h((q_0,p_0),\cdot)$ in \Cref{alg:nuts-doubling} is different from the implementation used in Stan, which relies instead on a recursive construction of the full binary trees which appear in the indexing of the intervals $(I_k)_{k=0}^{\Kmax}$.

\begin{minipage}{.60\linewidth}
\begin{algorithm}[H]
    \caption{
        Doubling dynamic orbits with the No U-Turn stopping rule.
    }
    \label{alg:nuts-doubling}
    \begin{algorithmic}
        \Statex \textbf{Input}
                   \Statex  initial position and momentum $(q_0, p_0)\in (\Rset^d)^2$
        \Statex            leapfrog parameters $h>0$
        \Statex            sequences of \iid~random variables 
        \Statex            $(V_k)_{k =0}^{\Kmax}$, $(\bar{U}_{k})_{k=0}^{\Kmax}$ and $(\tilde{U}_{k})_{k=0}^{\Kmax}$ with distribution
\Statex            $\mathrm{Ber}(1/2)$, $\mathrm{Unif}([0,1])$ and $\mathrm{Unif}([0,1])$  respectively
\Statex \textbf{Initialize} $I_0 \gets \{ 0 \}$, $I_0' \gets \{ 0 \}$,  $k \gets 0$, $j_0 \gets 0$,
NoUTurns$\gets$ True 
        \While {NoUTurns and $k < \Kmax$}
            \If {$V_k > 0$}
            \Statex $I^{\text{new}}_k=\bigcup_{l=1}^{2^k} \{ \max I'_{l} + l \} $
            \Statex $I_{k+1}' \gets I_k' \cup I^{\text{new}}_k$
                \For {$\ell = \max I'_{k}  + 1 \, : \, \max I'_{k} + 2^k$}
                    \Statex $q_\ell, p_\ell \gets \Phiverlet[h][1](q_{\ell-1}, p_{\ell-1})$
                \EndFor
            \Else
            \Statex $I^{\text{new}}_k \gets \bigcup_{l=1}^{2^k} \{ \min I'_k - l \}$
            \Statex $I'_{k+1}  \gets I'_k \cup I^{\text{new}}_k$
                \For {$\ell = \min I'_k - 1 \, : \, \min I'_k - 2^k$}
                    \Statex $q_\ell, p_\ell \gets \Phiverlet[h][-1](q_{\ell+1}, p_{\ell+1})$
                \EndFor
            \EndIf
            \Statex $j_{k+1}\gets j_k $
            \Statex NoUTurns $\gets \operatorname{NoUTurns}(I_{k+1} ', \mathcal{O}_{I_{k+1}'}(q_0,p_0))$ 
            \If {NoUTurns}
            \Statex $I_{k+1} \gets I_{k+1} '$
            \Statex $(j_{k+1},i_{k+1}') \gets$ Update($j_k,I^{\text{new}}_k,I_k,q_0,p_0,\bar{U}_k,\tilde{U}_k$)
                \EndIf
            \Statex $k \gets k+1$
        \EndWhile
        \Statex $I_f \gets I_k $,  $j_f \gets j_k $
        \Statex \textbf{return} $(q_{j_f},I_f,j_f)$
    \end{algorithmic}
\end{algorithm}
\end{minipage}
\begin{minipage}{.4\linewidth}
\begin{algorithm}[H]
    \caption{
        U-turn checking.
    }
    \label{alg:nuts-uturn}
    \begin{algorithmic}
        \Statex \textbf{Subroutine} $\operatorname{NoUTurns}(\msj, \orbit_\msj(q_0,p_0))$
        \Statex \textbf{Input}
            set $\msj \subset \mathbb{Z}$ of $2^k$ consecutive integers for some $k > 0$,
            orbit $\orbit_\msj(q_0,p_0)$ indexed by $\msj$
                  \Statex \textbf{Initialize} $K \gets \log_2 |\msj|$, $j_0 \gets \min \msj$, $\mathrm{UTurnFound} \gets \mathrm{False}$  
        \For {$k = 1, \dots, K-1$ \footnote{Note that the first step is always accepted \phantom{aaa} in the index set $I$.}}
        \For {$h = 1, \dots, 2^{K-k}$}
\State         $\ell \gets j_0 +(h-1)2^k$
\If {$p_{\ell + 2^k-1}^\top (q_{\ell + 2^k-1} - q_{\ell}) < 0$ \\
\quad \qquad   or $p_{\ell}^\top (q_{j_0 + h 2^k - 1} - q_{\ell}) < 0$}
                    \State $\mathrm{UTurnFound} \gets \mathrm{True}$
                \EndIf
            \EndFor
        \EndFor
        \If {$\mathrm{UTurnFound}$} \textbf{return} $\mathrm{False}$
        \Else \,\textbf{return} $\mathrm{True}$
        \EndIf
    \end{algorithmic}
  \end{algorithm}
\end{minipage}
\vspace{4mm}

We stress that the distinction is solely computational in that the recursive implementation is considerably more memory-efficient; there is no difference in terms of the Markov transitions defined by the different implementations.
We consider the details of the fully recursive implementation in Section 7 of the Supplementary Material A.
\begin{algorithm}[t]
    \caption{
        Elementary step of the recursive sampling on the index set
    }
    \label{alg:nuts-samplingj}
    \begin{algorithmic}
        \Statex \textbf{Subroutine} Update($j_k, I^{\text{new}}_k, I_k,q_0,p_0,\bar{U}_k,\tilde{U}_k$)
        \Statex \textbf{Input}
            $j_k$ the current index,
            index sets $I^{\text{new}}_k,I_k$,
           initial states $q_0,p_0$,             random variables $\bar{U}_k,\tilde{U}_k$
        \State $\bar{\pi}_{i} \gets \left. \tpi(\Phiverlet[h][i](q_0,p_0))\middle. / \sum_{m \in I^{\text{new}}_k} \tpi(\Phiverlet[h][m](q_0,p_0))\right. $, for $i \in I_k\cup I_k^{\text{new}}$
        \State $i'_{k+1} \sim \operatorname{Multinomial}((\bar{\pi}_i)_{i \in I^{\text{new}}_k},I^{\text{new}}_k,\tilde{U}_k)$
         \State $\bar{V}_k =  \1\defEns{\bar{U}_k \leq [1 \wedge   \sum_{m \in I^{\text{new}}_k} \bar{\pi}_m / \sum_{m \in I_k} \bar{\pi}_m]}$
        \State $j_{k+1} \gets j_k $
        \If {$\bar{V}_k=1$}
        		\State $j_{k+1} \gets i'_{k+1}$
	\EndIf
	\State \textbf{return} $(j_{k+1},i_{k+1}')$
    \end{algorithmic}
  \end{algorithm}

  \begin{algorithm}[t]
      \begin{algorithmic}
        \Statex \textbf{Subroutine} Multinomial($(\bar{\pi}_i)_{i \in I^{\text{new}}_k},I^{\text{new}}_k,\tilde{U}_k$)
        \Statex \textbf{Input}
            $(\bar{\pi}_i)_{i\in I_k^{\text{new}}}$ the weights, the
            index set $I^{\text{new}}_k$,
            random variable $\tilde{U}_k$
        \State $i_{k+1}' \gets \sum_{i\in I_k^{\text{new}}} i \mathbbm{1}_{[\sum_{j\in I_k^{\text{new}}: j\leq i-1} \bar{\pi}_j:\sum_{j\in I_k^{\text{new}}: j\leq i} \bar{\pi}_j] }(\tilde{U}_k)$\footnote{$\sum_{i\in \emptyset}i=0$}
	\State \textbf{return} $i_{k+1}'$
    \end{algorithmic}
\end{algorithm}

\begin{figure}[!h]
    \begin{center}
    \includegraphics[width=100mm]{images/I_trajectory_correct.png}
    \end{center}
    \caption{Scheme of the construction of the index set $I_f$ in  \Cref{alg:nuts-doubling} based on \cite[Figure 1]{hoffman2014no}}
    \label{I_construction}
  \end{figure}

In the rest of this section we give explicit expressions for the orbit and index selection kernels $\rmp_h$ and $\rmq_h$ and use \Cref{prop:trajectory-invariance} to prove, in full detail, that $\pi$ is invariant for the NUTS transition kernel $\KkerU_h$.
An entirely different, but also fully mathematically detailed, proof has been given in \cite{andrieu2020general} for the original slice variant of NUTS.
We rely on the more general \Cref{prop:trajectory-invariance} and show that the orbit and index selection kernels $\rmp_h$ and $\rmq_h$ satisfy the assumption and conclusion of \Cref{cor:trajectory-invariance-symmetric}; the details consist in the more or less straightforward but laborious matter of unpacking the complexity inherent in \Cref{alg:nuts-doubling} and \Cref{alg:nuts-samplingj}.
An alternative, considerably less detailed, presentation of the ideas is given in \cite{Betancourt}.
\begin{theorem} 
    \label{thm:invariance_target}
    Assume \Cref{hyp:regularity}.
    The target distribution $\pi$ is invariant for $\KkerU_h$.
\end{theorem}
\begin{proof}
    The result is a consequence of \Cref{prop:trajectory-invariance} in \Cref{sec:general-dynamic-hmc}, \Cref{prop:rmpcondition_check} in \Cref{sectionrmp} and \Cref{lemma:invariance_q} in \Cref{sectionrmq}.
\end{proof}

\subsection{The orbit selection kernel $\rmp_h$}
\label{sectionrmp}

We let $(q_0,p_0) \in (\Rset^d)^2$ be fixed throughout this section and specify here the orbit selection kernel $\rmp_h$ for which $\rmp_h(.|q_0,p_0)$ is the distribution of $I_f$, the index set returned by \Cref{alg:nuts-doubling} starting from $q_0,p_0$.
Let $(V_k)_{k=0}^{\Kmax-1} \in \{0,1\}^{\Kmax}$ be the sequence of i.i.d random variables used in \Cref{alg:nuts-doubling}.
The sequence $(I_k)_{k=0}^{\Kmax}$ of subintervals of $\mathbb{Z}$ is constructed recursively via $(V_k)_{k=0}^{\Kmax-1}$ by setting $I_0 = \{0\}$, and
if $V_k=1$ the $2^k$ consecutive integers to the right of $I_k$ are added to define $I_{k+1}$,
and if $V_k=0$ the $2^k$ consecutive integers to the left of $I_k$ are added to define $I_{k+1}$.

We first aim for an explicit expression for $I_k$ in terms of the sequence $(V_k)_{k=0}^{\Kmax-1}$.
To this end, we introduce some additional notation.
We denote binary sequences $(v_k)_{k=0}^{K-1} \in \{0,1\}^{K}$, for $K\in \Nset^*$, by $v_{K-1} \ldots v_0$. 
We identify $\{0,1\}^K$ with $B_K=[0:2^K-1]$ via the bijection $v_{K-1} \dots v_0 \to \sum_{k=0}^{K-1} v_k 2^k$, 
i.e., any element $v\in B_K$ is identified with its unique length $K$ binary representation $v_{K-1}\ldots v_0\in \{0,1\}^K$. 
We define the concatenation of $a \in B_K$ and $b \in B_{K'}$ 
 by $ab = a_{K-1} \dots a_0 b_{K'-1} \dots b_0$ and denote the truncation of $a$ 
 to its $n$ last bits by $a|_n = a_{n-1} \dots a_0$.
 Finally, for $v\in B_K$ we denote
\begin{equation}
    B_K(v) = \{-T_-^{(K)}(v), \dots, T_+^{(K)}(v)\}=B_K-(2^K-1-v)\ee,
    \quad \text{where}
\end{equation}
\begin{equation}
\label{eq:T_-}
    T_-^{(K)}(v) = \sum_{k=0}^{K-1} (1-v_k) 2^k
    \quad \text{and} \quad
    T_+^{(K)}(v) = \sum_{k=0}^{K-1} v_k 2^k=v\ee.
\end{equation}
Equipped with these notations, by setting $V=\sum_{i=0}^{\Kmax-1} V_i 2^i$ we have $I_K=B_K(V|_K)$. 
The construction of $B_K(V|_K)$ is depicted in \Cref{I_construction}.

We define $K_f$ as the step at which the algorithm, starting from $q_0$ and $p_0$, stops due to a U-turn. Moreover, let $I_f(V)$ be the random index set  returned by the algorithm:
$I_f(V) = B_{K_f}(V|_{K_f})$.
 In the following, we show that $K_f=(\funS_f-1)\wedge \Kmax$ where $\funS_f$ is a stopping time for the filtration 
 generated by the i.i.d.~sequence of Bernoulli $(V_k)_{k=0}^{\Kmax-1}$.
 By establishing this relationship, we will be able to express $I_f$ as a function of $(V_k)_{k=0}^{\Kmax-1}$ and  specify $\rmp_h$.
  We now focus on the precise definition of  $\funS_f$ based on \Cref{alg:nuts-uturn}.
  We say that a U-turn occurs between indices $k$ and $k'$ belonging to $[-2^{\Kmax}+1:2^{\Kmax}-1]$ if, 
denoting $(q_k, p_k) = \Phiverlet[h][k](q_0,p_0)$, at least one of the two following inequalities holds:
\begin{equation}
  p_{k'}^\top (q_{k'}-q_k) < 0
    \quad \text{or} \quad
    p_k^\top (q_{k'}-q_k) < 0 \eqsp.
  \end{equation}
 We need to specify the set of pairs of indices in $I_k$ that are considered in \Cref{alg:nuts-uturn}.
Define  
\begin{gather}
    \scrU_{k,l,+}^{(K)} = \{ v\in B_{K}:\ee p_{i_{+}(K,k,l,v)}^\top (q_{i_{+}(K,k,l,v)} - q_{i_{-}(K,k,l,v)}) < 0 \} \eqsp,
    \\
    \scrU_{k,l,-}^{(K)} = \{v\in B_{K}: \ee p_{i_{-}(K,k,l,v)}^\top (q_{i_{+}(K,k,l,v)} - q_{i_{-}(K,k,l,v)}) < 0 \} \eqsp,
\end{gather}
where
$$i_{-}(K,k,l,v)=-T_-^{(K)}(v)+(l-1) 2^k,\qquad i_{+}(K,k,l,v)=-T_-^{(K)}(v) + l 2^k-1 $$
for $v\in B_{K}$, $K\in [\Kmax]$, $k\in [K-1]$ and $l\in [2^{K-k}]$, and further
\begin{equation}
    \label{eq:scrU}
    \scrU_k^{(K)}(q_0,p_0) = \bigcup_{l=1}^{2^{K-k}} \big( \scrU_{k,l,+}^{(K)} \cup \scrU_{k,l,-}^{(K)} \big)
    \quad
    \text{and}
    \quad
    \scrU^{(K)}(q_0,p_0) = \bigcup_{k=1}^{K-1} \scrU_k^{(K)}(q_0,p_0) \eqsp,
  \end{equation}
with the convention $\cup_{k=1}^{0} = \emptyset$.
  Then,  the event $\{V|_K \in \scrU^{(K)}\}$ corresponds to the event that a U-turn occurs  at the $K$-th stage of the algorithm  between two indices in $I_K$; 
more precisely between $-T_-^{(K)}(V|_K) + (l-1) 2^k$ and $-T_-^{(K)}(V|_K) + l 2^k - 1$ for some $k\in [K-1]$ and $l\in [2^{K-k}]$.
It is worth pointing out that, by construction,
the event $\{V|_K \in \scrU^{(K)}\}$ does not consider \emph{all} the pairs of points in $\orbit_{B_K(V|_K)}(q_0,p_0)$. 
For instance, \Cref{alg:nuts-uturn} does not verify if there is a U-turn between the pair of indices $1$ and $2$ or, more generally, between pairs of indices with different parity. 
The reason for this is primarily computational and allows for a significant reduction in the amount of memory used by the algorithm\footnote{
    We remark here that, in one of our minor intentional differences from Stan 2.32, the stopping rule implemented in Stan checks slightly more pairs of indices for U-turns.
    Namely, the sets $\scrU_{k,h,+}^{(K)}$ and $\scrU_{k,h,-}^{(K)}$ are augmented with additional checks given by
    \begin{equation}
        \scrU_{k,h,++}^{(K)} = \{ p_{-T_-^{(K)}(v|_K) + h 2^k}^\top (q_{-T_-^{(K)}(v|_K) + h 2^k} - q_{-T_-^{(K)}(v|_K) + (h-1) 2^k}) < 0 \} \eqsp
    \end{equation}
    and various symmetrizations in order to plug some gaps that are left in the checks in $\scrU_{k,h,+}^{(K)}$ and $\scrU_{k,h,-}^{(K)}$.
    The notation being heavy already, we do not write out these additional checks.
    While the additional checks can make a significant difference for the computational performance of the algorithm, our results and methods are independent of these details.
}.

We are ready to define $\funS_f$. 
Defining $\Sfun: \cup_{K=1}^{\Kmax} B_K \times (\Rset^d)^2\to [\Kmax]$ by
\begin{equation}
    \label{eq:nuts-stopping}
 \Sfun(v,q_0,p_0) = \inf\{k \in [K] \, : \, v|_k \in \scrU^{(k)}(q_0,p_0)\} \eqsp,  \quad v \in B_K \eqsp,
\end{equation}
with $\inf \emptyset = \infty$, we set $\Sfun_f=\Sfun(V,q_0,p_0) $.
Finally, $\rmp_h(\cdot \mid q_0,p_0)$ is defined as the distribution of the random variable $I_f(V)$:
\begin{equation}
    \label{eq:nuts-I}
    I_f(V) = B_{K_f}(V|_{K_f}) \eqsp, \text{ with } K_f=(\funS_f-1)\wedge \Kmax \eqsp.
  \end{equation}
  By construction, $\funS_f$ is a stopping time with respect to the filtration generated by the sequence $(V_k)_{k =0}^{\Kmax - 1}$ and for any $v \in B_{K'}$, $K' \geq K$, it holds that
  \begin{equation}
    \label{eq:prop_funS}
\text{    if $\funS(v,q_0,p_0) = K$} \eqsp, \quad \text{ $\funS(v,q_0,p_0) = \funS(v|_K,q_0,p_0)$ }\eqsp.
  \end{equation}
  With a slight abuse of notation we drop the dependence on $q_0,p_0$ in $\funS$ defined in \eqref{eq:nuts-stopping} and simply denote $\funS(v,q_0,p_0)$ by $\funS(v)$ as long as $q_0$, $p_0$ are considered fixed. 

The following result gives an expression for the orbit selection probabilities $\rmp_h$ that will be used in verifying that the symmetry condition \eqref{eq:trajectory-condition2} in \Cref{cor:trajectory-invariance-symmetric} is satisfied by NUTS.
\begin{lemma}\label{lemma:nuts-trajectory-probability}
    Assume \Cref{hyp:regularity}.
    For any $K\in [\Kmax]$ and $a \in B_K$,
    \begin{equation}
      \label{eq:4}
        \rmp_h(B_K(a) \mid q_0, p_0)
        =
        \begin{cases}
            \sum_{b \in \{0, 1\}} 2^{-K-1} \mathbbm{1}\{ \funS(b a) = K+1 \},
            & K < \Kmax
            \\
            2^{-\Kmax} \mathbbm{1}\{ \funS(a) > \Kmax \},
            & K = \Kmax \eqsp.
        \end{cases}
    \end{equation}
  If $\msj \subset \Zset$ is not of the form $\msj = B_K(a)$ for some $K \in [\Kmax]$ and $a \in B_K$, $\rmp_h(\msj |q_0,p_0) = 0$.
Finally, $(q_0',p_0') \in \mathbb{R}^d \to \rmp_h(B_K(a) \mid q_0', p_0')$ is measurable.
\end{lemma}

\begin{proof}
    For $K < \Kmax$ and $a \in B_K$ we have
    \begin{align}
& \textstyle        \rmp_h(B_K(a) \mid q_0, p_0)
         =
            \sum_{b \in B_{\Kmax}: b|_K=a|_K} \mathbb{P}(V_0=b_0, \dots, V_{\Kmax-1}=b_{\Kmax-1})
            \mathbbm{1}\{S(b) = K+1\}
        \\
        &  \textstyle  =
        \sum_{b \in B_{\Kmax}: b|_K=a|_K} 2^{-\Kmax} \mathbbm{1}\{\funS(b) = K+1\}
         =
        \sum_{b' \in B_{\Kmax-K}} 2^{-\Kmax} \mathbbm{1}\{\funS(b'a|_K) = K+1\}
        \\
        &  \textstyle =
        \sum_{b' \in B_{\Kmax-K-1}} \sum_{b'' \in \{0,1\}} 2^{-\Kmax} \mathbbm{1}\{\funS(b'b''a|_K) = K+1\}  \eqsp,
    \end{align}
 which completes the proof of \eqref{eq:4} using  \eqref{eq:prop_funS}.
    The case $K = \Kmax$ is clear and this completes the proof of the first part of the statement.

    Finally, the measurability of $(p,q)\in (\Rset^d)^2 \mapsto \rmp_h(B_K(a)\mid p,q)$ can be deduced from the measurability of $(q,p)\in (\Rset^d)^2  \mapsto S(a,(q,p))=\inf_{k\in [K]} k/[\mathbbm{1}_{\scrU^{(k)}(q,p)}(v)] $, 
    which in turn is implied by the measurability of 
    \begin{equation}
    (q,p)\in (\Rset^d)^2 \mapsto \mathbbm{1}_{\scrU^{(k)}(q,p)}(v)=\mathbbm{1}_{\scrU^{(k)}(v)}(q,p) 
  \end{equation}
    for $k\in [K]$ where $\scrU^{(k)}(v)=\{(q,p) \in (\Rset^d)^2: v\in \scrU^{(k)}(q,p) \}$ are open sets.
    Namely, they are pre-images of open sets under continous functions as the maps $(q,p)\in (\Rset^2)^2\mapsto \Phiverlet[h][j](q,p)$ for $j\in [-2^{K}+1:2^K-1]$ are continuous by \Cref{hyp:regularity}.
  \end{proof}

  We may then deduce 
  \begin{proposition}
    \label{prop:rmpcondition_check}
    Assume \Cref{hyp:regularity}.
    The orbit selection kernel $\rmp_h$ satisfies the condition \eqref{eq:trajectory-condition2}.
\end{proposition}
\begin{proof}
    Let $\msj \subset \mathbb{Z}$ with $\rmp_h(\msj \mid q_0, p_0) >0$ and let $-j\in \msj$. 
Then,  by \Cref{lemma:nuts-trajectory-probability}, $\msj = B_K(a)$ for $K \in[\Kmax]$ and $a \in B_K$.
Let $c = c_{K-1} \dots c_0 \in \{0, 1\}^K$ denote the unique binary sequence for which $B_K(c) = \msj+j$.
Then $\orbit_{\msj}(q_0,p_0)=\orbit_{\msj+j}(\Phiverlet[h][-j](q_0,p_0))=\orbit_{B_K(c)}(\Phiverlet[h][-j](q_0,p_0))$ by \eqref{eq:6} and as a result of the construction $S(b'c,\Phiverlet[h][-j](q_0,p_0))=S(b'a,q_0,p_0)$ for any $b'$ in $\{0,1\}$.
It follows from \Cref{lemma:nuts-trajectory-probability} and $\msj+j = B_K(c)$ that
\begin{align}
 \rmp_h(\msj+j\mid \Phiverlet[h][-j](q_0,p_0))&=2^{-K-1}\sum_{b'\in\{0,1\}} \mathbbm{1}\{S(b'c,\Phiverlet[h][-j](q_0,p_0)) = K+1\}\\
 &=2^{-K-1}\sum_{b'\in\{0,1\}} \mathbbm{1}\{S(b'a,q_0,p_0) = K+1\}=\rmp_h(\msj \mid q_0, p_0) \eqsp
\end{align}
for $K < \Kmax$.
 For $K=\Kmax$, applying  \Cref{lemma:nuts-trajectory-probability} again yields
 $$\rmp_h(\msj+j\mid \Phiverlet[h][-j](q_0,p_0))=2^{- \Kmax}=\rmp_h(\msj \mid q_0, p_0) \eqsp.$$

\end{proof}


\subsection{The index selection kernel $\rmq_h$}
\label{sectionrmq}
We still consider $(q_0, p_0) \in (\Rset^d)^2$ to be fixed, and in addition we consider an index set $\msi \subset \Zset$ satisfying $\rmp_h(\msi \mid q_0, p_0) > 0$ to also be fixed throughout this section.
By \Cref{cor:trajectory-invariance-symmetric} and \Cref{prop:rmpcondition_check},
in order to keep the target distribution $\pi$ invariant by $\KkerU_h$ it is sufficient to show that the
finitely supported distribution $\bar{\pi}(\cdot \mid \msi, q_0, p_0)$ defined as
\begin{equation}
    \label{eq:omega_bar}
    \bar{\pi}(j \mid \msi,q_0, p_0 )
    =
    \frac{\tpi(\Phiverlet[h][j](q_0, p_0)) }{\sum_{j'\in \msi} \tpi(\Phiverlet[h][j'](q_0, p_0))} \eqsp,
    \qquad
    j \in \msi \eqsp,
\end{equation}
is invariant for the transition kernel $\bar{\rmq}_h(\cdot, \cdot \mid \msi,q_0,p_0 )$ defined as
\begin{equation}
    \label{eq:transition_prob_q}
    \bar{\rmq}_h(j,a\mid \msi, q_0,p_0 ) = \rmq_h(a-j \mid \msi-j, \Phiverlet[h][j](q_0, p_0)) \eqsp,
    \qquad
    a, j \in \Zset \eqsp.
\end{equation}
A crucial step in showing this property is an explicit expression for the index selection kernel $\rmq_h$ defined by \Cref{alg:nuts-doubling}.

We remark that a simple way to ensure that $\bar{\pi}(\cdot \mid \msi, q_0, p_0)$ is invariant for $\bar{\rmq}_h(\cdot, \cdot \mid \msi, q_0, p_0)$ would be to replace the index selection kernel $\rmq_h$, which was defined as the distribution of $j_f$ from \Cref{alg:nuts-doubling}, with sampling independently from $\bar{\pi}(\cdot \mid \msi, q_0, p_0)$.
This, in fact, was how index selection was implemented in certain older versions of Stan.
However, this choice would not encourage the selection of distant states and thus can be expected to be less efficient, as discussed in \cite{Betancourt}.

As implemented by Algorithms \ref{alg:nuts-doubling} and \ref{alg:nuts-samplingj},
the index selection kernel $\rmq_h$ can be expressed recursively as follows.
Let $v=(v_k)_{k=0}^{\Kmax-1}\in B_{\Kmax}$ and denote, for any $K\in [\Kmax]$, $\msi_{v,K}^\text{old}=B_{K}(v|_{K})$ and $\msi_{v,K}^\text{new}=\msi_{v,K+1}^\text{old}\setminus \msi_{v,K}^\text{old}$.
For $K < S(v, q_0, p_0)$ and $j\in \msi_{v,K}^\text{old}$, the index selection kernel $\rmq_h$ satisfies
\begin{equation}
    \label{def:recur_def_rmq}
    \rmq_h(j \mid \msi_{v,K}^\text{old}, q_0, p_0)
    =
    (1-R_{v|_{K-1}}) \rmq_h(j \mid \msi_{v,K-1}^\text{old}, q_0, p_0)
    + R_{v|_{K-1}} \bar{\pi}(j \mid \msi_{v,K-1}^\text{new}, q_0, p_0),
\end{equation}
where $R_{v|_K} = 1 \wedge [\tpi(\msi_{v, K}^\text{new})/\tpi(\msi_{v, K}^\text{old})]$ with the shorthand notation
\begin{equation}
    \label{eq:def_tile_pi_J}
    \tpi(\msj) = \tpi_{q_0,p_0}(\msj)= \sum_{j\in \msj} \tpi(\Phiverlet[h][j](q_0,p_0)) \eqsp,
    \qquad
    \msj \subset \Zset \eqsp,
\end{equation}
and where $\rmq_h(j\mid\{0\},q_0,p_0)=\mathbbm{1}_0(j)$. 
This mechanism is depicted in  \Cref{scheme_mul_prob}.
It is apparent that this index selection favors the selection of states in an area far from the starting point with a high energy level.

When we expand the recursion starting from $B_K(v|K) = \msi_{v, K}^\text{old}$, we obtain the following formula:
\begin{multline}
    \rmq_h(j \mid B_K(v|_K), q_0, p_0)
    =
    \sum_{k=0}^{K-1} \left(\prod_{\ell=k+1}^{K-1} (1-R_{v|_\ell})\right) R_{v|_k} \bar{\pi}(j \mid \msi_{v, k}^\text{new}, q_0, p_0)
\\    + \mathbbm{1}_{0}(j) \prod_{\ell=0}^{K-1} (1-R_{v|_\ell})\eqsp.
\end{multline}
Note that since $\msi_{v,k}^\text{new}$ and $\msi_{v,k'}^\text{new}$ are disjoint for $k, k' \in [\Kmax]$, $k \neq k'$, exactly one of the terms in the expression above is nonzero for $j \in B_K(v|_K)$.
This is the desired explicit expression for $\rmq_h(\cdot \mid \msi, q_0, p_0)$, though the notation is quite cumbersome for analyzing its properties.

In order to obtain a more manageable form of the explicit expression for $\rmq_h(\cdot \mid \msi, q_0, p_0)$ we perform
the following reduction to the case where $\msi$ is replaced by $B_K$ for some $K\in [\Kmax]$.
For the rest of this section, we let $K\in [\Kmax]$ and $ v\in B_{\Kmax}$ be fixed such that $\msi=B_K(v|_K)=B_K-(2^K-1-v|_K)$. In addition, consider $\iota$, the unique increasing bijection from $\msi$ to $B_K$, i.e.,
\begin{equation}
    \label{eq:iota}
\iota(c)= c+(2^K-1-v|_K) \text{ for any $c\in B_K$ } \eqsp.
\end{equation}
\begin{lemma}
    \label{lemma:comeback}
    Assume \Cref{hyp:regularity}.
 For any $a,b\in B_K$ we have
\begin{equation}
    \bar{\rmq}_h(a,b\mid B_K,\Phiverlet[h][-(2^K-1-v|_K)](q_0, p_0))=\bar{\rmq}_h(\iota^{-1}(a),\iota^{-1}(b) \mid \msi,q_0,p_0) \eqsp.
  \end{equation}
\end{lemma}
\begin{proof}
    See Section 2.1 of the Supplementary Material A.
\end{proof}

For $a,b \in B_K$, denote by
\begin{equation}
  \label{eq:def_hatPi}
    \hat{\rmq}_{h,K}(a,b) = \bar{\rmq}_h(a,b\mid B_K,\Phiverlet[h][-(2^K-1-v|_K)](q_0, p_0)) \eqsp ,
\end{equation}
\begin{equation}
    \label{eq:hat_omega}
    \hat{\pi}_K(a) = \bar{\pi}(a\mid B_K,\Phiverlet[h][-(2^K-1-v|_K)](q_0, p_0))\eqsp.
\end{equation}
Then, \Cref{lemma:comeback} and definitions \eqref{eq:omega_bar}--\eqref{eq:transition_prob_q} imply that for any $a,b \in \msi$, 
\begin{equation}
  \label{eq:reduction_q_h_pi_h}
\text{  $\hat{\rmq}_{h,K}(\iota(a),\iota(b))= \bar{\rmq}_h(a,b\mid \msi,q_0,p_0)$ and $\hat{\pi}_K(\iota(a)) = \bar{\pi}(a\mid \msi,q_0,p_0)$} \eqsp.
\end{equation}
As a result and as already stated, we can restrict to the case $\msi=B_K$. 

\Cref{alg:nuts-doubling} defines recursively a binary tree as in \Cref{I_construction}.
Indeed, each Bernoulli random variable $V_k$ increases the depth of the tree by adding a new branch to the left ($V_k=0$) or to the right ($V_k=1$).
Then, to select the index $j_f \in I_f=\msi$, this binary tree is explored backward (see \Cref{scheme_mul_prob}), i.e., starting from the root corresponding to the last bit $V_{K-1}$.
Based on this observation we introduce the following notation.
For any $n \in [K]$,  $u \in B_n$ define
 \begin{equation}
   \label{eq:def_check_omega}
   \check{\pi}_n(u) = \tpi_{\Phiverlet[h][-(2^K-1-v|_K)](q_0, p_0)}(\{a \in B_K:\ee a|^n=u \}) \eqsp,
 \end{equation}
 where we recall that  $a|^n=a_{K-1}\ldots a_{K-n}$  is the truncation to the $n$ last bits.
The quantity $\check{\pi}_n(u)$ is the sum of the weights of states associated to indices $a$ such that $a|^n=u$.
\begin{figure}[!h]
    \begin{center}
    \includegraphics[width=110mm]{images/scheme_rmq_good_one.png}
    \end{center}
    \caption{Construction of probabilities $\rmq_h$ in the example of sampling $q_0,p_0$ and $\Phiverlet[h][3](q_0,p_0)$ with $\rmq(\cdot|\msi-3,\Phiverlet[h][3](q_0,p_0))=\bar{\rmq}_h(3,\cdot|\msi,q_0,p_0)$.}
    
    \label{scheme_mul_prob}
\end{figure}
Writing $1^c = 0$ and $0^c = 1$,  for any $t\in [K]$ and $a\in B_K$ we define
\begin{equation}
    \label{eq:Omega}
     \ee \Pi(a,t)=\prod_{i=0}^{t-1} \left( 1 - \left( 1 \wedge \frac{\check{\pi}_{i+1}(a|^ia_{K-i-1}^c)}{\check{\pi}_{i+1}(a|^i a_{K-i-1})} \right) \right)\ee,\ee \Pi(a,0)=1\eqsp,
\end{equation}
where we denote $a|^0b=b$ for any $a\in B_K$ and $b\in \{0,1\}$.
We are ready to state the explicit expression for $\hat{\rmq}_h$ and thus for $    \bar{\rmq}_h$ and $\rmq_h$.
\begin{lemma}
    Assume \Cref{hyp:regularity}.
    \label{lemma:definition_continuity_q}
    For $a, b \in B_K$, 
    \begin{equation}
        \label{eq:stan-index-selection2}
        \hat{\rmq}_{h,K}( a,b)
        =
        \begin{cases}
        \Pi(a,K) & \text{if}\ee n=K \\
        \Pi(a,n) \Big( 1 \wedge \frac{\check{\pi}_{n+1}(a|^n a_{K-n-1}^c)}{\check{\pi}_{n+1}(a|^na_{K-n-1})} \Big) 
        \frac{\hat{\pi}_K(b)}{\check{\pi}_{n+1}( a|^na_{K-n-1}^c)} & \text{otherwise}
        \end{cases},
    \end{equation}
    where $n=\max \big(\{i\in [K]:\ee a|^i= b|^i\}\cup \{0\}\big)$.
    In addition, it holds that
    $$\rmq_h(a\mid \msi,q_0,p_0)=\hat{\rmq}_{h,K}( \iota(0),\iota(a)) \eqsp,$$
    where $\iota$ is defined in \eqref{eq:iota}.
    Moreover, $\tilde{q}_0,\tilde{p}_0 \to \rmq_h(a\mid \msi,\tilde{q_0},\tilde{p_0}) $ is continuous.
\end{lemma}
\begin{proof}
    The explicit expression follows from \eqref{def:recur_def_rmq}; see Section 2.2 of the Supplementary Material A.
\end{proof}

We remark that if $a|^n\neq b|^n$ for some for $a,b\in B_K$ and $n\in [K]$, then
\begin{equation}
    \label{rm:growthmaximality}
    a|^l\neq b|^l \eqsp,
\end{equation}
for any $l\in[n:K]$, so that for any $l\in [0:n]$ we have $a|^l=b|^l $ when $n=\max [\{i\in [K]:\ee a|^i= b|^i\}\cup \{0\}]$.
 With the explicit expression of \Cref{lemma:definition_continuity_q}, the reversibility of $\hat{\rmq}_{h,K}$ follows easily.
\begin{proposition}
    \label{lemma:invariance_q}
    The transition kernel $\hat{\rmq}_{h,K}$ is reversible for $\hat{\pi}_K$ which implies that
 the transition kernel $\bar{\rmq}_h(\cdot, \cdot \mid \msi,q_0,p_0)$ leaves $\bar{\pi}(.\mid \msi,q_0,p_0)$ invariant.
\end{proposition}

\begin{proof}
    Let $a,b\in B_K$ and $n$ be defined as in \Cref{lemma:definition_continuity_q}.
    When $b\neq a$ (the case $b=a$ is trivial),
     we have
    \begin{align}
        \hat{\pi}_K(a) \hat{\rmq}_{h,K}( a,b)
        & =
       \Pi(a,n) \Big( 1 \wedge \frac{\check{\pi}_{n+1}(a|^n a_{K-n-1}^c)}{\check{\pi}_{n+1}(a|^na_{K-n-1})} \Big) 
        \frac{\hat{\pi}_K(b)\hat{\pi}_K(a) }{\check{\pi}_{n+1}( a|^na_{K-n-1}^c)}
        \\
        & =
        \Pi(a,n)\Big( \frac{1}{\check{\pi}_{n+1}( a|^na_{K-n-1}^c)} \wedge \frac{1}{\check{\pi}_{n+1}(a|^na_{K-n-1})} \Big) \hat{\pi}_K(b) \hat{\pi}_K(a)
        \\
        & =
        \hat{\pi}_K(b) \hat{\rmq}_{h,K}( b,a)
    \end{align}
    since $a|^n = b|^n$.
    From \eqref{eq:reduction_q_h_pi_h}-\eqref{eq:omega_bar}-\eqref{eq:transition_prob_q} and \Cref{cor:trajectory-invariance-symmetric}, the implications are clear.
\end{proof}


  In \Cref{ergodicity_section}, we are interested in the irreducibility of $\KkerU_h$.
 To this
hend, we rely on the irreducibility of $\bar{\rmq}_h(\cdot\mid \msi,q_0,p_0)$ which is a consequence of the following result.
    In words, while $\hat{\rmq}_{h,K}( a,b)$ can be equal to $0$ for most of elements $b \in \msi$ (i.e., when $b$ is in an old set with small weights; see \Cref{scheme_mul_prob}), 
  it can be shown that there exists $j_0\in \{1,2\}$ such that $\hat{\rmq}_{h,K}^{j_0}( a,b)>0$.
    \begin{proposition}
      \label{lemma:irreducibility_q}
  Assume \Cref{hyp:regularity}.
\begin{enumerate}[wide, labelwidth=!, labelindent=0pt, label=(\alph*)] 
\item \label{lemma:item1_irreducibility_q} For any $a,b$ in $B_K$, there exists $j_0$ in $\{1,2\}$ such that $\hat{\rmq}_{h,K}^{j_0}( a,b)>0$. In particular, it follows that $\bar{\rmq}_h(\cdot\mid \msi,q_0,p_0)$ is irreducible.
  
\item \label{lemma:item2_irreducibility_q}  In addition, suppose that for any $n\in[K]$ and $a,b\in B_K$ such that $a|^n\neq b|^n$, we have $\check{\pi}_n(a)\neq \check{\pi}_n(b) $. Then, for
  any $a,b\in B_K$ and for any $j\in \Nset$, we have $\hat{\rmq}_{h,K}^{2+j}(a,b)>0$. In particular, it follows that  $\bar{\rmq}_h(\cdot\mid \msi,q_0,p_0)$ is irreducible and aperiodic.
\end{enumerate}
\end{proposition}
\begin{proof}
    The proof is postponed to Section 3 of the Supplementary Material A.
\end{proof}



\section{Ergodicity}
\label{ergodicity_section}

The purpose of this section is to establish ergodicity of the NUTS sampler defined in the previous section.
To this end, we consider three explicitly verifiable conditions, each of which guarantees ergodicity: one concerns the step size $h > 0$ and maximum number of leapfrog steps $\Kmax$ (\Cref{hyp:lipschitz_hard}($h, \Kmax$)), one restricts to real analytic potentials $U$ and assumes that its Hessian vanishes at infinity (\Cref{hyp:analytic_potential}), and the last one assumes $\pi$ is a Gaussian distribution (\Cref{hyp:pure_gaussian}).

To state \Cref{hyp:lipschitz_hard}, define an auxiliary function by $\mathcal{V}_1(s)=1+s/2+s^2/4$ for $s \geq 0$.
\begin{assumption}[$h,\Kmax$]
    \label{hyp:lipschitz_hard}
    The step size $h > 0$ and $\Kmax$ satisfy the following inequality:
    \begin{equation}
 [(1+h \ltt_1^{{1/2}}\mathcal{V}_1(h \ltt_1^{1/2}))^{2^{\Kmax}}-1]<1/4 \eqsp .
    \end{equation}
\end{assumption}
This assumption is nearly the same condition that is considered \cite[Eq (10), p.10]{Durmus2017-tf} to prove HMC ergodicity.
When $h\ll 1/\ltt_1^{1/2}$, this assumption is nearly equivalent to $\ltt_1^{{1/2}}2^{\Kmax}h<1/4 $.
Thus \Cref{hyp:lipschitz_hard}($h,\Kmax$) implies a limit on the maximal integration time $h 2^{\Kmax}$.

Instead of assuming a bound on the maximal integration time $h 2^{\Kmax}$ we may impose an additional regularity condition on $U$ in order to prove ergodicity.
To this end, we recall that a function $f: \Rset^d \to \Rset$ is said to be real analytic if it can be locally expanded as a power series, i.e., for every $x_0 \in \Rset^d$ there exists a neighborhood $\msv$ and a sequence $(P_n)_{n=0}^\infty$ of $n$-homogeneous polynomials\footnote{I.e. polynomials whose nonzero terms all have degree $n$.} such that for any $x \in \msv$, $f(x) = \sum_{n=0}^{\infty} P_n (x-x_0)$.

\begin{assumption}
    \label{hyp:analytic_potential}
    The potential $U: \Rset^d \mapsto \Rset$ is real analytic and in addition
    $\lim_{|q| \to \infty} \| \nabla^2 U(q) \| = 0$.
\end{assumption}

While \Cref{hyp:analytic_potential} excludes quadratic potentials, we consider this case separately.

\begin{assumption}
    \label{hyp:pure_gaussian}
    The potential $U$ is a quadratic form, i.e., $\pi$ a non-generate Gaussian distribution.
\end{assumption}
%

Before stating our first results, we introduce some definitions relative to Markov chain theory which are at the basis of our statements.
A kernel $\Kker$ is said to be irreducible if it admits an accessible small set \cite[Definition 9.2.1]{douc2018markov}.
A set $\mathsf{E}\in \mathcal{B}(\mathbb{R}^d)$ is \emph{accessible} for the transition kernel $\Kker$ if for any $q \in \mathbb{R}^d$ we have $\sum_{n=0}^\infty \Kker^n(q, \mse) > 0$.
A set $\msc \subset \mathbb{R}^d$ is a $n$-\emph{small} for $\Kker$ with $n\in \Nset^*$ if there exist  $\varepsilon > 0$ 
and a probability measure $\mu$ on $\mathbb{R}^d$ such that $\Kker^n(q, \msa) \geq \varepsilon \mu(\msa)$ for any $q \in \msc$ and any measurable set $\msa \subset \mathbb{R}^d$.
Let $(X_n)_{n\geq 0}$ be the canonical chain associated with $\Kker$ defined on the canonical space
$((\Rset)^\Nset,\msb(\Rset^d)^{\otimes \Nset})$.
Defining for any measurable set $\msa \subset \mathbb{R}^d$ $N_\msa=\sum_{i=0}^{\infty} \mathbb{1}_\msa(X_i)$ the number of visits to $\msa$, then $\msa$ is said to be reccurent if $\mathbb{E}_x(N_\msa)=+\infty $ for any $x\in \msa$ \cite[Definition 10.1.1]{douc2018markov}.
 The Markov chain $\Kker$ is said to be recurrent if all
accessible sets are recurrent. In particular, if $\Kker$ admits an invariant probability measure and is irreducible then $\Kker$ is called positive \cite[Definition 11.2.7]{douc2018markov} which implies that $\Kker$ recurrent \cite[Theorem 10.1.6.]{douc2018markov}. 
 The period of an accessible small set $\msc$ is the positive integer $d(\msc)$ defined by
\begin{equation}
\textstyle    d(\msc)=\operatorname{g.c.d} \left\{n\in \Nset^*\,:\,\inf_{x\in \msc} \Kker^n(x,\msc)>0 \right\} \eqsp .
\end{equation}
If $\Kker$ is an irreducible Markov kernel, the common period of all accessible small sets is called the period of $\Kker$ \cite[Definition 9.3.1]{douc2018markov}.
 If the period is equal to one, the kernel is said to be aperiodic.

\begin{theorem}
    \label{thm:ergodicity_compile}
    Let $\Kmax \in \nset_{>0}$ and $h >0$. Assume \Cref{hyp:regularity} and that either \Cref{hyp:lipschitz_hard}($h,\Kmax$) or \Cref{hyp:analytic_potential} holds.
    Then we have the following.
\begin{enumerate}[label=(\roman*),wide, labelwidth=!, labelindent=0pt]
    \item \label{thm:item_i_ergo1} The NUTS transition kernel $\KkerU_h$ is irreducible, aperiodic, the Lebesgue measure is an irreducibility measure and any compact set of $\mathbb{R}^d$ is small.
    \item \label{thm:item_ii_ergo2}$\KkerU_h$ is positive recurrent with invariant probability $\pi$ and for $\pi$-almost every $q \in \mathbb{R}^d$, 
    $$ \textstyle \lim _{n \rightarrow+\infty}\left\|\delta_q (\KkerU_h)^n-\pi\right\|_{\mathrm{TV}}=0\eqsp .$$
 \end{enumerate}
    In addition, if \Cref{hyp:pure_gaussian} holds there exists a countable subset $\msh_0 \subset \Rset_{>0}$ such that the conclusions above hold for all $h \in \Rset_{>0} \setminus \msh_0$.
\end{theorem}
\begin{proof}
The proofs of this section are postponed to Section 4 of the Supplementary Material A.
We remark that \ref{thm:item_ii_ergo2} is a consequence of \ref{thm:item_i_ergo1} by \cite[Theorem 13.3.4]{markovchainmeyn2012markov} .
\end{proof}

In order to establish results analogous to \Cref{thm:ergodicity_compile} for the HMC kernel with $T \geq 2$ number of leapfrog steps, \cite{Durmus2017-tf} establish a degree of geometric control over the nonlinear deterministic maps $p \mapsto \operatorname{proj}_1 \Phiverlet[h][T](q_0, p)$, which allows them to conclude that every open set is accessible for HMC in one step.
However, establishing a similar one-step accessibility result for NUTS seems difficult, if not impossible. Instead, we prove accessibility in at most two steps. The proof of NUTS ergodicity presents additional difficulties due to the index selection kernel $\rmq_h$, which encourages the selection of points far from the initial state.
To overcome these challenges, we derive important properties of the stopping times $S$ defined by \eqref{eq:nuts-stopping}.  These results are detailed in Lemma S6 of the Supplementary Material A
These conditions relate to maps $F^{T_1, T_2}_q$ defined for any $q \in \mathbb{R}^d$ and $T_1, T_2 \in \mathbb{Z}^2$ with $T_1 \neq T_2$ by
\begin{equation}
    \label{eq:F_T1T2}
F^{T_1, T_2}_q: p \in \mathbb{R}^d \mapsto p_{T_1}^\top (q_{T_2} - q_{T_1}) \eqsp,
\end{equation}
where $q_i, p_i = \Phiverlet[h][i](q_0, p_0)$ for any $i \in \mathbb{Z}$, and read for $h >0$ and $\Kmax \in \nset_{>0}$:

\begin{assumption}[$(h,\Kmax)$]
    \label{hyp:3}
    \begin{enumerate}[label=(\roman*),wide, labelwidth=!, labelindent=0pt]
        \item\label{hyp:item_i_topo}
         For any $q\in \Rset^d$, the following set is dense,
         \begin{equation}
            \label{eq:msf-0}
            \msf_{q,-0}=\{p\in \Rset^d:   F^{T_1,T_2}_q(p)\neq 0,\, T_1,T_2\in [-2^{\Kmax}+1:2^{\Kmax}-1]^2,\, T_1\neq T_2\} \eqsp.
        \end{equation}
\item \label{hyp:item_ii_homeo} For any $q_0\in \Rset^d$, there exist $p_0\in \Rset^d,\,r_H>0$ such that for any $T\in [-2^{\Kmax}+1:2^{\Kmax}-1]$ with $T\neq 0$,
    \begin{equation}
        \psi_{q_0}^{(T)}|_{\mathrm{B}(p_0,r_H)}: p\in \mathrm{B}(p_0,r_H)\mapsto \operatorname{proj}_1\Phiverlet[h][T](q_0,p) 
    \end{equation}
    is a local homeomorphism. 
\end{enumerate}
\end{assumption}

In contrast to the easily verifiable \Cref{hyp:lipschitz_hard}($h,\Kmax$), \Cref{hyp:analytic_potential} or \Cref{hyp:pure_gaussian}, the condition \Cref{hyp:3}$(h,\Kmax)$-\ref{hyp:item_i_topo} is technical but less stringent, since it focuses on pathological cases related to the stopping time that cause the main technical difficulties in the proof of irreducibility.
This assumption allows us to make clear and precise the different steps of the proof of \Cref{thm:ergodicity_compile}. 
\begin{theorem}
    \label{thm:ergodic_general}
    Assume \Cref{hyp:regularity}, \Cref{hyp:3}$(h,\Kmax)$, for $h >0$ and $\Kmax \in \nset_{>0}$.
     Then, the conlusions \ref{thm:item_i_ergo1} and \ref{thm:item_ii_ergo2} of \Cref{thm:ergodicity_compile} hold.
\end{theorem}
\Cref{thm:ergodic_general} is a consequence of the general result \cite[Theorem 14.0.1]{markovchainmeyn2012markov} and our results \Cref{thm:nuts-accessibility} and \Cref{thm:3small} below.
At this stage, \Cref{thm:ergodicity_compile} follows after we show that the set of assumptions (\Cref{hyp:regularity}, \Cref{hyp:lipschitz_hard}($h$)), (\Cref{hyp:regularity}, \Cref{hyp:analytic_potential}) and \Cref{hyp:pure_gaussian} are strictly stronger than \Cref{hyp:3}$(h,\Kmax)$:

\begin{proposition}
  \label{prop:general_condition_check}
  Let $\Kmax \in \nset_{>0}$. 
    \begin{enumerate}[label=(\alph*),wide, labelwidth=!, labelindent=0pt]
    \item
    \label{prop:condition_check}
    Assume \Cref{hyp:regularity} and \Cref{hyp:lipschitz_hard}($h,\Kmax$) or \Cref{hyp:analytic_potential}. Then the NUTS transition kernel $\KkerU_h$ satisfies hypothesis \Cref{hyp:3}$(h,\Kmax)$ for any $h >0$.
    \item     \label{prop:condition_check_gaussian}
    Assume \Cref{hyp:pure_gaussian}. Then, there exists a countable set $\msh_0\subset \Rset_{\geq0}$ such that for any $h\in \Rset_{>0}\setminus \msh_0$ the NUTS transition kernel $\KkerU_h$ satisfies \Cref{hyp:3}$(h,\Kmax)$.
    \end{enumerate}
\end{proposition}

The main technical challenge in proving accessibility for the NUTS transition kernel arises from the dependence of index selection probabilities on the entire trajectory, which in turn relies on the global geometry of the potential energy function $U$. However, the following result overcomes this challenge and establishes accessibility from every point in either one or two steps.
\begin{theorem}
    \label{thm:nuts-accessibility}
    Assume \Cref{hyp:regularity} and \Cref{hyp:3}$(h,\Kmax)$-\ref{hyp:item_i_topo}, for $h >0$ and $\Kmax \in \nset_{>0}$.
    For the NUTS transition kernel $\KkerU_h$,
    every open set $\mse \subset \mathbb{R}^d$ is accessible for all $q$ in $\Rset^d$.
    Moreover, for every open set $\mse \subset \mathbb{R}^d$ and for any $q_0\in \Rset^d$, there exist $\msw(q_0)$ a neighborhood of $q_0$, $m_\msw(q_0)>0$ and $j(q_0)\in \{1,2\}$ such that for any $q\in \msw(q_0)$:
    \begin{equation}
        (\KkerU_h)^{j(q_0)}(q,\mse)\geq m_\msw(q_0)>0 \eqsp .
      \end{equation}
\end{theorem}

We also need to show that the transition kernel admits small sets.

\begin{theorem}
    \label{thm:1nuts-small-sets}
    Assume \Cref{hyp:regularity} and \Cref{hyp:3}$(h,\Kmax)$-\ref{hyp:item_ii_homeo}, for $h >0$ and $\Kmax \in \nset_{>0}$. 
    For every $q \in \Rset^d$ there exists an $r > 0$ for which $\mathrm{B}(q, r)$ is $1$-small for the NUTS transition kernel $\KkerU_h$.
\end{theorem}
Finally, as a byproduct of the proofs of \Cref{thm:nuts-accessibility} and \Cref{thm:1nuts-small-sets}:
\begin{theorem}
    \label{thm:3small}
    Assume \Cref{hyp:regularity} and \Cref{hyp:3}$(h,\Kmax)$, for $h >0$ and $\Kmax \in \nset_{>0}$. All compact sets are 3-small for the NUTS transition kernel $\KkerU_h$. Consequently, the NUTS transition kernel $\KkerU_h$ is aperiodic.
\end{theorem}

\section{Geometric ergodicity}

 \label{section:ergo_geo}
In this section, we give conditions on the potential $U$ which imply that the NUTS kernel converges geometrically to its invariant distribution.
Let $\VFL: \Rset^d\to [1,+\infty )$ be a measurable function and $\Kker$ be a Markov kernel on $(\Rset^d,\mathcal{B}(\Rset^d))$.
Recall that the definition of $\VFL$-uniformly geometrically ergodicity is given in \eqref{eq:v-uniform-ergodicity}.
By \cite[Theorem 16.0.1]{markovchainmeyn2012markov}, if $\Kker$ is aperiodic, irreducible and satisfies a Foster--Lyapunov drift condition,
i.e., there exist a small set $\msc\in \mathcal{B}(\Rset^d)$ for $\Kker$, $\lambda \in [0,1)$ and $b<+\infty$ such that
\begin{equation}
    \label{Foyster-Lyapunov-condtion}
    \Kker \VFL\leq \lambda \VFL+b \mathbbm{1}_\msc\eqsp,
\end{equation}
then $\Kker$ is $\VFL$-uniformly geometrically ergodic. If a function $\VFL: \Rset^d \to [0,+\infty)$ satisfies \eqref{Foyster-Lyapunov-condtion},
 then $\VFL$ is said to be a Foster--Lyapunov function for $P$.

Define for $a> 0$ and $q\in \Rset^d$, the function
\begin{equation}
    \label{eq:Va}
    \VFL_a(q)=\exp(a|q|)\eqsp .
\end{equation}
In what follows we show that, for any $a > 0$, $\VFL_a$ is a Foster--Lyapunov function for the NUTS kernel under the same assumptions on the potential $U$ considered for HMC in \cite{Durmus2017-tf}. 
Let $m\in(1,2]$.
\begin{assumption}[$m$]
    \label{hyp:rappel}
    \begin{enumerate}[label=(\roman*),wide, labelwidth=!, labelindent=0pt]
        \item \label{hyp:rappel:item_tailgrad} There exists $\msm_1\geq 0$ such that for any $q\in \Rset^d$,
    $$|\nabla U(q)|\leq \msm_1 (1+|q|^{m-1})\eqsp .$$
    \item \label{hyp:rappel:item_rappel}There exist $A_1> 0$ and $A_2\in\Rset$ such that for any $q\in \Rset^d$,
        $$\big(\nabla U(q)\big)^\transpose q \geq A_1 |q|^m-\msa_2 \eqsp .$$
    \item \label{hyp:rappel:item_tailgradv2}$U\in \mathrm{C}^3(\Rset^d)$ and there exists $A_3 > 0$ such that for any $q\in \Rset^d$ and k=2,3:
    $$ |\dd^k U(q)|\leq A_3 (1+|q|)^{m-k}\eqsp . $$
    \item \label{hyp:rappel:item_rappelv2}There exist $A_4 > 0$ and $R_U\in \Rset_{\geq0}$ such that for any $q\in\Rset^d$,$ |q|\geq R_U$,
    $$\left(\nabla U(q)\right)^\transpose \dd^2 U(q) \nabla U(q) \geq A_4 |q|^{3m-4} \eqsp .$$
    \end{enumerate}
\end{assumption}

We remark that \Cref{hyp:rappel}(m), originating from \cite{Durmus2017-tf}, concerns the geometry of the tail of the target distribution $\pi$.
 Conditions \Cref{hyp:rappel}($m$)-\ref{hyp:rappel:item_rappel} and \Cref{hyp:rappel}($m$)-\ref{hyp:rappel:item_tailgrad} induce a restoring force in the tails of $\pi$ and will imply the stability of the proposal kernel.
 This will be more transparent after \Cref{lemma:normposition} below.
 Conditions \Cref{hyp:rappel}($m$)-\ref{hyp:rappel:item_rappelv2} and \Cref{hyp:rappel}($m$)-\ref{hyp:rappel:item_tailgradv2} are both strenghtenings of \Cref{hyp:rappel}($m$)-\ref{hyp:rappel:item_rappel} and \Cref{hyp:rappel}($m$)-\ref{hyp:rappel:item_tailgrad}, respectively, and are needed in order to guarantee that proposals which move away from the center are rejected with probability approaching one in the tails of $\pi$.
 These last conditions are pretty mild: a smooth perturbation of a Gaussian target satisfies \Cref{hyp:rappel}($m$). 
More generally, they are satisfied by $m$-homogeneous quasi-convex functions and by perturbations of such functions (see \cite[Proposition 6]{Durmus2017-tf}).
Recall that a function $U_0$ is $m$-homogeneous quasi-convex outside a ball of radius $R_1$ if the following conditions are satisfied:
\begin{itemize}
  \item For all $t \geq 1$ and $q \in \mathbb{R}^d,|q| \geq R_1, U_0(t q)=t^m U_0(q)$.
\item  For all $q \in \mathbb{R}^d,|q| \geq R_1$, the level sets $\left\{x: U_0(x) \leq U_0(q)\right\}$ are convex.
\end{itemize}

In the case $m=2$ we propose the following milder alternative to \Cref{hyp:rappel}($2$)-\ref{hyp:rappel:item_tailgradv2},\ref{hyp:rappel:item_rappelv2}:
\begin{assumption}
    \label{hyp:gaussian_perturbation}
There exists a twice continuously differentiable $\tilde{U}: \mathbb{R}^d \rightarrow \mathbb{R}$ and a positive definite matrix $\boldsymbol{\Sigma}$ such that $U(q)=q^\top\boldsymbol{\Sigma} q / 2+\tilde{U}(q)$, and there exist $A_5 \geq 0$ and $\varrho \in[1,2)$ such that for any $q, x \in \mathbb{R}^d$
$$
\begin{gathered}
|\tilde{U}(q)| \leq A_5\left(1+|q|^{\varrho}\right), \quad|\nabla \tilde{U}(q)| \leq A_5\left(1+|q|^{\varrho-1}\right)\eqsp, \\
|\nabla \tilde{U}(q)-\nabla \tilde{U}(x)| \leq A_5|q-x|\eqsp .
\end{gathered}
$$
\end{assumption}
\begin{remark}
    \label{rmk:implication_gaussian}
     It is straightforward to check that under \Cref{hyp:gaussian_perturbation}, the conditions \Cref{hyp:regularity} and \Cref{hyp:rappel}($2$)-\ref{hyp:rappel:item_tailgrad},\ref{hyp:rappel:item_rappel} hold.
\end{remark}

The following Lemma gives the main ingredients to establish the drift condition on the kernel $\KkerU_h$.
\begin{lemma}
    \label{lemma:sketchofproof}
    Assume either \Cref{hyp:regularity} and \Cref{hyp:rappel}($m$)-\ref{hyp:rappel:item_tailgrad} for some $m\in (1,2]$, or \Cref{hyp:gaussian_perturbation}.
    Let $\gamma \in ((m-1)/2,m-1)$ and denote
    $\mathrm{B}(q_0) = \{ p\in \Rset^d : |p|\leq |q_0|^{\gamma}\}$ for any $q_0\in \Rset^d$.
    Let $h > 0$ and suppose that
     there exists $R_0>0$ such that for any $q_0\in \Rset^d$ with $|q_0|\geq R_0$ and $p_0\in \mathrm{B}(q_0)$
     we have for any $j\in[-2^{\Kmax},2^{\Kmax}]\setminus \{0\}$
     \begin{equation}
        \label{eq:norm_position_condition}
        \ee |\operatorname{proj}_1\Phiverlet[h][j](q_0, p_0)|-|q_0|\leq -1 
     \end{equation}
     and for $ j\in\{-1,1\}$
    \begin{equation}
        \label{eq:lemma_hamiltonian_condition}
        H(\Phiverlet[h][j](q_0,p_0))-H(q_0,p_0)\leq 0\eqsp . 
    \end{equation}
     Then, there exist $\lambda \in (0,1)$ and $b,R'>0$ such that 
     \begin{equation}
         \KkerU_{h}\VFL_a \leq\lambda  \VFL_a +b \mathbbm{1}_{\bar{\mathrm{B}}(0,R')} \eqsp .
     \end{equation}
    
     
\end{lemma}
\begin{proof}
    The proof is postponed to Section 5.1 of the Supplementary Material A.
\end{proof}

Based on the previous lemma, we shall analyze the dynamics when the norm of the position $|q_0|$ is large enough and when the norm of the momentum $|p_0|$ is smaller than $|q_0|^\gamma$ with $\gamma\in ((m-1)/2,m-1) $.
 In that case, we aim to establish that the positions on the orbit $\orbit_{[-2^\Kmax+1:2^\Kmax-1] \setminus \{0\}}(q_0, p_0) = \{ \Phiverlet[h][j](q_0, p_0) : j \in [-2^\Kmax+1:2^\Kmax-1] \setminus \{0\} \}$ lie in the ball $\mathrm{B}(0, |q_0|-1)$, and one of these points is always accepted by the index selection rule due to \eqref{eq:lemma_hamiltonian_condition}.
This is done in \Cref{lemma:normposition} and \Cref{prop:degrowth_energy} below, respectively.
\begin{proposition}
    \label{prop:degrowth_energy}
    Assume either \Cref{hyp:regularity} and \Cref{hyp:rappel}($m$) for some $m\in (1,2]$, or \Cref{hyp:gaussian_perturbation}. Let $\gamma \in (0,m-1)$.
    \begin{enumerate}[label=(\alph*)]
        \item If $m\in(1,2)$ and $h > 0$, there exists $R_H >0$ such that for any $(q_0,p_0)\in (\Rset^d)^2$ with $|q_0|\geq R_H$ and $|p_0|\leq |q_0|^\gamma$ we have $H(\Phiverlet[h][j](q_0,p_0))-H(q_0,p_0)\leq 0$ for $j\in \{-1,1\}$.
        \item If $m=2$, there exists $\bar{S}>0$ such that for any $h\in(0,\bar{S}]$, there exists $R_H >0$ such that for any $(q_0,p_0)\in (\Rset^d)^2$ with $|q_0|\geq R_H$ and $|p_0|\leq |q_0|^\gamma$ we have $H(\Phiverlet[h][j](q_0,p_0))-H(q_0,p_0)\leq 0$ for $j\in \{-1,1\}$.
    \end{enumerate}

\end{proposition}
\begin{proof}
    The proof is postponed to Section 5.2 of the Supplementary Material A.
\end{proof}


\begin{lemma}
    \label{lemma:normposition}
    Assume either \Cref{hyp:regularity},  \Cref{hyp:rappel}($m$)-\ref{hyp:rappel:item_tailgrad},\ref{hyp:rappel:item_rappel} for some $m\in (1,2]$ or \Cref{hyp:gaussian_perturbation} and let $T\in \Nset^*$.
    \begin{enumerate}[label=(\alph*)]
        \item If $m<2$, let $\gamma \in \big(\max\big(2(m-1)-1,(m-1)/2\big),m-1\big)$.
    Then, for any $h>0$, there exists $R_0>0$ such that for any $(q_0,p_0)\in (\Rset^d)^2$ with $|q_0|\geq R_0$ and $p_0\leq |q_0|^\gamma$, for any $j\in[-T:T]$ with $j\neq 0$ we have
     $$\ee |\operatorname{proj}_1\Phiverlet[h][j](q_0, p_0)|-|q_0|\leq -1\eqsp . $$
     \item If $m=2$, let $\gamma= 2/3$. Denote
     \begin{equation}
         \label{eq:V2}
     \mathcal{V}_2(s)=\msm_1/\ltt_1^{\frac{1}{2}}+\msm_1 s/2+\ltt_1^{\frac{1}{2}}\msm_1 s^2/4 \eqsp ,
     \end{equation} 
     $\msm_1$ is well defined even under \Cref{hyp:gaussian_perturbation} by \Cref{rmk:implication_gaussian}.
     Let $\bar{S}>0$ be such that $\Theta(s)<A_1$ for any $s\in(0,\bar{S}]$, with
     \begin{align}
        \Theta(s)=
     \begin{multlined}[t]
         2\ltt_1^{\frac{1}{2}} \mathcal{V}_2(s)\big(\exp\big(\ltt_1^{\frac{1}{2}}s\mathcal{V}_1(\ltt_1^{\frac{1}{2}}s)\big)-1\big)\\
         +6 s^2\big[\msm_1^2+\ltt_1 \mathcal{V}_2^2(s)\big(\exp\big(\ltt_1^{\frac{1}{2}}s\mathcal{V}_1(\ltt_1^{\frac{1}{2}}s)\big)-1\big)^2\big] \eqsp
     \end{multlined}
    \end{align}
     and where $\mathcal{V}_1$ is defined in \Cref{hyp:lipschitz_hard}($h,\Kmax$).
     Then for any $h\in(0,\bar{S}/T]$, there exists $R_0>0$ such that for any $(q_0,p_0)\in (\Rset^d)^2$ with $|q_0|\geq R_0$ and $p_0\leq |q_0|^\gamma$, for any $j\in[-T,T]$ with $j\neq 0$ we have
     $$\ee |\operatorname{proj}_1\Phiverlet[h][j](q_0, p_0)|-|q_0|\leq -1 \eqsp .$$
    \end{enumerate}
\end{lemma}
\begin{lemma}
    The proof is postponed to Section 5.3 of the Supplementary Material A.
\end{lemma}

The geometric ergodicity of the NUTS sampler follows.
\begin{theorem}
    \label{thm:ergo_geo}
    Assume \Cref{hyp:3}$(h,\Kmax)$, for $h >0$ and $\Kmax \in \nset_{>0}$.
    Assume either \Cref{hyp:regularity}, \Cref{hyp:rappel}($m$) for some $m\in (1,2]$, or \Cref{hyp:gaussian_perturbation}.
    \begin{enumerate}[label=(\alph*)]
        \item \label{thm:last_a} Case $m<2$: for $a> 0$,
        the No U-turn sampler kernel $\KkerU_h$ is $\VFL_a$-uniformly geometrically ergodic.
        \item \label{thm:last_b} Case $m=2$: there exists $\bar{S}>0$ such that for any $a> 0$ and $h>0$ such that $h2^{\Kmax} \leq \bar{S}$ and \Cref{hyp:lipschitz_hard}($h,\Kmax$),
        the No U-turn sampler kernel $\KkerU_h$ is $\VFL_a$-uniformly geometrically ergodic.
    \end{enumerate}
     
\end{theorem}
\begin{proof}
    The proof is postponed to Section 5.4 of the Supplementary Material A.
\end{proof}

We remark that only the condition \Cref{hyp:rappel}($m$)-\ref{hyp:rappel:item_tailgrad} is imposed in \Cref{lemma:sketchofproof},
compared to \Cref{lemma:normposition} where \Cref{hyp:rappel}($m$)-\ref{hyp:rappel:item_rappel} is also needed.
Conditions \Cref{hyp:rappel}($m$)-\ref{hyp:rappel:item_rappelv2}, \ref{hyp:rappel:item_tailgradv2} are used for \Cref{prop:degrowth_energy}.
Regarding the conditions on the potential, the bottleneck of the demonstration is \Cref{prop:degrowth_energy}, which relies on \cite[Proposition 7]{Durmus2017-tf} and the symmetry of the Hamiltonian in the momentum variable, i.e., $H(\cdot.,p)=H(\cdot,-p)$ for any $p\in\Rset^d$.
The most restrictive assumption on the stepsize appears in \Cref{lemma:normposition} for the case $m=2$.
Compared with the geometric ergodicity of the HMC sampler in the case $m=2$ \cite[Theorem 9]{Durmus2017-tf}, instead of having $hT\leq \bar{S} $ where $T$ is the number of leapfrog steps, we have $h 2^{\Kmax}\leq \bar{S} $ where $2^{\Kmax}$ is the maximum number of leapfrog steps for the NUTS sampler.

\section{General properties on Hamiltonian Monte Carlo}
\label{section:general_HMC}
In this section, we extend and improve some results presented in \cite{Durmus2017-tf}.
 Our aim is to establish the convergence of the HMC kernel under milder conditions on the stepsize.

 Let $\Khmc\in \Nset^*$ be the number leapfrog steps and $h>0$ be the stepsize. 
 The HMC kernel is defined, for any $q_0\in \Rset^d$, $\msa \in \mathcal{B}(\Rset^d)$, by
 \begin{align}
   \label{eq:def_KkerH}
    &\KkerH(q_0,\msa)=  \int \rho_0(p_0)\alpha_{h,\Khmc}(q_0,p_0) \updelta_{\operatorname{proj}_1\Phiverlet[h][\Khmc](q_0,p_0)}(\msa)  \dd p_0+(1-\alpha_{h,\Khmc}(q_0,p_0))\updelta_{q_0}(\msa)\eqsp ,
\end{align}
where for any $q_0,p_0 \in (\Rset^d)^2$, the acceptance ratio is
\begin{equation}
    \alpha_{h,\Khmc}(q_0,p_0)= 1\wedge \exp\defEns{H(q_0,p_0)-H(\Phiverlet[h][\Khmc](q_0,p_0))} \eqsp .
\end{equation}

We consider in this section the following assumption on the potential $U$.

 \begin{assumption}
    \label{hyp:gaussian_perturbation_nice}
There exist $\tilde{U}: \mathbb{R}^d \rightarrow \mathbb{R}$, twice continuously differentiable and  a real positive definite matrix $\boldsymbol{\Sigma}$ such that $U(q)=q^\top\boldsymbol{\Sigma} q / 2+\tilde{U}(q)$.
 In addition, there exist  $A_5 \geq 0$ and $\varrho \in[1,2)$ such that for any $q, x \in \mathbb{R}^d$
$$
\begin{gathered}
|\tilde{U}(q)| \leq A_5\left(1+|q|^{\varrho}\right), \quad|\nabla \tilde{U}(q)| \leq A_5\left(1+|q|^{\varrho-1}\right)\eqsp, \\
|\nabla \tilde{U}(q)-\nabla \tilde{U}(x)| \leq A_5|q-x| \eqsp .
\end{gathered}$$
\end{assumption}
The main result of this section is the following.
\begin{theorem}
  \label{thm:ergodicity_compile_HMC}
  Assume \Cref{hyp:regularity} and let $h >0$ and $T \in\nsets$. 
Suppose in addition \Cref{hyp:gaussian_perturbation_nice} or
    \begin{equation}
        \label{eq:transitivity-condition}
        \ltt_1 h^2 < 2 (1 - \cos (\uppi/T))\eqsp.
      \end{equation}
Then, there exists a countable set $\msh_0\subset \Rset_{\geq0}$, defined in \Cref{lemma:gaussian_case} under \Cref{hyp:gaussian_perturbation_nice} and $\msh_0 = \Rset_{\geq0}$ otherwise, such that if $h \not\in  \msh_0 $ we have,
\begin{enumerate}[label=(\alph*),wide, labelwidth=!, labelindent=0pt]
    \item \label{thm:item_i_ergo1_HMC} the HMC kernel $\KkerH$ is irreducible, aperiodic, the Lebesgue measure is an irreducibility measure and any compact set of $\mathbb{R}^d$ is 1-small.
    \item \label{thm:item_ii_ergo2_HMC}$\KkerH$ is positive recurrent with invariant probability $\pi$ and for $\pi$-almost every $q \in \mathbb{R}^d$, 
    $$\textstyle \lim _{n \rightarrow+\infty}\|\delta_q (\KkerH)^n-\pi\|_{\mathrm{TV}}=0\eqsp .$$
 \end{enumerate}
\end{theorem}

\begin{proof}
    The proof is postponed to Section 6.3 of the Supplementary Material A.
  \end{proof}
  Let us compare our result with \cite[Theorem 1]{Durmus2017-tf}.
  
   First, in the case $\limsup_{|x| \to\plusinfty}[ |\nabla U(x)|/|x|^2] \not =0$, \cite[Theorem 1]{Durmus2017-tf} only shows that HMC is ergodic if $h$ and $T$ satisfy
  \begin{equation}
    \label{eq:2_cond_old}
    [(1+h \ltt_1^{\frac{1}{2}}\mathcal{V}(h \ltt_1^\frac{1}{2}))^{2^{T}}-1]<1 \eqsp,
  \end{equation}
  where $\mathcal{V}(s)=1+s/2+s^2/4$ for $s\in \Rset_{\geq0}$.  We show in
  Section 4.4 of the Supplementary Material A
   that this condition is
  strictly stronger than \eqref{eq:transitivity-condition}. Finally,
  under \Cref{hyp:gaussian_perturbation_nice}, we obtain ergodicity
  for HMC for any given number of leapfrog steps $T$ and for
  $\Leb$-almost every choice of stepsize $h$ and in particular, for
  $\Leb$-almost every choice of integration time $hT$ (for a fixed
  $T$). This result is in accordance with the ergodicity properties of
  the ideal HMC algorithm (i.e., the exact Hamiltonian dynamics
  \eqref{eq:hamiltonian_system} for a fixed integration time $\tint >0$ instead of the leapfrog scheme
  $\Phi_h^{(T)}$ in \eqref{eq:def_KkerH}) in the case where $\pi$ is a
  Gaussian distribution. Indeed, in that case, explicit expression of
  the Hamiltonian dynamics \cite[Proposition
  3.1]{Geo_integratorsbou2018geometric} shows that there exists a
  countable set $\tilde{\msh}^0$ included in $\Rset_{\geq0}$ such that if
  the integration time $\tint \in \tilde{\msh}^0$, the resulting ideal
  HMC algorithm is periodic, whereas if $\tint \not \in \tilde{\msh}^0$, the algorithm is ergodic (and even geometrically ergodic).

  To show \Cref{thm:ergodicity_compile_HMC}, we extend part of the results obtained in \cite{Durmus2017-tf}. First, the proof of the ergodicity of HMC in   \cite{Durmus2017-tf}  use that the map $p_0 \mapsto \operatorname{proj}_1 \Phiverlet[h][T](q_0, p_0)$ is a bi-Lipschitz homeomorphism for any $q\in\Rset^d$ by assuming \Cref{hyp:regularity} and \eqref{eq:2_cond_old}. We show that in fact this is still true under \eqref{eq:transitivity-condition}.  


%
\begin{theorem}\label{thm:trajectory-transitivity}
  Assume \Cref{hyp:regularity} and let $h >0$ and $T \in\nsets$ satisfying \eqref{eq:transitivity-condition}.
  Let $q_0 \in \rset^d$. Then, for any $\tilde{q}$, there exists a unique pair $(p,\tilde{p}) \in(\Rset^d)^2$ such that
  \begin{equation}
        \Phiverlet[h][T] \big(q_0, p \big) = (\tilde{q}, \tilde{p}) \eqsp.
      \end{equation}
      Therefore, $\psi_q : p\in \Rset^d \to \operatorname{proj}_1 \Phiverlet[h][T](q,p)$ is a one-to-one continuous map. Finally, its inverse $\psi^{\inv}_q$ is Lipschitz. 
\end{theorem}
\begin{proof}
    The proof is postponed to Section 6.1 of the Supplementary Material A.
\end{proof}

\begin{remark}
    \label{rm:Tequal1}
    For $T = 1$ the corresponding statement  trivially holds  without the condition \eqref{eq:transitivity-condition}.
\end{remark}

The condition \eqref{eq:transitivity-condition} is sharp in the sense of the following counterexample.
Consider the standard Gaussian target $U(q) = |q|^2/2$  which satisfies \Cref{hyp:regularity} with $\ltt_1 = 1$.
    Given a number of leapfrog steps $T$ and choosing the stepsize $h^2 = 2 (1 - \cos (\uppi/T))$,  explicit calculations show that for any $(q_0, p_0) \in (\Rset^d)^2$, it holds that $\Phiverlet[h][T](q_0, p_0) = (-q_0, -p_0)$. Therefore, the conclusion of \Cref{thm:trajectory-transitivity} cannot hold in this situation.

    However, for a given number of leapfrog steps $T$, in the Gaussian case still, we can show that $\tilde{q} \mapsto \bar{\Psi}_h^{(T)}(q, \tilde{q})$ is a $\rmC^1$-diffeomorphism still, if $h$ do not belong to a countable subset $\msh_0 \subset \rset$, as illustrated by the following result.
    Indeed, when $\pi$ is Gaussian, leapfrog iterates can be explicitly expressed polynomial in $h$ and linear in $q_0,p_0\in (\Rset^d)^2$. Therefore  their analysis can be simplified.
Finally, note that we state here further properties of leapfrog iterates that are used to prove the convergence of the NUTS kernel as $\pi$ is Gaussian.

\begin{lemma}
    \label{lemma:gaussian_case}
    If there exists a real positive definite matrix $\Sigma$ such that for any $q\in \Rset^d$  $U(q)=q^\top \Sigma q/2$, then 
    there exists a countable set $\msh_0\subset \Rset_{\geq0}$ such that for any $h\in \Rset_{>0}\setminus \msh_0$, $ (T,T_1,T_2)\in \Zset^3$ with $T_2\neq T_1, \, T\neq 0$ and $q_0\in \Rset^d$,
        the functions $\psi_{q_0}^{(T)},\nabla F_{q_0}^{(T_1,T_2)}$ are linear one-to-one maps, where 
    \begin{align}
      &       F_{q_0}^{(T_1,T_2)}:\,p_0\in \Rset^d\mapsto p_{T_1}^T(q_{T_2}-q_{T_1}),\quad \psi_{q_0}^{(T)}:\,p_0\in \Rset^d \mapsto q_T\eqsp ,
    \end{align}
    denoting for $i\in\zset$ $q_i=\operatorname{proj}_1\Phiverlet[h][i](q_0,p_0)$ and $ p_{i}=\operatorname{proj}_2\Phiverlet[h][i](q_0,p_0)$.
\end{lemma}
\begin{proof}
    The proof is postponed to Section 6.2 of the Supplementary Material A.
  \end{proof}
  Based on \Cref{lemma:gaussian_case}, the proof of \Cref{thm:ergodicity_compile_HMC} under \Cref{hyp:gaussian_perturbation_nice} then follows from an homotopy argument. 

\hypersetup{
    colorlinks = false,
    linkbordercolor = {white},
    linkcolor = {white}
  }
  
{\color{white}\eqref{eq:F_T1T2}\eqref{eq:transitivity-condition}\eqref{eq:norm_position_condition}\eqref{eq:lemma_hamiltonian_condition}\eqref{eq:T_-}\ref{rm:growthmaximality}\eqref{eq:hat_omega}\eqref{eq:Omega}\eqref{eq:def_check_omega}\eqref{def:recur_def_rmq}\eqref{eq:nuts-I},\eqref{eq:def_hatPi},\eqref{eq:transition_prob_q}\ref{prop:condition_check}\ref{prop:condition_check_gaussian}
    \ref{prop:general_condition_check}
\ref{section:general_HMC} \ref{thm:trajectory-transitivity} \ref{lemma:gaussian_case}
\ref{ergodicity_section}\ref{thm:3small}\ref{thm:1nuts-small-sets} \ref{thm:nuts-accessibility}
\ref{prop:condition_check}\ref{hyp:3}\ref{thm:ergodicity_compile} \ref{eq:scrU} \ref{hyp:regularity} \ref{eq:iteration_verlet}}

\hypersetup{
    colorlinks = True,
    linkbordercolor = {black},
    linkcolor = {red}
  }
\bibliographystyle{plain}
\bibliography{bibliography}

\newtheorem{acks}{Acknoledgement}
  \begin{acks}[Acknowledgments]
  MK, MV and ES are supported by Academy of Finland (Finnish Centre of Excellence in Randomness and Structures, grants 346311 and 346305).

  A.D. would like to thank the Isaac Newton Institute for Mathematical Sci- ences for support and hospitality during the programme The
mathematical and statistical foundation of future data-driven engineering when work
on this paper was undertaken. This work was supported by: EPSRC grant number
EP/R014604/1.
\end{acks}

\appendix



\section{Proof of \Cref{sec:dynamic-hmc}}

\subsection{Proof of \Cref{prop:trajectory-invariance}}
\label{sec_supp:trajectory_invariance}

      Let $\varphi : \rset^d \to \rset$ be a measurable bounded function.
      Using \eqref{eq:transition-kernel} and Fubini's theorem,
      we get
        \begin{align}
    &        \int  \pi(q') \int \Kker_h(q', \dd q) \varphi(q) \dd q' =    \int  \tpi(q', p') \int \tilde{\Kker}_h((q', p'), \dd q) \varphi(q) \dd q' \dd p'
            \\
            &\quad =
            \begin{multlined}[t]
                \int \dd q' \dd p' \, \tpi(q', p')
                \sum_{\msj \subset \mathbb{Z}} \sum_{j \in \msj}
                 \rmpp_h(\msj \mid q', p') \rmqq_h(j \mid \msj, q', p') \varphi(\operatorname{proj}_1\{\Phiverlet[h][j](q', p')\})
            \end{multlined}
            \\
            &\quad =
            \begin{multlined}[t]
                \sum_{\msj \subset \mathbb{Z}} \sum_{j \in \msj}
                \int \dd q' \dd p' \, \tpi(q', p') 
                 \rmpp_h(\msj \mid q', p') \rmqq_h(j \mid \msj, q', p') \varphi(\operatorname{proj}_1\{\Phiverlet[h][j](q', p')\}) \eqsp.
            \end{multlined}
        \end{align}
    
        By the change of variables $(q,p) = \Phiverlet[h][j](q', p')$, and using the fact that the Jacobian determinant of $(q',p') \mapsto \Phiverlet[h][j](q', p')$ is equal to $1$ (see e.g.,
        \cite[Theorem and Proposition 2.1]{Geo_integratorsbou2018geometric}), 
        \begin{align}
        &        \int  \pi(q') \int \Kker_h(q', \dd q) \varphi(q) \dd q' \\
            & =
            \begin{multlined}[t]
                \sum_{\msj \subset \mathbb{Z}} \sum_{j \in \msj}
                \int \dd q \dd p \, \tpi(\Phiverlet[h][-j](q, p))
                 \rmpp_h(\msj \mid \Phiverlet[h][-j](q, p)) \rmqq_h(j \mid \msj, \Phiverlet[h][-j](q, p)) \varphi(q)
            \end{multlined}
            \\
            & =
            \begin{multlined}[t]
                \int \dd q \dd p \sum_{\msj \subset \mathbb{Z}} \sum_{j \in \mathbb{Z}} \mathbbm{1}_\msj(j) \tpi(\Phiverlet[h][-j](q, p))
                 \rmpp_h(\msj \mid \Phiverlet[h][-j](q, p)) \rmqq_h(j \mid \msj, \Phiverlet[h][-j](q, p)) \varphi(q)
            \end{multlined}
            \\
            & =
            \begin{multlined}[t]
                \int \dd q \dd p \sum_{\msj' \subset \mathbb{Z}} \sum_{j' \in \mathbb{Z}} \mathbbm{1}_{\msj'}(0)\tpi(\Phiverlet[h][-j'](q, p))
                \\
    \qquad  \times            \rmpp_h(\msj'+j' \mid \Phiverlet[h][-j'](q, p)) \rmqq_h(j' \mid \msj'+j', \Phiverlet[h][-j'](q, p)) \varphi(q) \eqsp,
    \end{multlined}
        \end{align}
    where for the last line, we use the change of indices, $\msj \mapsto \msj -j$.
    Using  \eqref{eq:trajectory-invariance} implies that $$   \int  \pi(q') \int \Kker_h((q', p'), \dd q) \varphi(q) \dd q' = \int  \pi(q) \varphi(q) \dd q$$ and concludes the proof.
        

\section{Proof of \Cref{sec:nuts-its-invariance}}
    \subsection{Proof of \Cref{lemma:comeback}}
    \label{sec_supp:comeback}
    Recall $\msi = B_K(v|_K)=B_K-(2^K-1-v|_K)$ with $v\in B_{\Kmax}$.
Let $a,b\in \msi$.
By the definition of $\bar{\rmq}_h$ \eqref{eq:transition_prob_q}, we have
    \begin{align}
        \bar{\rmq}_h(a,b\mid \msi,q_0,p_0)&= \rmq_h(b-a  \mid \msi - a, \Phiverlet[h][a](q_0, p_0))  \\
    &= \rmq_h(b-a  \mid B_K-(2^K-1+a-v|_K), \Phiverlet[h][a](q_0, p_0))\eqsp .
    \end{align}
    Using $\iota(b)-\iota(a)=b-a$ and $\Phiverlet[h][\iota(a)](\Phiverlet[h][-(2^K-1-v|_K)](q_0, p_0))=\Phiverlet[h][a](q_0,p_0)$, we get
    $$\bar{\rmq}_h(a,b\mid \msi,q_0,p_0)= \rmq_h(\iota(b)-\iota(a)  \mid B_K-\iota(a), \Phiverlet[h][\iota(a)](\Phiverlet[h][-(2^K-1-v|_K)](q_0, p_0))) $$
    $$=\bar{\rmq}_h(\iota(a),\iota(b) \mid B_K,\Phiverlet[h][-(2^K-1-v|_K)](q_0, p_0))\eqsp , $$
    which completes the proof.

    \subsection{Proof of \Cref{lemma:definition_continuity_q}}
    \label{sec_supp:definition_continuity_q}
%
    Let $a,b\in B_K$ and set $(q_0',p_0') = \Phiverlet[h][-(2^K-1-v|_K)](q_0, p_0)$.
    We have  by definitions \eqref{eq:def_hatPi}, \eqref{eq:transition_prob_q} that
    $$\hat{\rmq}_{h,K}(a,b)=\bar{\rmq}_h(a,b|B_K,q_0',p_0')=\rmq_h(b-a|B_K-a,\Phiverlet[h][a](q_0',p_0')) \eqsp ,$$
    where $\rmq_h(b-a|B_K-a,\Phiverlet[h][a](q_0',p_0'))$ is the conditional probability that Algorithm \ref{alg:nuts-doubling} 
     returns 
     $$\Phiverlet[h][b-a](\Phiverlet[h][a](q_0',p_0'))=\Phiverlet[h][b](q_0',p_0')$$
      when the initial state is $\Phiverlet[h][a](q_0',p_0')$,
      given that \Cref{alg:nuts-doubling} outputs $I_f=B_K-a=I_K $ with $K_f=K$ (recall that $I_f$ and $K_f$ are defined in \eqref{eq:nuts-I}).
    In the rest of the proof we work on the event $\{I_f=B_K-a\}$ which is contained in the event $ \{K_f=K\}$.
    
     To manipulate more intuitively the definition of $\rmq_h$ in \eqref{def:recur_def_rmq},
      we use the variables defined in \Cref{alg:nuts-samplingj} when \Cref{alg:nuts-doubling} is starting from $\Phiverlet[h][a](q_0',p_0') $ and using the sequences of \iid~random variables $(V_k)_{k =0}^{\Kmax}$, $(\bar{U}_{k})_{k=0}^{\Kmax}$ and $(\tilde{U}_{k})_{k=0}^{\Kmax}$ with distribution
                $\mathrm{Ber}(1/2)$, $\mathrm{Unif}([0,1])$ and $\mathrm{Unif}([0,1])$  respectively.
    We denote by $(I_k)_{k=0}^{K}$,$(I_k^{\text{new}})_{k=0}^{K-1}$ the sequences of random intervals defined in \Cref{alg:nuts-doubling},
     by $i'=(i_k')_{k=1}^{K}$ the sequence of multinomial random variables, by $j=(j_k)_{k=0}^{K}$ the random indices output by \Cref{alg:nuts-samplingj} in \Cref{alg:nuts-doubling},
     and $(\bar{V}_k)_{k=0}^{K}$ the sequence of binomial random variable ruling the acceptation/rejection mechanism in \Cref{alg:nuts-samplingj}.
    With these notations we have,
    \begin{equation}
        \hat{\rmq}_{h,K}(a,b)=\mathbb{P}(j_f=b-a|I_f=B_K-a)=\mathbb{P}(j_K=b-a|I_f=B_K-a) \eqsp.
      \end{equation}
    We distinguish the two cases, $b=a$ and $b \neq a$.
    We first state a technical result which is at the core of the proof.
    \begin{lemma}
      \label{lemma:pivot_defenition_q}
    On the event $\{I_f=B_K-a\}$, the sets $I_k,I_k^{\text{new}}$ are fixed, for any $k\in[0:K-1]$,
    \begin{equation}
    \label{eq:I_cond1a}
    I_k+a= \{c \in B_K:\ee c|^{K-k}=a|^{K-k-1}a_{k} \}=B_k+a|^{K-k}\ee 
    \end{equation}
    \begin{equation}
    \label{eq:I_cond2a}
    \text{and}
    \ee I_k^{\text{new}}+a= \{c \in B_K:\ee c|^{K-k}=a|^{K-k-1}a_{k}^c \} \eqsp.
    \end{equation}
    In addition, for any $k\in[0:K-1]$,
     \begin{equation}
      \label{eq:Ik_Bk_ak}
      I_k=B_k-a|_k \eqsp ,
     \end{equation}
     and,
     \begin{equation}
      \label{eq:pitildegoal}
   \tpi_{\Phiverlet[h][a](q_0',p_0')}(I_k)=\check{\pi}_{K-k}(a|^{K-k-1}a_{k})\ee, \ee
  \tpi_{\Phiverlet[h][a](q_0',p_0')}(I_k^{\text{new}})=\check{\pi}_{K-k}(a|^{K-k-1}a_{k}^c)\eqsp   .
  \end{equation}
  \end{lemma}

  \begin{proof}
    First we remark that \eqref{eq:I_cond1a} implies \eqref{eq:I_cond2a}. 
    Indeed if \eqref{eq:I_cond1a} holds, we have for any $k\in [0:K-1]$
    \begin{align}
    I_k^{\text{new}}+a&=(I_{k+1}+a)\setminus (I_k+a)\\
    &= \{c \in B_K:\ee c|^{K-k-1}=a|^{K-k-2}a_{k+1},\ee c|^{K-k}\neq a|^{K-k-1}a_{k}\} \\
    &= \{c \in B_K:\ee c|^{K-k-1}=a|^{K-k-1},\ee c|^{K-k}\neq a|^{K-k-1}a_{k} \} \\
    &= \{c \in B_K:\ee c|^{K-k-1}=a|^{K-k-1},\ee a|^{K-k-1}c_k \neq a|^{K-k-1}a_{k} \} \\
    &= \{c \in B_K:\ee c|^{K-k-1}=a|^{K-k-1},\ee c_k =a_{k}^c \} \\
    &=\{c \in B_K:\ee c|^{K-k}=a|^{K-k-1}a_{k}^c \} \eqsp .
    \end{align}
    Now, we show \eqref{eq:I_cond1a}. We first specify the value of $I_k$ for $k\in [0:K-1]$.
    Proceeding similarly to the proof of \Cref{lemma:comeback}, there exists $v\in B_K$ such that 
    $B_K(v)=B_K-a $, and therefore that
     \begin{equation}
      \label{eq:I_rep_v}
      I_k=B_k(v|_k)=B_k-(2^k-1-v|_k)
     \end{equation}
       for $k\in [0:K]$.
        We now express $v|_k$ with respect to $a$. 
    Using \eqref{eq:I_rep_v} and $B_K(v)=B_K-a$, we have $2^K-1-v=a$.
     From this equality, we deduce that $v_i=a_i^c$ for $i\in [0:K-1]$, $2^k-1-v|_k=a|_k$ and $B_k(v|_k)=B_k-(2^k-1-v|_k)=B_k-a|_k$ for any $k\in[K]$.
    This proves \eqref{eq:Ik_Bk_ak}.
    
    Then, for any $k\in [0:K-1]$ using \eqref{eq:I_rep_v}, $a+I_k=B_k-a|_k+a=B_k+a|^{K-k}$ and
    \begin{equation}
      \label{eq:2}
      \{c \in B_K:\ee c|^{K-k}=a|^{K-k-1}a_{k} \}=\{c+a|^{K-k}: c\in B_k \}  =B_k+a|^{K-k}=a+I_k \eqsp ,
    \end{equation}
    which completes the proof of \eqref{eq:I_cond1a}.
    We have $\tpi_{\Phiverlet[h][a](q_0',p_0')}(\msj)=\tpi_{q_0',p_0'}(a+\msj)$ for any $\msj\in \mathcal{P}(\Zset): |\msj|<\infty$, therefore \eqref{eq:pitildegoal}
    follows by definition of $\tpi$ \eqref{eq:def_check_omega} and \eqref{eq:I_cond1a}, \eqref{eq:I_cond2a}.
  \end{proof}

  \textbf{If $b=a$.} Then, by construction in \Cref{alg:nuts-doubling,alg:nuts-uturn},
    
  \begin{align}
   \hat{\rmq}_{h,K}(a,a)&=\mathbb{P}(\bar{V}_0=0,\ldots,\bar{V}_{K-1}=0|I_f=B_K-a)= \prod_{k=0}^{K-1}\mathbb{P}(\bar{V}_k=0|I_f=B_K-a)
   \\
  &=\prod_{k=0}^{K-1}(1-1\wedge [\tpi_{\Phiverlet[h][a](q_0',p_0')}(I_k^{\text{new}})/\tpi_{\Phiverlet[h][a](q_0',p_0')}(I_k)])\eqsp ,
  \label{eq:bequala_mainexpression}
\end{align}
  the second equality holds because of the mutual independence of $(\bar{V}_i)_{i=0}^{K-1}$ given that $I_f=B_K-a$ and the third equality holds on the event $\{I_f=B_K-a\}$ where the sets $I_k,I_k^{\text{new}}$ are fixed.
    Finally, we have by setting $k\gets K-k-1 $ and by using \eqref{eq:pitildegoal}, \eqref{eq:Omega}, \eqref{eq:bequala_mainexpression}:
    \begin{align}
        \hat{\rmq}_{h,K}(a,a)&=\bar{\rmq}_{h,k=0}^{K-1}(1-1\wedge [\check{\pi}_{K-k}(a|^{K-k-1}a_{k}^c)/\check{\pi}_{K-k}(a|^{K-k-1}a_{k})]) \\
       &=\Pi(a,K)\eqsp .
    \end{align}
    
    \textbf{If $b\neq a$.}    
    We define $N_0=\max \{ k\in [0:K-1]: \bar{V}_k=1\}$, which exists since $a\neq b$, $N_0$ is a random variable
    well defined on the event $\cup_{k\in [0:k-1]} \{\bar{V}_k=1\} $, that is why, we work on the event  $\cup_{k\in [0:k-1]} \{\bar{V}_k=1\} \cap \{N_0=n_0\}$.
    On this event, we have,
    \begin{equation}
      \label{eq:jk_in0}
      j_K=i'_{n_0+1}\in  I^{\text{new}}_{n_0}\eqsp .
    \end{equation}
    Denoting by $n=K-1-n_0 $, we have
    \begin{align}
        \hat{\rmq}_{h,K}(a,b) 
        & = 
        \mathbb{P}(i'_{n_0+1}=b-a,\bar{V}_{n_0}=1,\bar{V}_{n_0+1}=0,\ldots,\bar{V}_{K-1}=0|I_f=B_K-a)
        \\
        & =
        \begin{multlined}[t]
           \mathbb{P}(\bar{V}_{n_0+1}=0,\ldots,\bar{V}_{K-1}=0|I_f=B_K-a) 
            \\
            \times \mathbb{P}(\bar{V}_{n_0}=1|I_f=B_K-a)\mathbb{P}(i'_{n_0+1}=b-a|I_f=B_K-a)
        \end{multlined}
        \\
        & = 
        \Pi(a,n)\left( 1 \wedge \frac{\check{\pi}_{n+1}(a|^n a_{K-n-1}^c)}{\check{\pi}_{n+1}(a|^n a_{K-n-1})} \right) 
    \frac{\hat{\pi}_K(b)}{\check{\pi}_{n+1}( a|^na_{K-n-1}^c)}\eqsp ,
    \end{align}
    where in the last equality we use \eqref{eq:pitildegoal}, \Cref{lemma:pivot_defenition_q},
    the second inequality holds because of the mutual independance of $(\bar{V}_k)_{k=0}^{K-1}\cup(i_k')_{k=1}^{K} $ given that $I_f=B_K-a$ and the following Lemma.
    \begin{lemma}
      \begin{equation}
      \mathbb{P}(i'_{n_0+1}=b-a|I_f=B_K-a)=\frac{\hat{\pi}_K(b)}{\check{\pi}_{n+1}( a|^na_{K-n-1}^c)}
      \end{equation}
    \end{lemma}
    \begin{proof}
      By definition of $i'_{n_0+1} $ in \Cref{alg:nuts-samplingj}, it is sampled according to the multinomial on $I_{n_0}^{\text{new}}$ with the weights 
      \begin{equation}
        (\frac{\tilde{\pi}(\Phiverlet[h][i](\Phiverlet[h][a](q_0',p_0'))}{\tpi_{\Phiverlet[h][a](q_0',p_0')}(I_{n_0}^{\text{new}})})_{i\in I_{n_0}^{\text{new}}}
      \end{equation}
       when \Cref{alg:nuts-doubling} is initialized with $\Phiverlet[h][a](q_0',p_0')$.
       On the event $\{I_f=B_K-a\}$, 
      we have,
      \begin{equation}
        \mathbb{P}(i'_{n_0+1}=b-a|I_f=B_K-a)=\frac{\tilde{\pi}(\Phiverlet[h][b-a](\Phiverlet[h][a](q_0',p_0'))}{\tpi_{\Phiverlet[h][a](q_0',p_0')}(I_{n_0}^{\text{new}})}\eqsp .
      \end{equation}
        By using \eqref{eq:pitildegoal} and the definition of $n$,
       \begin{equation}
        \tpi_{\Phiverlet[h][a](q_0',p_0')}(I_{n_0}^{\text{new}})= \check{\pi}_{K-n_0}(a|^{K-n_0-1}a_{n_0}^c)=\check{\pi}_{n+1}(a|^{n}a_{K-n-1}^c) \eqsp ,
       \end{equation}
       then using $\hat{\pi}_K(c)=\tilde{\pi}(\Phiverlet[h][c](q_0',p_0'))$ for any $c\in B_K$ by definition of $\hat{\pi}_K$ \eqref{eq:def_check_omega},
       \begin{equation}
        \tilde{\pi}(\Phiverlet[h][b-a](\Phiverlet[h][a](q_0',p_0'))=\tilde{\pi}(\Phiverlet[h][b](q_0',p_0'))=\hat{\pi}_K(b)\eqsp, 
       \end{equation}
       which completes the proof.
       
    \end{proof}
    
    Now, setting $n'=\max \{i\in [K]:\ee a|^i= b|^i\}\cup \{0\}$, it remains to show that $n=n'$. 
    Setting $n_0'=K-1-n'$, we show that $n_0'=n_0$ by contradiction.
    First note that\footnote{using the convention $\sum_{0}^{-1}=0$ and $a|_0=0$}
    \begin{equation}
      \label{eq:n0_prime_b_minus_a}
    b-a=\sum_{k=0}^{K-n'-2} (b_k-a_k)2^k+\epsilon 2^{K-n'-1}= b|_{n_0'}-a|_{n_0'}+\epsilon 2^{n_0'}\eqsp ,
    \end{equation}
     with $\epsilon=b_{K-n'-1}- a_{K-n'-1}=a_{n_0'}^c-a_{n_0'}\in\{-1,1\}$, by maximality of $n'$.
    In addition, we have $ i_{n_0+1}'\in I^{\text{new}}_{n_0}$ by definition of the random variable, and using \eqref{eq:Ik_Bk_ak}, we have $I^{\text{new}}_{n_0}=(B_{n_0+1}-a|_{n_0+1})\setminus (B_{n_0}-a|_{n_0}) $. 
    It implies that
    \begin{equation}
      \label{eq:cases_an0}
      I^{\text{new}}_{n_0}=
      \begin{cases}
        [\max (B_{n_0}-a|_{n_0}):\max (B_{n_0}-a|_{n_0})+2^{n_0}] & a_{n_0}=0 \\
        [\min (B_{n_0}-a|_{n_0})-2^{n_0}:\min (B_{n_0}-a|_{n_0})] & a_{n_0}=1 \eqsp .
      \end{cases}
    \end{equation}
    We complete the proof by assuming $a_{n_0}=0$, the case $a_{n_0}=1$ being similar.
    On the event $\{i_{n_0+1}'=b-a\}$, by \eqref{eq:jk_in0} and \eqref{eq:cases_an0}, we have,
    \begin{equation}
        \label{eq:second_bma}
    2^{n_0}-1-a|_{n_0}< b-a\leq 2^{n_0+1}-1-a|_{n_0+1} \eqsp .
    \end{equation}
    On the other hand, bounding from below and above $b-a$ with its expression \eqref{eq:n0_prime_b_minus_a} by using that $ b|_{n_0'}\in B_{n_0'}$ by definition,
    \begin{equation}
    \label{eq:first_bma}
    \epsilon 2^{n_0'}-a|_{n_0'}\leq b-a\leq 2^{n_0'}(1+\epsilon)-1-a|_{n_0'} \eqsp .
    \end{equation}
     \eqref{eq:first_bma} and \eqref{eq:second_bma} implies that $\epsilon=1$ ($\epsilon=-1$ if $a_{n_0}=1$).
      Then, from \eqref{eq:first_bma} and \eqref{eq:second_bma}, we deduce by contradiction that $n_0'= n_0 $.
    
      The fact that $q_0,p_0 \to \rmq_h(a\mid \msi,q_0,p_0) $ is continuous comes from the fact that $q,p\to \tpi(\Phiverlet[h][k](q,p))$ is continuous for any $k\in \Zset$.
    

  \section{Proof of \Cref{lemma:irreducibility_q}}
    \label{sec_supp:irreducibility_rmq}
    \indent
\vspace{5mm}

\Cref{lemma:irreducibility_q} is a consequence of the following technical result. 
\begin{lemma}\label{lemma:main_irreducibility_q}
   Assume \Cref{hyp:regularity}.
\begin{enumerate}[wide, labelwidth=!, labelindent=0pt, label=(\alph*),leftmargin=*] 
\item \label{lemma:item2_technical_core_irrq} 
      Let $a,b \in B_K$. $\hat{\rmq}_{h,K}(a,b)>0 $ if and only if  $\Pi(a,n)>0 $ with $n=\max[\{i\in [K]:\ee a|^i= b|^i\}\cup \{0\}]$.
    Moreover, for any $ k\in[K]$, $\Pi(a,k)=\Pi(c,k)$ for any $c \in B_K$ such that $c|^k=a|^k $.
     Consequently, $\Pi(a,n)=\Pi(b,n)$ and thus $\hat{\rmq}_{h,K}(a,b)\hat{\rmq}_{h,K}(b,a)>0 $ if and only if $\Pi(a,n)>0 $.
   \item     \label{lemma:item1_first_pivot}
   Suppose in addition that for any $n\in[K]$ and $a,b\in B_K$ such that $a|^n\neq b|^n$,  we have $\check{\pi}_n(a)\neq \check{\pi}_n(b) $. Then, there exists $\pivot\in B_K$ such that for any $a$ in $B_K$
  $$\hat{\rmq}_{h,K}(a,\pivot)\hat{\rmq}_{h,K}(\pivot,a)>0 \eqsp .$$ 
\end{enumerate}
\end{lemma}

\begin{proof}
\ref{lemma:item2_technical_core_irrq}     If $n=K$,
    by \Cref{lemma:definition_continuity_q},
$      \hat{\rmq}_{h,K}(a,b)=\Pi(a,K)$,  
    the main claim is then straightforward. 

    If $n<K$, by \Cref{lemma:definition_continuity_q},
    \begin{equation}
      \hat{\rmq}_{h,K}(a,b)=\Pi(a,n) \Big( 1 \wedge \frac{\check{\pi}_{n+1}(a|^n (a_{K-l-1})^c)}{\check{\pi}_{n+1}(a|^na_{K-n-1})} \Big) \frac{\hat{\pi}_K(b)}{\check{\pi}_{n+1}(a|^n(a_{K-n-1})^c)} \eqsp. 
    \end{equation}
    the main claim is then straightforward since by \Cref{hyp:regularity}, and the definition of $\check{\pi}_{n+1},\hat{\pi}_K$ in \eqref{eq:def_check_omega}-\eqref{eq:hat_omega} we have,
    \begin{equation}
      \Big( 1 \wedge \frac{\check{\pi}_{n+1}(a|^n (a_{K-l-1})^c)}{\check{\pi}_{n+1}(a|^nc_{K-n-1})} \Big) \frac{\hat{\pi}_K(\xi)}{\check{\pi}_{n+1}(a|^n(a_{K-n-1})^c)}>0 \eqsp .
    \end{equation}
    The second statement is a consequence of the fact that $\Pi(a,n)$ only  depends on $a|^n$ (see \eqref{eq:Omega}) and the third statement is straightforward given the two first statements.

    \ref{lemma:item1_first_pivot}   We define $\pivot=\pivot_{K-1}\ldots \pivot_0\in B_K$ such that $\pivot_{K-1}=\argmax_{s\in\{0,1\}} \check{\pi}_1(s) $ and for $n\in [K-1]$, 
  $\pivot_n=\argmax_{s\in\{0,1\}} \check{\pi}_{n+1}(\pivot_{K-1}\ldots \pivot_{K-n} s)$ (the maximum is unique by hypothesis). 
   Let $a\in B_K$ and define $n=\max \left[ \{i\in [K]:\ee a|^i= \pivot|^i\}\cup \{0\}\right]$.
    Applying \Cref{lemma:main_irreducibility_q}-\ref{lemma:item2_technical_core_irrq}, we have $ \hat{\rmq}_{h,K}(a,\pivot)\hat{\rmq}_{h,K}(\pivot,a)>0$ if $\Pi(\pivot,n)>0$
    which is the case
   since by construction  $ \check{\pi}_{i+1}(\pivot|^i\pivot_{K-i-1}^c)<\check{\pi}_{i+1}(\pivot|^i\pivot_{K-i-1})$ for any $i\in [K-1]$. 
\end{proof}

  \begin{proof}[Proof of \Cref{lemma:irreducibility_q}]
 \indent

   Let $a,b\in B_K$, we first show \ref{lemma:item2_irreducibility_q} and assume that for any $n\in[K]$ and $c,e\in B_K$ such that $c|^n\neq e|^n$, we have $\check{\pi}_n(c)\neq \check{\pi}_n(e) $. Then,
 \Cref{lemma:main_irreducibility_q}-\ref{lemma:item1_first_pivot} applies,  and
  it exists $\pivot\in B_K$ such that $\hat{\rmq}_{h,K}(\pivot, \pivot)\hat{\rmq}_{h,K}(\pivot,a)\hat{\rmq}_{h,K}(\pivot,a)>0 $ .
   
  Then, for any $j\in \Nset$,
  $$\ee \hat{\rmq}_{h,K}^{j+2}(a,b)\geq\hat{\rmq}_{h,K}(a,\pivot) \hat{\rmq}_{h,K}^{j}(\pivot,\pivot)\hat{\rmq}_{h,K}(\pivot,b)>0 \eqsp .$$
    
  We now show \ref{lemma:item1_irreducibility_q}. Note that we can only consider the case where there exist $n\in [K]$ and $ c, e\in  B_K^2$ such that $c|^n\neq e|^n$ and $\check{\pi}_n(c)= \check{\pi}_n(e) $.

    We choose $\pivot_{K-1}\in \argmax_{s\in\{0,1\}} \check{\pi}_1(s)$ and by recursion for any $k\in [K-1]$ , 
    \begin{equation}
      \label{eq:pivot_choose}
        \pivot_{K-1-k} \in \argmax_{s\in\{0,1\}} \check{\pi}_{k+1}(\pivot_{K-1}\ldots \pivot_{K-k}s)\subset \{0,1\} \eqsp .
    \end{equation}
   
    Then, we define 
    \begin{equation}
      \msm_{\check{\pi}}=\{ i\in [K-1] : \check{\pi}_{i+1}(\pivot_{K-1}\ldots \pivot_{K-i} 0)=\check{\pi}_{i+1}(\pivot_{K-1}\ldots \pivot_{K-i} 1) \}\cup\{0: \check{\pi}_1(1)=\check{\pi}_1(0) \}\eqsp.
    \end{equation}
    If $\msm_{\check{\pi}} $ is empty, in \eqref{eq:pivot_choose} the argmax set is always a singleton and thus $\pivot=\pivot_{K-1}\ldots \pivot_{0}$ is uniquely determined, we come back to the previous case.
    Otherwise, $\msm_{\check{\pi}} $ is not empty and we define $l=\min \msm_{\check{\pi}}  $ such that $\pivot_{K-1-n} $ is uniquely determined for any $i\in [0:l-1] $.
    Then, we define two pivots, $\pivot^0:=\pivot_{K-1}\ldots \pivot_{K-l} \underbrace{0\ldots 0}_{K-l} \in B_K$\footnote{If $l=0$, we define $\pivot^0=0\ldots 0\in B_K$ and $\pivot^1=1\ldots 1\in B_K$ .} and 
   $\pivot^1:=\pivot_{K-1}\ldots \pivot_{K-l} \underbrace{1\ldots 1}_{K-l} \in B_K$. In particular with this construction, we have for any $k\in [0:l]$ and any $m\in\{0,1\}$,
   \begin{equation}
    \label{eq:pivot_positivity}
      \Pi(\pivot^m,k)>0\eqsp .
   \end{equation}


    We consider two cases, namely $\pivot^0|^{l}=b|^{l} $ or $\pivot^0|^{l}\neq b|^{l} $ and then sub-cases, namely $\pivot^0|^{l}=a|^{l} $ or $\pivot^0|^{l}\neq a|^{l} $.
    To address it, we establish the following lemma.
    \begin{lemma}
      \label{lemma:2technical_irrq}
      For any $\xi \in B_K$ and $m_0\in\{0,1\}$ such that  $\pivot^{m_0}|^{l}\neq \xi|^{l} $, then \begin{equation}
        \hat{\rmq}_{h,K}(\xi,\pivot^{m_0})\hat{\rmq}_{h,K}(\pivot^{m_0},\xi)>0 \eqsp.
      \end{equation}
      For any $\xi \in B_K$ such that  $\pivot^0|^{l}= \xi|^{l} $, denoting by $m=\xi_{K-l-1}^c \in \{0,1\}$, then 
      \begin{equation}
        \hat{\rmq}_{h,K}(\xi,\pivot^m)\hat{\rmq}_{h,K}(\pivot^m,\xi)>0 \eqsp .
      \end{equation}
    \end{lemma}
    \begin{proof}
      Let $\xi \in B_K$ and $m_0\in\{0,1\}$ such that  $\pivot^{m_0}|^{l}\neq \xi|^{l} $. Define 
      $$n=\max \left[\{i\in [K]:\ee \xi|^i= \pivot^{m_0}|^i\}\cup \{0\}\right]\eqsp .$$
       Then, $\Pi(\xi,n)=\Pi(\pivot^{m_0},n)$ since $\xi|^n= \pivot^{m_0}|^n$ and $n<l$ by  \Cref{rm:growthmaximality}.
        Moreover, we deduce that $\Pi(\pivot^{m_0},n)>0$ by \eqref{eq:pivot_positivity} since $n<l$,
       and therefore using \Cref{lemma:main_irreducibility_q}-\ref{lemma:item2_technical_core_irrq}, we have $\hat{\rmq}_{h,K}(\pivot^{m_0},\xi)>0 $ and $\hat{\rmq}_{h,K}(\xi,\pivot^{m_0})>0 $.

       Let $\xi \in B_K$ such that  $\pivot^0|^{l}= \xi|^{l} $. Define $n_m=\max \left[\{i\in [K]:\ee \xi|^i= \pivot^m|^i\}\cup \{0\}\right]$. Distinguishing the case $\xi_{K-l-1} = 0$ or $\xi_{K-l-1} =1$, we obtain that $n_m=l$.
       Then, $\Pi(\xi,l)=\Pi(\pivot^0,l)$ since $\xi|^l= \pivot^0|^l$ and $\Pi(\pivot^0,l)=\Pi(\pivot^m,l)>0$ by \eqref{eq:pivot_positivity}.
       Using \Cref{lemma:main_irreducibility_q}-\ref{lemma:item2_technical_core_irrq}, we have $\hat{\rmq}_{h,K}(\xi,\pivot^m)>0 $ and $\hat{\rmq}_{h,K}(\pivot^{m_0},\xi)>0 $.
    \end{proof}

    We now distinguish four cases.
    \begin{itemize}
      \item  If $\pivot^0|^{l}=b|^{l} $. Then, denoting by $m=b_{K-l-1}^c\in \{0,1\}$, the second statement of \Cref{lemma:2technical_irrq} implies that $\hat{\rmq}_{h,K}(\pivot^m,b)>0$.
   
   \item If $\pivot^0|^{l}=b|^{l} $ and $\pivot^0|^{l}\neq a|^{l} $.
     Then,  denoting by  $m=b_{K-l-1}^c\in \{0,1\}$, the first statement of \Cref{lemma:2technical_irrq} implies that $\hat{\rmq}_{h,K}(a,\pivot^m)>0$ and therefore 
     $\hat{\rmq}_{h,K}^2(a,b)\geq \hat{\rmq}_{h,K}(a,\pivot^m)\hat{\rmq}_{h,K}(\pivot^m,b)>0 $.
   
    \item If $\pivot^0|^{l}=b|^{l} $ and $\pivot^0|^{l}= a|^{l} $. 
       Then,  denoting by  $m=b_{K-l-1}^c\in \{0,1\}$, there exists $j_0\in \{1,2\}$ such that $\hat{\rmq}_{h,K}^{j_0}(a,b)>0$.
       Indeed,
    if $a_{K-l-1}=b_{K-l-1}^c$, denoting by 
    \begin{equation}n_{a,b}=\max \left[\{i\in [K]:\ee a|^i= b|^i\}\cup \{0\}\right]\eqsp ,\end{equation}
    we have $n_{a,b}=l $ since $a|^{l} =\pivot^0|^{l} =b|^{l}$. Then, we have
    $\Pi(a,l)=\Pi(\pivot^m,l)>0$ since $a|^{l}= \pivot^m|^{l}$ and \eqref{eq:pivot_positivity}.
     Therefore $\hat{\rmq}_{h,K}(a,b)>0$ by applying \Cref{lemma:main_irreducibility_q}-\ref{lemma:item2_technical_core_irrq}, which gives the result with $j_0=1$. In the case $a_{K-l-1}= b_{K-l-1}$, we apply the second statement of \Cref{lemma:2technical_irrq} which gives $\hat{\rmq}_{h,K}(a,\pivot^m)>0$.
     Then, $\hat{\rmq}_{h,K}^2(a,b)=\hat{\rmq}_{h,K}(\pivot^m,b)\hat{\rmq}_{h,K}(a,\pivot^m)>0$ , which gives the result with $j_0=2$.
     
     \item 
      If $b|^l \neq \pivot^0|^l$ and $a|^l= \pivot^0|^l$, then we have,
        $\hat{\rmq}_{h,K}^2(a,b)>0$ .
        Indeed,
          the second statement of \Cref{lemma:2technical_irrq} gives $\hat{\rmq}_{h,K}(a,\pivot^m)>0$ with $m=a_{K-l-1}^c$. 
          Then, the first statement of \Cref{lemma:2technical_irrq} gives $\hat{\rmq}_{h,K}(b,\pivot^m)>0$ since $\pivot^0|^l=\pivot^m|^l$.
          Therefore,
          $$ \hat{\rmq}_{h,K}^2(a,b)=\hat{\rmq}_{h,K}(a,\pivot^m)\hat{\rmq}_{h,K}(\pivot^m,b)>0 \eqsp .$$

  \item If $b|^l \neq \pivot^0|^l$ and $a|^l\neq \pivot^0|^l$, then we have $\hat{\rmq}_{h,K}^2(a,b)=\hat{\rmq}_{h,K}(b,\pivot^0)\hat{\rmq}_{h,K}(\pivot^0,b)>0 $.       
    Indeed, we apply two times the first statement of \Cref{lemma:2technical_irrq} with $m_0=0$.
    \end{itemize}
This completes the proof that it exists $j_0\in\{1,2\}$ such that $\hat{\rmq}_{h,K}^{j_0}(a,b)>0$.
  \end{proof}

\section{Proofs of \Cref{ergodicity_section}}
\label{sec_supp:proofs_ergo_NUTS}

The presented proof can be extended to different choices of stopping time and index selection kernel $\rmqq_h$, 
but it is out of the paper's scope and relegated to future work.
First we show a prelimenary Lemma used to prove that the stopping time is locally constant in \Cref{lemma:nuts-locally-constant-trajectories}.
We introduce the sets
\begin{equation}
\label{eq:top_scru}
\scrU^{(K)}(a)=\{(q_0,p_0) \in (\mathbb{R}^d)^2\,:\, a\in \scrU^{(K)}(q_0,p_0) \} 
\end{equation} 
for any $K\in [\Kmax]$ and $a\in B_K$ where $\scrU^{(K)}(q_0,p_0)$ for $q_0,p_0\in (\Rset^d)^2$ is defined in \eqref{eq:scrU}. 
\begin{lemma}
\label{lemma:scru_open}
Under \Cref{hyp:regularity},
for any $K,K'\in [\Kmax]$, $a\in B_K$ and $q_0,p_0 \in \Rset^d$, 
$\scrU^{(K)}(a)$ (defined in \eqref{eq:top_scru}) is open and
$\scrU^{(K')}(a|_{K'})\subset \scrU^{(K)}(a)$ when $K>K'$.
Moreover, under \Cref{hyp:3}$(h,\Kmax)$-\ref{hyp:item_i_topo}, for $h >0$ and $\Kmax \in \nset_{>0}$, for any $q\in \Rset^d$, the following set is dense, 
\begin{equation}
   \label{eq:scrU-0}
   \msu_{q,-0}=\{ p\in \Rset^d: (q,p)\notin \partial \scrU^{(K)}(a),\, K\in [\Kmax],\, a\in B_K \} \eqsp .
\end{equation}

\end{lemma}
\begin{proof}[Proof of \Cref{lemma:scru_open}]
We express $\scrU^{(K)}(a)$ for any $K\in [\Kmax]$ and $a\in B_K$ as an union of open sets.
For any $K\in [\Kmax]$, $k\in [K-1]$, $l\in [2^{K-k}]$ and for any $a\in B_K$, we define the sets
\begin{align}
    \scrU_{k,l,+}^{(K)}(a) &=
    \big\{ (q_0, p_0) \in (\Rset^d)^2: p_{\ee i_{+}(K,k,l,a)}^\top (q_{i_{+}(K,k,l,a)} - q_{i_{-}(K,k,l,a)}) < 0 \big\} \eqsp,
    \\
    \scrU_{k,l,-}^{(K)}(a) &=
    \big\{ (q_0, p_0) \in (\Rset^d)^2:p_{i_{-}(K,k,l,a)}^\top (q_{\ee i_{+}(K,k,l,a)} - q_{i_{-}(K,k,l,a)}) < 0 \big\}\eqsp,
    \end{align}
    where
\begin{equation}
i_{-}(K,k,l,a)=-T_-^{(K)}(a)+(l-1) 2^k,\qquad i_{+}(K,k,l,a)=-T_-^{(K)}(a) + l 2^k-1\eqsp ,
\end{equation}
and we have used the convention $  (q_k, p_k) = \Phiverlet[h][k](q_0, p_0)$ for any $q_0,p_0\in (\Rset^d)^2$ and $T_-^{(K)}$ is defined in \eqref{eq:T_-}.
Further we have by definitions \eqref{eq:top_scru} and \eqref{eq:scrU},
\begin{equation}
\label{eq:scrUforfrontier}
\scrU_k^{(K)}(a) = \bigcup_{l=1}^{2^{K-k}} \big( \scrU_{k,l,+}^{(K)}(a) \cup \scrU_{k,l,-}^{(K)}(a) \big)
\quad \text{where} \quad
\scrU^{(K)}(a) = \bigcup_{k=1}^{K-1} \scrU_k^{(K)}(a)\eqsp .
\end{equation}
Thus $\scrU^{(K)}(a)$ with $a\in B_K$ is the set of initial positions and momentums for which a 
U-turn occurs on some pair of indices at stage $K$ of the trajectory 
construction, when the Bernoulli variables (which determine the randomness 
of the construction) take the values $a = (a_k)_{k=0}^{K-1}$.
With these previous expressions, for all $K\in [\Kmax]$, the fact that $\scrU^{(K)}(a)$ is open follows 
from the continuity of the maps $(q_0, p_0)\in (\Rset^d)^2 \mapsto \Phiverlet[h][k](q_0, p_0)$ for any $k \in [\Kmax]$
and the definition of the $\scrU_{k,l,\pm}^{(K)}(a)$ as preimages of open sets\footnote{These properties are unchanged if additional U-turn checks are added to $\scrU^{(K)}(a)$.}.
The fact that $\scrU^{(K')}(a|_{K'}) \subset \scrU^{(K)}(a)$ for all $K > K'$ follows directly from construction.

Let $K\in [\Kmax]$ and $a\in B_K$, from \eqref{eq:scrUforfrontier} and $\partial \scrU_k^{(K)}(a) \subset  \bigcup_{l=1}^{2^{K-k}} \big(\partial \scrU_{k,l,+}^{(K)}(a) \cup \partial \scrU_{k,l,-}^{(K)}(a) \big) $,
we deduce that $(q,p) \in \partial \scrU^{(K)}(a)$ implies there exist $T_1,T_2 \in [-2^{\Kmax}+1:2^{\Kmax}-1]^2$ with $T_1\neq T_2$ such that $F^{T_1,T_2}_q(p)= 0$
and thus $\msf_{q,-0}\subset \msu_{q,-0}$. Therefore, for any $q\in \Rset^d$, $\msu_{q,-0}$ is dense under \Cref{hyp:3}$(h,\Kmax)$-\ref{hyp:item_i_topo}.

\end{proof}

In the following, we will frequently consider a fixed binary sequence $a = (a_k)_{k=0}^{\Kmax-1}\in B_{\Kmax}$, to this end, for any $(q, p) \in (\mathbb{R}^d)^2$ we set
\begin{equation}
\Sfun_a(q, p) = \Sfun(a,q,p),\ee  K_{f,a}(q,p)=(\Sfun_a(q, p)-1) \wedge \Kmax \eqsp ,
\end{equation}
where $\Sfun(a,q,p)$ is defined in \eqref{eq:nuts-stopping},
so that the index set constructed by the NUTS algorithm 
starting at $(q, p) \in (\Rset^d)^2$ is $B_{K_{f,a}(q,p)}(a|_{K_{f,a}(q,p)})$ when the Bernoulli variables $(V_k)_{k=0}^{\Kmax-1}$ take the values $a=(a_k)_{k=0}^{\Kmax-1}$.

\begin{lemma}
\label{lemma:nuts-locally-constant-trajectories}

Assume \Cref{hyp:regularity}.
For any $a \in B_{\Kmax}$ there exists a dense open set $\mathrm{G}_a \subset (\mathbb{R}^d)^2$ such that $(q, p) \mapsto K_{f, a}(q, p)$ is locally constant on $\mathrm{G}_a$, i.e., for any $(q_0, p_0) \in \mathrm{G}_a$ there exists an open neighborhood $\msw \subset \mathrm{G}_a$ containing $(q_0, p_0)$ such that $K_{f, a}$ is constant on $\msw$.
Moreover, under \Cref{hyp:3}$(h,\Kmax)$-\ref{hyp:item_i_topo}, for $h >0$ and $\Kmax \in \nset_{>0}$, for any $(q, p) \in (\Rset^d)^2$, $\epsilon > 0$, there exists $p_\epsilon \in \mathrm{B}(p, \epsilon)$ such that $(q, p_\epsilon) \in \mathrm{G}_a$.
\end{lemma}

\begin{proof}
Fix $a \in B_{\Kmax}$.
Let $\mathrm{G}_a$ denote the set of points $(q, p) \in (\mathbb{R}^d)^2$ that have an open neighborhood $\msw \subset (\mathbb{R}^d)^2$ such that $K_{f, a}$ is constant on $\msw$.
By definition, $\mathrm{G}_a$ is an open set.

In order to show that $\mathrm{G}_a$ is dense, let $(q, p) \in (\mathbb{R}^d)^2$ be an arbitrary point and $\msv \subset (\mathbb{R}^d)^2$ be an open neighborhood of $(q, p)$.
Let $K = \min_{(q', p') \in \msv} K_{f, a}(q', p')$.
If $K \geq \Kmax$, then $(S_a(q', p')-1) \wedge \Kmax = \Kmax$ for any $(q', p') \in \msv$.
If $K < \Kmax$, then by definition of $K$, the open set $\scrU^{(K+1)}(a) \cap \msv$ is non empty and $S_a(q', p') - 1 = K$ for any $(q', p') \in \scrU^{(K+1)}(a) \cap \msv$.
We have shown that any neighborhood of an arbitrary point $(q, p) \in (\mathbb{R}^d)^2$ contains a point $(q_0, p_0) \in \mathrm{G}_a$.



Now, we prove the final claim. Let $(q,p)\in (\Rset^d)^2$.
If $(q, p) \in \mathrm{G}_a$, the claim is clear since $\mathrm{G}_a$ is open.
Otherwise, $(q, p) \notin \mathrm{G}_a$, but we have $(q, p) \in \partial \mathrm{G}_a$ since $\mathrm{G}_a$ is dense,
and denote by
$$
   S' = \lim_{\epsilon \to 0^+} \min_{(q', p') \in \mathrm{B}((q, p), \epsilon)} S_a(q', p'),
$$
i.e., $S'$ is the smallest value of $S_a(q', p')$ that is obtained in every open neighborhood of $(q, p)$.

We prove by contradiction that $\dist((q, p), \partial \scrU^{(S')}(a))= 0$.
 Assume $\dist((q, p), \partial \scrU^{(S')}(a))=\dist_0>0$
By definition of $S'$,
 there exists $(q',p')\in \mathrm{B}((q,p),\dist_0/2) $ such that $S_a(q', p') = S'$.
  This proves $\dist((q, p), \scrU^{(S')}(a))< \dist((q, p), \partial \scrU^{(S')}(a))$ and thus $(q, p)\in \scrU^{(S')}(a)$ since $\scrU^{(S')}(a)$ is open.
   Moreover, $S'$ being the limit of integers, it exists $\epsilon>0$ such that $S'=\min_{(q', p') \in \mathrm{B}((q, p), \epsilon)} S_a(q', p')$.
   Since $(q, p)\in \scrU^{(S')}(a)$ which is open, it exists $\epsilon'<\epsilon$ such that $\mathrm{B}((q, p), \epsilon')\subset \scrU^{(S')}(a)$.
 By the inclusions $ \scrU^{(K)}(a) \subset \scrU^{(K')}(a)$ for $K < K'$, we deduce that $q',p'\in \mathrm{B}((q, p), \epsilon')\mapsto  S_a(q', p')$ takes the constant value $S'$.
 This implies $(q,p)\in \mathrm{G}_a $ which contradicts the assumption $(q, p) \in \partial \mathrm{G}_a$.

We thus have $\dist((q, p), \partial \scrU^{(S')}(a)) = 0$, which implies $(q, p) \in \partial \scrU^{(S')}(a)$ since the boundary of any set is closed.
It implies that
\begin{equation}
   \{p'\in\Rset^d: (q,p')\notin \mathrm{G}_a \}=\{p'\in\Rset^d: (q,p')\in \partial \mathrm{G}_a \}\subset \Rset^d \setminus \msu_{q,-0} 
\end{equation}
where $\msu_{q,-0}$ is defined in \eqref{eq:scrU-0}.
Using \Cref{lemma:scru_open} with \Cref{hyp:3}$(h,\Kmax)$-\ref{hyp:item_i_topo}, we have that $\msu_{q,-0}$ is dense and thus $\{p'\in\Rset^d: (q,p')\notin \mathrm{G}_a \} $
has its complementary $\{p'\in\Rset^d: (q,p')\in \mathrm{G}_a \} $ which is dense. This completes the proof.


\end{proof}

\begin{remark}
In slightly less precise terms, this lemma says that the integer-valued function $(q,p) \in (\Rset^d)^2\mapsto S_a(q,p)$ is continuous in $\mathrm{G}_a$ and that nearly every point is in $\mathrm{G}_a$ if we allow a little perturbation of the initial momentum $p$.
\end{remark}

\subsection{Proof of \Cref{thm:nuts-accessibility}}

Let $\mathsf{E} \subset \Rset^d$ be open, $q_0 \in \Rset^d$.
It suffices to prove that for any $q_1 \in \mathsf{E}$ and $M>0$ such that $\mathrm{B}(q_1,M)\subset \mse$, there exist, $m(q_0)>0$ and $r_{q_0}>0 $ such that for any $q\in \mathrm{B}(q_0, r_{q_0})$,
\begin{equation}
   \label{eq:access_to_prove}
   \KkerU_h(q, \mathrm{B}(q_1, M)) + (\KkerU_h)^2(q, \mathrm{B}(q_1, M)) \geq m(q_0) \eqsp .
\end{equation}
This is proved using  \Cref{lemma:rmq_positive} and \ref{lemma:rmq_0} in what
follows, which are both based on the next technical result.
\begin{lemma}
   \label{lem:1}
   Assume \Cref{hyp:regularity} and \Cref{hyp:3}$(h,\Kmax)$-\ref{hyp:item_i_topo} with $h>0$ and $\Kmax\in \Nset_{>0}$.
 Let $\mathsf{E} \subset \Rset^d$ be open, $q_0 \in \Rset^d$.
  Then, for any $q_1\in \mse$ and $M>0$ such that $\mathrm{B}(q_1,M)\subset  \mse $,
  there exist $p_0\in (\Rset^d)^2$ and $r_q,r_p >0$, $\msj \subset B_{\Kmax}$ such that $\{0,1\}\subset \msj $, $\msj=B_K$ for some $K\in \Nset^*$,
 $   \operatorname{proj}_1\Phiverlet[h][1](q_0, p_0)=q_1$,
 and for any $q, p\in \mathrm{B}(q_0, r_{q})\times \mathrm{B}(p_0, r_{p})$,
\begin{equation}
   \label{eq:acces_rmp_constant}
   \rmp_h\big(\msj \mid q, p\big)=\rmp_h(\msj |q_0, p_0)  > 0\eqsp ,\quad  \operatorname{proj}_1\Phiverlet[h][1](q, p)\in \mathrm{B}(q_1, M) \subset \mse \eqsp .
\end{equation}

\end{lemma}
\begin{proof}
   Let $q_1\in \mse$ and $M>0$ such that $\mathrm{B}(q_1,M)\subset  \mse $.
   By \Cref{rm:Tequal1}, for any $h> 0$,
   \begin{equation}
     \label{eq:def_psi_q_proof_1}
     \psi_{q_0}^{(1)}: p \in \Rset^d \mapsto \operatorname{proj}_1\Phiverlet[h][1](q_0, p)
   \end{equation}
   is an homeomorphism, and therefore there exist $p_0, p_1 \in (\mathbb{R}^d)^2$ such that $\Phiverlet[h][1](q_0, p_0) = (q_1, p_1)$. 

   Noticing that for any $q,p\in (\Rset^d)^2$, $\scrU^{(1)}(q,p)$ is empty by \eqref{eq:scrU}, then for any $a\in B_{\Kmax}$, $K_{f,a}(q,p)\geq 1$.
 By defining $a=2^{\Kmax}-1\in B_{\Kmax}$, we have for any $K\in[\Kmax], B_{K}(a|_K)=B_K $.

  Then $\{0,1\}\subset B_{K_{f,a}(q_0,p_0)}(a|_{K_{f,a}(q_0,p_0)})=\msj= B_{K_{f,a}(q_0,p_0)}$. We distinguish two cases.

  \textbf{If $(q_0,p_0) \in \mathrm{G}_a$.}  By \Cref{lemma:nuts-locally-constant-trajectories},  there exist $r_q', r_p' >0$ such that for any $q,p \in \mathrm{B}(q_0, r_{q}')\times \mathrm{B}(p_0, r_{p}')$,
   $K_{f,a}(q,p)=K_{f,a}(q_0,p_0)$ and thus for any $q,p \in \mathrm{B}(q_0, r_{q}')\times \mathrm{B}(p_0, r_{p}')$,
  \begin{equation}
   \label{eq:msj_constant}
   \msj=B_{K_{f,a}(q,p)}(a|_{K_{f,a}(q,p)})\eqsp .
  \end{equation}
     By \Cref{lemma:nuts-trajectory-probability}, \eqref{eq:msj_constant} implies for any $q, p \in \mathrm{B}(q_0, r_{q}')\times \mathrm{B}(p_0, r_{p}')$
 \begin{equation}
   \label{eq:rmp_constant_bis}
   \rmp_h\big(\msj \mid q, p\big)=\rmp_h(\msj |q_0, p_0) > 0\eqsp .
 \end{equation}

 Using the continuity of $q,p \in (\Rset^d)^2 \mapsto \operatorname{proj}_1\Phiverlet[h][1](q, p) $ and $\psi^{(1)}_{q_0}(p_0)=q_1 $, there exist $r_q, r_p\in (0,r_q')\times(0,r_p')$ such that for any $q, p\in \mathrm{B}(q_0, r_{q})\times \mathrm{B}(p_0, r_{p})$,
\begin{equation}
   \operatorname{proj}_1\Phiverlet[h][1](q, p)\in \mathrm{B}(q_1, M) \subset \mse \eqsp .
\end{equation}
This, \eqref{eq:msj_constant} and \eqref{eq:rmp_constant_bis} completes the proof.

\textbf{If $(q_0,p_0) \notin \mathrm{G}_a$.}
 We show that this reduces  to the previous case using the last statement of \Cref{lemma:nuts-locally-constant-trajectories}.
  By \Cref{lemma:nuts-locally-constant-trajectories}, for any $\epsilon>0$, there exists $p_\epsilon\in \mathrm{B}(p_0, \epsilon)$, such that $(q_0,p_\epsilon) \in \mathrm{G}_a$.
Then, there exist $r_q', r_p' >0$ such that for any $q,p \in \mathrm{B}(q_0, r_{q}')\times \mathrm{B}(p_\epsilon, r_{p}')$,
   $K_{f,a}(q,p)=K_{f,a}(q_0,p_\epsilon)$ and thus for any $q,p \in \mathrm{B}(q_0, r_{q}')\times \mathrm{B}(p_\epsilon, r_{p}')$,
  \begin{equation}
   \label{eq:msj_constant_2}
   \msj=B_{K_{f,a}(q,p)}(a|_{K_{f,a}(q,p)})\eqsp .
  \end{equation}
By \Cref{lemma:nuts-trajectory-probability}, \eqref{eq:msj_constant_2} implies for any $q, p \in \mathrm{B}(q_0, r_{q}')\times \mathrm{B}(p_\epsilon, r_{p}')$
 \begin{equation}
   \rmp_h\big(\msj \mid q, p\big)=\rmp_h(\msj |q_0, p_\epsilon) > 0\eqsp .
 \end{equation}
 Since the function $ \psi_{q_0}^{(1)}$ defined in \eqref{eq:def_psi_q_proof_1} is continuous under \Cref{hyp:regularity} and $\psi_{q_0}^{(1)}(p_0)=q_1$, thus we can choose $\epsilon>0$ such that $\psi_{q_0}^{(1)}(p_\epsilon)=q_1^\epsilon \in \mathrm{B}(q_1, M)\subset \mathsf{E}$.
Since $\mathrm{B}(q_1, M)$ is open, there exists $r_\epsilon>0$ such that $\mathrm{B}(q_1^\epsilon, r_\epsilon) \subset \mathrm{B}(q_1, M) \subset \mse $.
 Using the continuity of $q,p \in (\Rset^d)^2 \mapsto \operatorname{proj}_1\Phiverlet[h][1](q, p) $ and $\psi^{(1)}_{q_0}(p_\epsilon)=q_1^\epsilon $, there exist $r_q, r_p\in (0,r_q')\times(0,r_p')$ such that for any $q, p\in \mathrm{B}(q_0, r_{q})\times \mathrm{B}(p_\epsilon, r_{p})$,
\begin{equation}
   \operatorname{proj}_1\Phiverlet[h][1](q, p)\in \mathrm{B}(q_1^\epsilon, r_\epsilon) \subset \mathrm{B}(q_1, M) \subset \mse \eqsp .
\end{equation}
 This completes the proof with $p_0=p_\epsilon$.

\end{proof}

We now prove \eqref{eq:access_to_prove} distinguishing two cases: $\rmq_h(1|\msj,q_0,p_0)>0$ (\Cref{lemma:rmq_positive})
 and $\rmq_h(1|\msj,q_0,p_0)=0$ (\Cref{lemma:rmq_0}).

\begin{lemma}
   \label{lemma:rmq_positive}
   Under the same conditions as \Cref{lem:1} and using the same notations, for any $q_1\in \mse $ and $M>0$ such that $\mathrm{B}(q_1,M)\subset  \mse $,
   if $\rmq_h(1|\msj,q_0,p_0)>0$, then there exist $m(q_0)>0$ and $r_{q_0}>0$ such that for any $q\in \mathrm{B}(q_0, r_{q_0})$ \eqref{eq:access_to_prove} holds.
\end{lemma}
\begin{proof}

Let $q_1\in \mse $ and $M>0$ such that $\mathrm{B}(q_1,M)\subset  \mse $.
By \Cref{lem:1}, using the continuity of $q,p\in (\Rset^2)^2 \mapsto \rmq_h(1|\msj,q,p)$ given in \Cref{lemma:definition_continuity_q} and $\rmq_h(1|\msj,q_0,p_0)>0$, there exist $\tilde{r}_q, \tilde{r}_p\in (0,r_q)\times(0,r_p)$ and $m'(q_0)>0$ such that for any $q,p\in \mathrm{B}(q_0, \tilde{r}_{q})\times \mathrm{B}(p_0, \tilde{r}_{p})$,
\begin{equation}
   \rmq_h(1|\msj,q,p)\geq m'(q_0) >0\eqsp, \qquad    \rmp_h\big(\msj \mid q, p\big)=\rmp_h(\msj |q_0, p_0)  > 0\eqsp ,
\end{equation}
\begin{equation}
   \operatorname{proj}_1\Phiverlet[h][1](q, p)\in \mathrm{B}(q_1, M) \subset \mse \eqsp.
\end{equation}

Thus, for any $q,p\in \mathrm{B}(q_0, \tilde{r}_{q})\times \mathrm{B}(p_0, \tilde{r}_{p})$, we have by \eqref{eq:transition-kernel},
\begin{equation}
   \KkerUE_h((q,p), \mathsf{E})
   \geq \rmp_h(\msj |q, p) \rmq_h(1 |\msj,q,p)\geq m'(q_0)\rmp_h(\msj |q_0, p_0) > 0\eqsp .
\end{equation}
By using the continuity and positivity of $p\in \Rset^d \mapsto \GaussStandard(p)$, there exists $m_p>0$ such that $\int_{\mathrm{B}(p_0,\tilde{r}_{p})} \GaussStandard(p)= m_p $.
Therefore for any $q \in \mathrm{B}(q_0, \tilde{r}_{q})$,
$$ \KkerU_h(q,\mathsf{E}) \geq \int_{\mathrm{B}(p_0,\tilde{r}_{p})} \GaussStandard(p) \KkerUE_h((q,p); \mathsf{E})\dd p \geq m(q_0) \eqsp , $$
where $m(q_0)=m_p m'(q_0)\rmp_h(\msj |q_0, p_0)>0$, this yields \eqref{eq:access_to_prove} with $r_{q_0}=\tilde{r}_{q}$.
\end{proof}
Before stating and proving \Cref{lemma:rmq_0}, we need the following technical lemma.
\begin{lemma}
   \label{lemma:technical_ball}
   Assume \Cref{hyp:regularity}.
   Let $r_1, r_2>0$, $z_p, z_q\in (\Rset^d)^2$ and $T\in \Zset^*$.
   Denoting by $\msz_T=\Phiverlet[h][T](\mathrm{B}(z_q, r_1)\times \mathrm{B}(z_p, r_2)) $ and $\msk_T=\Phiverlet[h][T](\mathrm{B}(z_q, r_1/2)\times \mathrm{B}(z_p, r_2/2))$, there 
   exists $r_v>0$ such that $\msk_T+\{0_d\}\times \mathrm{B}(0_d,r_v)\subset \msz_T$.
\end{lemma}
\begin{proof}
   Define $r_v=\dist(\bar{\msk}_T,\partial\msz_T )$.

   If $r_v>0$, since $\msk_T\subset\msz_T$ and $\msz_T$ is open, for any $x\in (\Rset^d)^2$ such that $\dist(x,\msk_T)<r_v$, we have $x\in \msz_T$.
   Therefore, we have $\msk_T+\{0_d\}\times \mathrm{B}(0_d,r_v)\subset \msz_T $.

   We prove $r_v>0$ by contradiction. Suppose $r_v=0$.
   Since $r_v=0$, there exists a minimizing sequence $(k_n,z_n)_{n\in\Nset}\in (\bar{\msk}_T\times \partial\msz_T)^\Nset$ such that $\dist(k_n,z_n)\to 0$ as $n\to \infty$.
   $\bar{\msk}_T\times \partial\msz_T$ is compact since it is closed and bounded, thus, there exist an increasing map $\phi :\Nset\mapsto \Nset$ and $k^*,z^* \in \bar{\msk}_T\times \partial\msz_T$ such that $k_{\phi(n)},z_{\phi(n)} \to k^*,z^* $ as $n\to \infty$
    and therefore we have $\dist(k^*,z^*)=0 $ and thus
     \begin{equation}
      \label{eq:equality_limit}
      k^*=z^* \eqsp .
     \end{equation}
   
   Using the fact that $\Phiverlet[h][T]$ is an homemorphism, we have 
   \begin{align}
      &\Phiverlet[h][-T](\bar{\msk}_T)=\bar{\mathrm{B}}(z_q, r_1/2)\times \bar{\mathrm{B}}(z_p, r_2/2) \eqsp ,
     \\
      &\Phiverlet[h][-T](\partial \msz_T)= \Phiverlet[h][-T](\bar{\msz}_T)\setminus \Phiverlet[h][-T](\msz_T)=\mathrm{S}(z_q, r_1)\times \mathrm{S}(z_p, r_2) \eqsp ,
    \end{align} where $\mathrm{S}(z, r)=\{x\in \Rset^d: |x-z|=r\}$ for any $z\in \Rset^d, r>0$. In particular, we have 
    \begin{equation}
      \label{eq:intersect_empty}
      \Phiverlet[h][-T](\bar{\msk}_T) \cap \Phiverlet[h][-T](\partial \msz_T)=\empty \eqsp .
    \end{equation}
    Consider for any $n\in\Nset$, $(k_n^{-1},z_n^{-1})=\left(\Phiverlet[h][-T](k_n),\Phiverlet[h][-T](z_n)\right)\in \Phiverlet[h][-T](\bar{\msk}_T)\times \Phiverlet[h][-T](\partial \msz_T)  $.
    Using the continuity of $\Phiverlet[h][-T]$, we have $(k_n^{-1},z_n^{-1})\to \left(\Phiverlet[h][-T](k^*),\Phiverlet[h][-T](z^*)\right)$ as $n\to \infty$ and by \eqref{eq:equality_limit}
    \begin{equation}
      \label{eq:belong_intersect}
       \Phiverlet[h][-T](k^*)=\Phiverlet[h][-T](z^*)\in\Phiverlet[h][-T](\bar{\msk}_T) \cap \Phiverlet[h][-T](\partial \msz_T)\eqsp .
   \end{equation}
   Combining \eqref{eq:intersect_empty} and \eqref{eq:belong_intersect} yields a  contradiction.
\end{proof}
\begin{lemma}
   \label{lemma:rmq_0}
   Under the same conditions as \Cref{lem:1} and using the same notations, for any $q_1\in \mse $ and $M>0$ such that $\mathrm{B}(q_1,M)\subset  \mse $,
    if $\rmq_h(1|\msj,q_0,p_0)=0$, then there exist $m(q_0)>0$ and $r_{q_0}>0$ such that for any $q\in \mathrm{B}(q_0, r_{q_0})$ \eqref{eq:access_to_prove} holds.
\end{lemma}

\begin{proof}
Let $q_1\in \mse $ and $M>0$ such that $\mathrm{B}(q_1,M)\subset  \mse $.
First, we bound from below $\rmp_h $ and $\rmq_h$ before showing \eqref{eq:access_to_prove}.

By \Cref{lem:1}, there exists $r_q,r_p >0$ such that 
   for any $q, p\in \mathrm{B}(q_0, r_{q})\times \mathrm{B}(p_0, r_{p})$,
   \begin{equation}
      \label{eq:acces_rmp_constant_2}
      \rmp_h\big(\msj \mid q, p\big)=\rmp_h(\msj |q_0, p_0)  > 0\eqsp ,\quad  \operatorname{proj}_1\Phiverlet[h][1](q, p)\in \mathrm{B}(q_1, M) \subset \mse \eqsp .
   \end{equation}
Since now by \eqref{eq:transition_prob_q}, $\rmq_h(1|\msj,q_0,p_0)=\bar{\rmq}_h(0,1|\msj, q_0,p_0)=0$ and $\msj=B_K$ for some $K\in \Nset^*$, applying \Cref{lemma:irreducibility_q} there exists $k\in \msj$ such that $\bar{\rmq}_h(0,k|\msj,q_0,p_0)\bar{\rmq}_h(k,1|\msj,q_0,p_0)>0$.
By using the continuity of $p \mapsto \bar{\rmq}_h(i,j|\msj,q_0,p) $ for any $(i,j) \in \msj^2$ given by \Cref{lemma:definition_continuity_q},
there exist $\tilde{r}_q, \tilde{r}_p\in (0,r_q)\times(0,r_p)$ and $m'(q_0)>0$ such that for any $q,p\in \mathrm{B}(q_0, \tilde{r}_{q})\times \mathrm{B}(p_0, \tilde{r}_{p})$,
\begin{equation}
   \label{eq:Pi_1}
   \bar{\rmq}_h(0,k|\msj,q,p)\bar{\rmq}_h(k,1|\msj,q,p)=\rmq_h(k|\msj,q,p)\rmq_h(1-k|\msj-k,\Phiverlet[h][k](q,p))\geq m'(q_0)>0 \eqsp .
\end{equation}

Denote by $\msv_k= \Phiverlet[h][k]( \mathrm{B}(q_0, \tilde{r}_{q})\times \mathrm{B}(p_0, \tilde{r}_{p}))$, 
$\msq_k=\operatorname{proj}_1(\msv_k)$,
 $\check{\msv}_k= \Phiverlet[h][k]( \mathrm{B}(q_0, \check{r}_{q})\times \mathrm{B}(p_0, \check{r}_{p}))$,
  $\check{\msq}_k=\operatorname{proj}_1(\msv_k)$ and $\check{r}_{q},\check{r}_{p}=\tilde{r}_{q}/2,\tilde{r}_{p}/2 $.
By \Cref{lemma:technical_ball}, there exists $r_v>0$,  
such that 
\begin{equation}
   \label{eq:inclusion_technical}
   \check{\msv}_k+\{0_d\}\times \mathrm{B}(0_d,r_v)\subset \msv_k \eqsp.
\end{equation}
With this notation, we consider the lower bound for any $q\in \mathrm{B}(q_0, \tilde{r}_{q})$, 
\begin{equation}
   \label{eq:Ker2ineq}
   (\KkerU_h)^2(q,\mathsf{E})\geq \KkerU_h(q,\check{\msq}_k) \min_{q \in \check{\msq}_k} \KkerU_h(q,\mathsf{E}) \eqsp . 
 \end{equation}
The rest of the proof consists in lower bounding the two terms in the right-hand side. 
 
First we bound from below the term $\KkerU_h(q,\check{\msq}_k)$ for any $q\in \mathrm{B}(q_0, \tilde{r}_{q})$.
By using the continuity and positivity of $p\in \Rset^d \mapsto \GaussStandard(p)$, there exists $m_p>0$ such that $\int_{\mathrm{B}(p_0,\check{r}_p)}\GaussStandard(p)= m_p $.
By using that $ \operatorname{proj}_1(\Phiverlet[h][k](\mathrm{B}(q_0, \check{r}_{q})\times \mathrm{B}(p_0, \check{r}_{p}))=\check{\msq}_k$,
we have by \eqref{eq:transition-kernel}, \eqref{eq:Pi_1} and \eqref{eq:acces_rmp_constant_2} for any $q \in \mathrm{B}(q_0, \check{r}_{q})$,
\begin{equation}
   \label{eq:first_bound_access}
   \KkerU_h(q,\check{\msq}_k)\geq  \int_{\mathrm{B}(p_0,\check{r}_p)} \GaussStandard(p)\rmp_h(\msj|q,p) \rmq_h(k|\msj,q,p)\dd p\geq \rmp_h(\msj|q_0,p_0)m'(q_0)m_p>0 \eqsp .
\end{equation}

We now bound from below $\inf_{q \in \check{\msq}_k} \KkerU_h(q,\mathsf{E})$.

Let $\check{q} \in \check{\msq}_k$ be fixed.
By the definition of $\check{\msq}_k$ there exists $ \check{p} \in \Rset^d$ such that $(\check{q},\check{p}) \in \check{\msv}_k$.
We have by \eqref{eq:inclusion_technical} for any $\check{p}'\in \mathrm{B}(\check{p},r_v)$, $\check{q},\check{p}'\in \msv_k$ and thus 
\begin{equation}
   \label{eq:belong_V0}
   \Phiverlet[h][-k](\check{q},\check{p}')\in \mathrm{B}(q_0, \tilde{r}_{q})\times \mathrm{B}(p_0, \tilde{r}_{p})\eqsp .
\end{equation}
 Then, for any $\check{p}'\in \mathrm{B}(\check{p},r_v)$,
\begin{equation}
   \rmp_h(\msj|q_0,p_0)=\rmp_h(\msj|\Phiverlet[h][-k](\check{q},\check{p}'))=\rmp_h(\msj-k|\check{q},\check{p}') \eqsp,
\end{equation}
where we have used \eqref{eq:acces_rmp_constant_2} for the first equality and \Cref{prop:rmpcondition_check} for the second since $k\in\msj$.
We have for any $\check{p}'\in \mathrm{B}(\check{p},r_v)$ by \eqref{eq:belong_V0} and \eqref{eq:Pi_1},
\begin{equation}
   \label{eq:rmq_bound_below_access}
   \rmq_h(1-k|\msj-k,\check{q},\check{p}')=\rmq_h(1-k|\msj-k,\Phiverlet[h][k](\Phiverlet[h][-k](\check{q},\check{p}')))\geq m'(q_0)>0
\end{equation}
Since the map $p\in \Rset^d\mapsto \GaussStandard(p)$ is positive and continuous and the set $\operatorname {proj}_2 \msv_k$ is bounded, there exists $m_p'>0$ such that for any $\check{p}\in \operatorname {proj}_2 \check{\msv}_k$, we have $\int_{\mathrm{B}(\check{p},r_v)} \GaussStandard(p)\geq m_p'$.
Using that for any $\check{p}'\in \mathrm{B}(\check{p},r_v)$, $\Phiverlet[h][1-k](\check{q},\check{p}')\in \msq_1\subset \mathrm{B}(q_1, M) $, we have
\begin{align}
   \KkerU_h(\check{q},\mathsf{E}) \geq \KkerU_h(\check{q},\mathrm{B}(q_1, M))&\geq   \int_{\mathrm{B}(\check{p},r_v)} \GaussStandard(p)\rmp_h(\msj-k|\check{q},p) \rmq_h(1-k|\msj-k,\check{q},p) \dd p \\
   & \geq m_p' m'(q_0)\rmp_h(\msj|q_0,p_0)>0 \eqsp ,
\end{align}
where the second inequality is given by \eqref{eq:rmq_bound_below_access}, \eqref{eq:acces_rmp_constant_2}.
The previous bound being independent of $\check{q}\in \check{\msq}_k$, we have,
\begin{equation}
   \label{eq:min_bound_access}
   \inf_{\check{q} \in \check{\msq}_k} \KkerU_h(\check{q},\mathsf{E})\geq  m_p' m'(q_0)\rmp_h(\msj|q_0,p_0)>0 \eqsp .
\end{equation}

Combining \eqref{eq:Ker2ineq}, \eqref{eq:first_bound_access} and \eqref{eq:min_bound_access}, we have for any $q \in \mathrm{B}(q_0, \check{r}_p)$,
$$ (\KkerU_h)^2(q,\mathsf{E}) \geq m_p m_p' (m'(q_0)\rmp_h(\msj|q_0,p_0))^2 >0\eqsp .$$
Setting $m(q_0)=m_p m_p' (m'(q_0)\rmp_h(\msj|q_0,p_0))^2 $ and $r_{q_0}=\check{r}_p$ completes the proof.
\end{proof}


\subsection{Proof of \Cref{thm:1nuts-small-sets}}
Let $q_0\in \Rset^d$ and $a\in B_{\Kmax}$,
using \Cref{hyp:3}$(h,\Kmax)$-\ref{hyp:item_ii_homeo} there exist $p_0 \in \Rset^d, r_H>0$ such that for any $T\in[-2^{\Kmax}+1:2^{\Kmax}-1]$,
$\psi_{q_0}^{(T)}|_{\mathrm{B}(p_0,r_H)}$ is a local homeomorphism. 
Using the last claim of \Cref{lemma:nuts-locally-constant-trajectories}, we may additionally assume that $(q_0, p_0) \in \mathrm{G}_a$. 
Indeed, using $\epsilon= r_H/2$, there exists $p_\epsilon\in \mathrm{B}(p_0,r_H/2)$ such that $q_0,p_\epsilon\in \mathrm{G}_a$ and $\psi_{q_0}^{(T)}|_{\mathrm{B}(p_\epsilon,r_H/2)}$ is a local homeomorphism, then we work with $p_\epsilon, r_H/2$ instead of $ p_0,r_H$.

Since $(q_0, p_0) \in \mathrm{G}_a$, there exist $r' > 0$, $r'' \in (0,r_H)$,  and $K \in \mathbb{N}^*$ such that $K_{a,f}(q,p) = K$ and thus $\rmp_h\big(B_K(a|_K) \mid q, p\big) = c > 0$ for any $(q, p) \in \mathrm{B}(q_0, r') \times \mathrm{B}(p_0, r'')$.
By continuity of $(q, p)\in \Rset^d \mapsto \rmq_h\big(j \mid B_K(a|_K), q, p\big)$ for every $j$ proven in \Cref{lemma:definition_continuity_q}, and 
the fact that $\sum_{j\in  B_K(a|_K)} \rmq_h\big(j \mid B_K(a|_K), q, p\big) = 1$ for any $(q, p) \in \mathrm{B}(q_0, r') \times \mathrm{B}(p_0, r'')$ since $\rmq_h\big(\cdot \mid B_K(a|_K), q, p\big)$ is a well defined probability on $B_K(a|_K)$, there exist $r \in (0,r')$ and $\tilde{r} \in (0,r'')$ such that $\inf_{(q, p) \in \mathrm{B}(q_0, r) \times \mathrm{B}(p_0, \tilde{r})} \rmq_h\big(j' \mid B_K(a|_K), q, p\big) = \varepsilon_1 > 0$ for some $j' \in B_K(a|_K)$.
Thus, for any $q \in \mathrm{B}(q_0, r)$ we can bound from below the transition kernel \eqref{eq:transition-kernel} by 
\begin{align}
   \KkerU_h(q_0, \msa)
   & \geq
   \begin{multlined}[t]
       \int \dd p \, \mathbbm{1}_{\mathrm{B}(p_0, \tilde{r})}(p) \GaussStandard(p) \rmp_h\big(B_K(a|_K) \mid q_0, p\big)
       \\
       \sum_{j \in \mathbb{Z}} \rmq_h\big(j \mid B_K(a|_K), q_0, p\big) \mathbbm{1}_\msa(\operatorname{proj}_1(\Phiverlet[h][j'](q_0, p)))
   \end{multlined}
   \\
   & \geq  \min_{p\in \bar{\mathrm{B}}(p_0, \tilde{r})} \{\GaussStandard(p)\} c \varepsilon_1
   \int \dd p \, \mathbbm{1}_{\mathrm{B}(p_0, \tilde{r})}(p) \mathbbm{1}_\msa(\operatorname{proj}_1(\Phiverlet[h][j'](q_0, p)))\eqsp .
\end{align}
To bound the last term, we use the following Lemma.

\begin{lemma}
   \label{lemma:ball_image_central}
Let $j\in \Zset$ and $q_0\in \Rset^d$. If for any $q\in \Rset^d$ the function $\psi_q^{(j')}:  p\in \Rset^d\to \operatorname{proj}_1(\Phiverlet[h][j'](q, p))$ is Lipschitz with Lipschitz constant denoted by $L_\psi$,
and that there exist $\tilde{q}, \tilde{p} \in (\Rset^d)^2$, $M, \tilde{M}, r_1 >0$  such that for any $q\in \mathrm{B}(q_0,r_1)$, we have
\begin{equation}\label{eq:ball-image-ball_hmc}
   \mathrm{B}(\tilde{q}, M) \subset \psi_q(\mathrm{B}(\tilde{p}, \tilde{M}))\eqsp .
\end{equation}
Then, for any $q\in \mathrm{B}(q_0,r_1)$,
\begin{equation}
   \int \dd p \, \mathbbm{1}_{\mathrm{B}(\tilde{p}, \tilde{M})}(p) \mathbbm{1}_\msa(\psi_{q}^{(j')}(p)) \geq L_\psi^{-d}\Leb (\mathrm{B}(\tilde{q}, M)\cap \msa) \eqsp .
\end{equation}
\end{lemma}
\begin{proof}
   This Lemma is a direct application of \cite[Proposition 11]{Durmus2017-tf} for different choice of $\Theta$, with their notations $\Theta=\psi_{q}^{(j')}$ for any $q \in \mathrm{B}(q_0,r_1)$ and $\mathrm{B}(\tilde{y}_0, \tilde{M})=\mathrm{B}(\tilde{p}, \tilde{M})$ .
   They define the measure $\lambda_\Theta$ on $(\Rset^d, \msb(\Rset^d))$ by setting for any $\msa \in \msb(\Rset^d)$,
   $$\lambda_\Theta(\msa)= \int \dd p \, \mathbbm{1}_{\mathrm{B}(\tilde{p}, \tilde{M})}(p) \mathbbm{1}_\msa(\psi_{q}^{(j')}(p))=\Leb((\psi_{q_0}^{(j')})^{-1}(\msa)\cap \mathrm{B}(\tilde{p}, \tilde{M})) \eqsp .$$
\end{proof}
If we show the assumptions of \Cref{lemma:ball_image_central}, we have for any $q\in \mathrm{B}(q_0,r_1)$,
$$\KkerU_h(q, \msa)\geq \min_{p\in \bar{\mathrm{B}}(p_0, \tilde{r})} \{\GaussStandard(p)\}  \varepsilon_1c L_{\psi}^{-d}\Leb (\mathrm{B}(\tilde{q},M)\cap \msa) \eqsp ,$$
which gives that $ \mathrm{B}(q_0,r_1)$ is a 1-small set.

The stated uniform Lipschitz continuity of $\psi_q^{(j')}$ follows from \cite[Lemma 17]{Durmus2017-tf} and \Cref{hyp:regularity}.
To show \eqref{eq:ball-image-ball_hmc},
we denote by $\tilde{\mathrm{P}}=\psi_{q_0}(\mathrm{B}(p_0, \tilde{r}))$.
As an image of an open set under a homemorphism, $\tilde{\mathrm{P}}$ is open.
Thus there exist $\tilde{q} \in \Rset^d, M > 0$ such that $\mathrm{B}(\tilde{q}, 2M) \subset \tilde{\mathrm{P}}$. 
According to \cite[Lemma 17]{Durmus2017-tf}, the map $q\in \Rset^d \mapsto \psi_q(p)$ is Lipschitz for any $p\in \Rset^d$ with a Lipschitz constant $L_{lip,2}$ independent of $p$.
It follows that with $r_1=M/(2L_{lip,2})>0$, for any $q\in \mathrm{B}(q_0,r_1)$ we have $\mathrm{B}(\tilde{q},M)\subset \psi_q(\mathrm{B}(\tilde{p},\tilde{M}))$ with $\tilde{M} = \tilde{r}$ and $\tilde{p} = p_0$.
This completes the proof.
\subsection{Proof of \Cref{thm:3small}}


We begin the proof with the following Lemma, which shows mainly that a set can be accessible in one step.
\begin{lemma}
   \label{lemma:pre_aperiodicity}
   Assume \Cref{hyp:regularity}.
    For any $q_0\in \Rset^d$ such that $\nabla U(q_0)\neq 0$,
     there exists $\epsilon_0>0$ such that for any $\epsilon \in (0,\epsilon_0)$,
     there exist $ m_0>0, r_{q_0}\in (0,\epsilon)$ such that
    for any $q\in \mathrm{B}(q_0,\epsilon)$, we have $\KkerU_h(q,\mathrm{B}(q_0,\epsilon))>0$ and 
    moreover for any $q\in \mathrm{B}(q_0,r_{q_0})$, we have $\KkerU_h(q,\mathrm{B}(q_0,\epsilon))\geq m_0$.
\end{lemma}
\begin{proof}
Let $q_0\in \Rset^d$ such that $\nabla U(q_0)\neq 0$.
By continuity of $\nabla U$, there exists $\epsilon_0>0$ such that $ \nabla U(q)\neq 0$ for any $q\in \mathrm{B}(q_0,\epsilon_0)$.
Let $\epsilon\in(0,\epsilon_0)$ be fixed.

Denote by $p_0(\alpha,q)=(h/2-\alpha/h)\nabla U(q)$ for any $\alpha\in (0,h^2/2)$.
We have, for any $\alpha\in (0,h^2/2)$ and $q\in \mathrm{B}(q_0,\epsilon)$,
\begin{align}
   &q_1: (q,p) \in (\Rset^d)^2 \mapsto q+hp-h^2\nabla U(q)/2\eqsp,\\
   &q_1(q,p_0(\alpha,q))=q-\alpha \nabla U(q)\eqsp,\\
   &p_0(\alpha,q)^\top(q_1(q,p_0(\alpha,q))-q)= -\alpha (h/2-\alpha/h)|\nabla U(q)|^2<0 \eqsp.
\end{align}
Let $M=\sup_{q\in \mathrm{B}(q_0,\epsilon)} |\nabla U(q)|$.
By setting $\alpha_q=(\epsilon-|q-q_0|) /(2M)$ we have, for any $q\in \mathrm{B}(q_0,\epsilon)$,
\begin{equation}
   q_1(q,p_0(\alpha_q,q))\in \mathrm{B}(q_0,\epsilon) \eqsp,
\end{equation}
where the condition $\alpha_q \in (0, h^2/2)$ may be verified by reducing $\epsilon_0$ if necessary (depending only on $h$ and $M$).
Let $\bar{q}\in \mathrm{B}(q_0,\epsilon)$ be fixed.
Then, by continuity of $q,p\in (\Rset^d)^2 \mapsto (p^\top(q_1(q,p)-q),q_1(q,p))$ (the preimage image of an open set is an open set),
 there exist $r_{\bar{q},1},r_{\bar{q},2}>0$ such that for any $(q,p)\in \mathrm{B}(\bar{q},r_{\bar{q},1})\times \mathrm{B}(p_0(\alpha_{\bar{q}},\bar{q}),r_{\bar{q},2})$, we have $q_1(q,p)\in \mathrm{B}(q_0,\epsilon)$ and 
$p^\top(q_1(p,q)-q)<0$.
Setting $a=2^{\Kmax}-1\in B_{\Kmax}$, for any $(q,p)\in \mathrm{B}(\bar{q},r_{\bar{q},1})\times \mathrm{B}(p_0(\alpha_{\bar{q}},\bar{q}),r_{\bar{q},2})$, we have $K_{f,a}(q,p)=1$ since the No U-turn criterion is activated and thus 
$B_{K_{f,a}(q,p)}(a|_{K_{f,a}(q,p)})=\{0,1\}$.
Therefore, for any $(q,p)\in \mathrm{B}(\bar{q},r_{\bar{q},1})\times \mathrm{B}(p_0(\alpha_{\bar{q}},\bar{q}),r_{\bar{q},2})$, we have $\rmp_h(\{0,1\}|q,p)=\mathbb{P}(V_0=1)=1/2 $ and thus $\rmq_h(0|\{0,1\},q,p)+\rmq_h(1|\{0,1\},q,p)=1$.
Using the fact that for any $(q,p)\in \mathrm{B}(\bar{q},r_{\bar{q},1})\times \mathrm{B}(p_0(\alpha_{\bar{q}},\bar{q}),r_{\bar{q},2})$ also $(q,q_1(q,p))\in \mathrm{B}(q_0,\epsilon)^2$ we have
\begin{equation}
   \label{eq:main_aperiodicity}
   \KkerU_h(q,\mathrm{B}(q_0,\epsilon))\geq \int_{\mathrm{B}(p_0(\alpha_{\bar{q}},\bar{q}),r_{\bar{q},2})} \rmp_h(\{0,1\}|q,p) \GaussStandard(p) \dd p = \int_{\mathrm{B}(p_0(\alpha_{\bar{q}},\bar{q}),r_{\bar{q},2})} \GaussStandard(p) \dd p/2  \eqsp .
\end{equation}
Using positivity and continuity of $\GaussStandard$ we conclude that $\KkerU_h(\bar{q},\mathrm{B}(q_0,\epsilon))>0$.
As $\bar{q}$ was arbitrary, for any $q\in \mathrm{B}(q_0,\epsilon)$ it follows that $\KkerU_h(q,\mathrm{B}(q_0,\epsilon))>0 $.
Taking $\bar{q}=q_0$ and setting $m_0=\int_{\mathrm{B}(p_0(\alpha_{q_0},q_0),r_{q_0,2})} \GaussStandard(p) \dd p/2>0$, by \eqref{eq:main_aperiodicity}, we have $\KkerU_h(q,\mathrm{B}(q_0,\epsilon))\geq m_0>0$ for any $q\in \mathrm{B}(q_0,r_{q_0,1})$.
This completes the proof.
\end{proof}

In order to prove \Cref{thm:3small} we show that there exists a probability measure $\mu$ such that for any $q_0\in \Rset^d$ there exists a neighborhood $\msw(q_0)$ of $q$ such that $\msw(q_0)$ is a 3-small set with the minorizing measure $\mu$.

Choose $q_m\in \Rset^d$ so that $\nabla U(q_m)\neq 0$, which is possible since $U$ is not constant by hypothesis \Cref{hyp:regularity}.
Applying \Cref{thm:1nuts-small-sets} to $q_m\in \Rset^d$, there exists $r_m>0$ such that $\mathrm{B}(q_m,r_m)$ is a 1-small set, i.e.,
there exists a probability measure $\mu$ on $\Rset^d$ and $\eta>0$ such that for any $q\in \mathrm{B}(q_m,r_m)$,
\begin{equation}
   \label{eq:3small_1small}
   \ee \KkerU_h(q,\msa)\geq \eta \mu(\msa)\eqsp . 
\end{equation}
\Cref{lemma:pre_aperiodicity} shows that there exists $\epsilon_0>0$ such that for any $\epsilon \in (0,\epsilon_0)$, there exist $ m_{q_m}>0, r_{q_m}\in (0,\epsilon)$ such that
for any $q\in \mathrm{B}(q_m,\epsilon)$ we have $\KkerU_h(q,\mathrm{B}(q_m,\epsilon))>0$ and further
for any $q\in \mathrm{B}(q_m,r_{q_m})$ we have
\begin{equation}
   \label{eq:main_3small}
   \KkerU_h(q,\mathrm{B}(q_m,\epsilon))\geq m_{q_m} \eqsp.
\end{equation}
Finally, choose $\epsilon = \min(r_m,\epsilon_0/2)$.

Fix $q_0\in \Rset^d$.
By applying \Cref{thm:nuts-accessibility} with $\mse=\mathrm{B}(q_m,r_{q_m})$,
 there exist a neighborhood $\msw(q_0) \ni q_0$, a constant $m_\msw(q_0)\in (0,1)$ and $j(q_0)\in \{1,2\}$ such that for any $q\in \msw(q_0)$ we have
\begin{equation}
   \label{eq:3small_access}
   (\KkerU_h)^{j(q_0)}(q,\mathrm{B}(q_m,r_{q_m}))\geq m_\msw(q_0)>0 \eqsp .
\end{equation}
If $j(q_0)=1$, for any $q\in \msw(q_0)$ we have, by \eqref{eq:3small_access} and \eqref{eq:main_3small},
\begin{align} 
   \label{eq:3small_align1}
   (\KkerU_h)^{2}(q,\mathrm{B}(q_m,r_m))&\geq  \KkerU_h(q,\mathrm{B}(q_m,r_{q_m})) \inf_{q'\in \mathrm{B}(q_m,r_{q_m})}\KkerU_h(q',\mathrm{B}(q_m,r_m)) \\
&\geq m_\msw(q_0) m_{q_m}>0 \eqsp . 
\end{align}
If $j(q_0)=2$, we use $r_{q_m}<r_m$ and \eqref{eq:3small_access} to obtain, for any $q\in \msw(q_0)$,
\begin{equation}
   \label{eq:3small_align2}
   (\KkerU_h)^{2}(q,\mathrm{B}(q_m,r_m))\geq  (\KkerU_h)^{2}(q,\mathrm{B}(q_m,r_{q_m}))\geq m_\msw(q_0)\geq m_{q_m} m_\msw(q_0) > 0 \eqsp . 
\end{equation}
Thus, for any $q_0\in \Rset^d$ and $q\in \msw(q_0)$ we have, by \eqref{eq:3small_align1}, \eqref{eq:3small_align2} and \eqref{eq:3small_1small},
\begin{align}
   (\KkerU_h)^{3}(q,\msa) &\geq (\KkerU_h)^{2}(q,\mathrm{B}(q_m,r_m)) \inf_{q'\in \mathrm{B}(q_m,r_m)} \KkerU_h(q',\msa) \\ 
&\geq  m_{q_m} m_\msw(q_0) \inf_{q'\in \mathrm{B}(q_m,r_m)} \KkerU_h(q',\msa) \\
& \geq m_{q_m} m_\msw(q_0) \eta \mu(\msa) \eqsp .
\end{align}
It implies that for any $q_0\in \Rset^d$, $\msw(q_0)$ is a 3-small set with constant $\eta'(q_0) > 0$ and minorizing measure $\mu$ with $\eta'(q_0)=m_{q_m} m_\msw(q_0) \eta$.

Let $K$ be a compact set of $\Rset^d$.
By compactness, there exists a finite set $F \subset K$ such that $K \subset \bigcup_{q_0 \in F} \msw(q_0)$.
 We define $\eta=\min_{q_0 \in F}\eta'(q_0)$ and it follows that
$K$ is a 3-small set associated to $\eta_K$ and $\mu$.
From \Cref{thm:nuts-accessibility} and \Cref{thm:1nuts-small-sets} it follows that $\KkerU_h$ is irreducible.
Combining this with the fact that $\Rset^d$ can be expressed as a countable increasing union of small sets $\Rset^d = \cup_{k\in \Nset^*} \bar{\mathrm{B}}(0_d,k)$, Proposition 9.4.11 in \cite{douc2018markov} implies that $\KkerU_h$ is aperiodic.

\subsection{Proof of \Cref{prop:general_condition_check}}
\label{sec_supp:condition_check_prop}
\subsubsection{Proof of \Cref{prop:general_condition_check}-\ref{prop:condition_check} under \Cref{hyp:regularity} and \Cref{hyp:lipschitz_hard}($h,\Kmax$)}

    We prove \Cref{hyp:3}$(h,\Kmax)$-\ref{hyp:item_ii_homeo} as a consequence of \Cref{thm:trajectory-transitivity}.
    Namely, hypothesis \Cref{hyp:lipschitz_hard}($h,\Kmax$) implies \eqref{eq:transitivity-condition} since for $h>0$,
    $$2\ltt_1^{\frac{1}{2}}h\leq 2^{\Kmax}h\ltt_1^{\frac{1}{2}} \leq 2^{\Kmax}h \ltt_1^{\frac{1}{2}}\mathcal{V}(h \ltt_1^\frac{1}{2})\leq [(1+h \ltt_1^{\frac{1}{2}}\mathcal{V}(h \ltt_1^\frac{1}{2}))^{2^{\Kmax}}-1]\eqsp , $$
    which tackles the case $\Kmax=1$ and gives the result for $\Kmax\geq 2$,
    $$2 \left(1 - \cos \frac{\pi}{2^{\Kmax}}\right)\geq \frac{\pi^2}{2^{2\Kmax}}-\frac{\pi^4}{3\times  2^{4\Kmax+2}}\geq \frac{1}{2^{2\Kmax}}\geq \frac{1}{4}\eqsp ,  $$
    gives the result for $\Kmax\geq 2$.
    It remains to show \Cref{hyp:3}$(h,\Kmax)$-\ref{hyp:item_i_topo}.

    As a prelimenary remark, using the equations \cite[(32),(33)]{Durmus2017-tf} and \Cref{hyp:regularity}, $F_q^{T_1,T_2}, \Phiverlet[h][T_1]$ are $\mathrm{C}^1$ for any $q\in \Rset^d$ and $T_1,T_2\in[-2^{\Kmax}+1:2^{\Kmax}-1]^2$ with $T_1\neq T_2$, where $F_q^{T_1,T_2}$ is defined in \eqref{eq:F_T1T2}.
    \begin{lemma}
        \label{lemma:densitygradient_imply}
    For any $q\in \Rset^d$, denote by
    \begin{equation}
        \label{eq:nabla_antecedents}
    \nabla\msf_{q,-0}=\{p\in \Rset^d:\, \nabla F_q^{T_1,T_2}(p)\neq 0,\, T_1,T_2\in[-2^{\Kmax}+1:2^{\Kmax}-1]^2,\, T_1\neq T_2 \}\eqsp .
    \end{equation}
    If $\nabla\msf_{q,-0}$ is dense, then $ \msf_{q,-0}$ is dense.
\end{lemma}
\begin{proof}
    Indeed, we have the following decomposition,
    \begin{equation}
    \msf_{q,-0}^c=\cup_{T_1,T_2\in[-2^{\Kmax}+1:2^{\Kmax}-1]^2,\, T_1\neq T_2 }\{p\in \Rset^d:\, F_q^{T_1,T_2}(p)= 0\}\eqsp,
    \end{equation}
    if $\msf_{q,-0}$ is not dense, the interior of $\msf_{q,-0}^c$ is not empty and then by the consequence of Baire's Theorem,
    there exists $T_1,T_2 \in  [-2^{\Kmax}+1:2^{\Kmax}-1]^2$ with $ T_1\neq T_2$ such that the interior of $\{p\in \Rset^d:\,  F_q^{T_1,T_2}(p)= 0\}$ is not empty.
    Therefore the interior of $\{p\in \Rset^d:\,  \nabla F_q^{T_1,T_2}(p)= 0\}$ is not empty which contradicts the density of $\nabla\msf_{q,-0}$.
\end{proof}
    Let $q_0\in \Rset^d$, we prove that $\nabla\msf_{q_0,-0}$ is dense. 
    Denote by $q_T(p),p_T(p)=\Phiverlet[h][T](q_0,p)$ for any $T\in\Zset$, $p\in \Rset^d$.
    Let $p_0\in \Rset^d$
    and let $T_1,T_2\in[-2^{\Kmax}+1:2^{\Kmax}-1]^2$ with $T_1\neq T_2$.
    Now, we explicit the value of $\nabla F_{q_0}^{T_1,T_2}(p_0)$.
     Using \cite[(32)--(33)]{Durmus2017-tf}, we have,
    \begin{equation}
    \label{eq:phiver_q_develop}
    q_{T_2}(p_0)-q_{T_1}(p_0)= (T_2-T_1)hp_0-(T_2-T_1)\frac{h^2}{2}\nabla U(q_0)-h^2 \Theta_{h,T_2-T_1}^{(1)}(q_{T_1}(p_0),p_{T_1}(p_0))\eqsp ,
    \end{equation}
    \begin{equation}
    \label{eq:phiver_p_develop}
    p_{T_1}(p_0)=p_0-\frac{h}{2}(\nabla U(q_0)+\nabla U(q_{T_1}(p_0)))-h\Theta_{h,T_1}^{(2)}(q_0,p_0)\eqsp ,
    \end{equation}
    where $\Theta_{h,T}^{(1)},\Theta_{h,T}^{(2)} : (\Rset^d)^2\mapsto \Rset^d $ are given, for any $(q,p)\in (\Rset^d)^2$, by
    \begin{equation}
    \label{eq:theta_1}
    \Theta_{h,T}^{(1)}(q,p)=\sum_{i=1}^{T-1}(T-i)\nabla U\circ \operatorname{proj}_1(\Phiverlet[h][i](q,p))\eqsp ,
    \end{equation}
    $$\Theta_{h,T}^{(2)}(q_0,p_0)=\sum_{i=1}^{T-1}\nabla U\circ \operatorname{proj}_1(\Phiverlet[h][i](q,p))\eqsp .$$ 
    In all the following, for any continuously differentiable function $f: (\Rset^d)^2\mapsto \Rset^d$, we will denote by $\dd_{p_0} f(q_0,p_0) $ the value of $\dd f(q_0,\cdot)(p_0)$.
    
    As the main step we show that $\nabla F_{q_0}^{T_1,T_2}(p_0)\neq 0$ whenever either $p_{T_1}(p_0) \neq 0$ or $(q_{T_2} - q_{T_1})(p_0) \neq 0$. 
    We have for any $\tilde{p}\in \Rset^d$,
    \begin{equation}
        \label{eq:nabla_F_expression}
         (\nabla F_{q_0}^{T_1,T_2}(p_0))^\top \tilde{p}= (\dd p_{T_1}(p_0).\tilde{p})^\top (q_{T_2}-q_{T_1})+p_{T_1}^\top (\dd(q_{T_2}-q_{T_1})(p_0).\tilde{p})\eqsp . 
    \end{equation}
    With the expression of $p_{T_1}$ and $q_{T_2}-q_{T_1}$, we have for any $\tilde{p}\in \Rset^d$
    \begin{equation}
    \label{eq:df_p_0} 
    \dd p_{T_1}(p_0).\tilde{p}=(\tilde{p}-A\tilde{p})\eqsp,
    \end{equation}
    where
    \begin{equation}
    \label{eq:A_def}
    A=h\frac{1}{2}\nabla^2 U(q_{T_1}(p_0)) \dd q_{T_1}(p_0)+\dd_{p_0} \Theta_{h,T_1}^{(2)}(q_0,p_0)\eqsp ,
    \end{equation}
    
    Similarly
    \begin{equation}
    \label{eq:df_qq}
    \dd(q_{T_2}-q_{T_1})(p_0).\tilde{p}=(T_2-T_1)h(\tilde{p}-B\tilde{p})\eqsp,
    \end{equation}
    with
    \begin{equation}
    \label{eq:B_def}
    B=\frac{h}{T_2-T_1} \dd_{p_0}(\Theta_{h,T_2-T_1}^{(1)}\circ \Phiverlet[h][T_1])(q_0,p_0)\eqsp . 
    \end{equation}
    We use \eqref{eq:nabla_F_expression} and a technical lemma to show that the gradient is not zero.

    \textbf{If $p_{T_1} \neq 0$ or $q_{T_2} - q_{T_1} \neq 0$}
    We aim to apply the following lemma with $q=q_{T_2}-q_{T_1}$, $p=(T_2-T_1)h p_{T_1}$ and $C=-A,D=-B $.
    \begin{lemma}
    \label{lemma:perp}
    Assume that $q,p\in(\Rset^d)^2$ are not both zero and satisfy $p^\top q=0$, let $C, D\in \Rset^{d\times d}$ satisfy $|C|, |D|\leq 1/4$.
     Then, there exists $\tilde{p}\in\Rset^d$ so that
    $$
    q^\top (\tilde{p}+C\tilde{p})+ p^\top (\tilde{p}+D\tilde{p})\not=0\eqsp .
    $$
    \end{lemma}
    \begin{proof}
    
    Assume e.g. that $q\neq 0$, in which case we choose $\tilde{p}=(\Idd+D)^{-1}q$.
     With this choice we have automatically $p^\top (\tilde{p}+D\tilde{p})=0$.
      Moreover, we may write
    $$
    (1+D)^{-1}=I+\sum_{k=1}^\infty (-1)^kD^k=1+D'\eqsp ,
    $$
    where $|D'|\leq 1/4+(1/4)^2+\ldots=1/3$.
     The Lemma follows by observing that
    $$
    q^\top (\tilde{p}+C\tilde{p})= q^\perp(1+C)(1+D')q\; \geq\;  |q|^2\big(1-|C| -|D|  |C| |D| \big)\; \geq \; \frac{5}{12}|q|^2 >0\eqsp .
    $$
    The case where $p=0$ is analogous.
    If $q=0$ but $p\neq 0$, we take $\tilde{p}=-p$ and the claim follows from $|p|^2(1-|D|)>0 $.
    \end{proof}
    We show that $\max(|A|,|B|)\leq 1/4 $ where $A$ and $B$ are defined in \eqref{eq:A_def}, \eqref{eq:B_def}.
    Applying \Cref{hyp:lipschitz_hard}($h,\Kmax$) and \cite[Lemma 17]{Durmus2017-tf} we have, for any $ k\in \Nset^*$ and for any $q,p,p'\in(\Rset^d)^3$,
    \begin{equation}
        |\operatorname{proj}_1(\Phiverlet[h][i](q,p))-\operatorname{proj}_1(\Phiverlet[h][i](q,p'))|\leq \ltt_1^{-\frac{1}{2}}(1+h \ltt_1^{\frac{1}{2}}\mathcal{V}(h \ltt_1^\frac{1}{2}))^k |p-p'|\eqsp .
    \end{equation}
    Moreover, in view of the last inequality we have, for any $q,p \in (\Rset^d)^2$,
    \begin{equation}
        |\dd_{p} \operatorname{proj}_1(\Phiverlet[h][k](q,p))|\leq  \ltt_1^{-\frac{1}{2}}(1+h \ltt_1^{\frac{1}{2}}\mathcal{V}(h \ltt_1^\frac{1}{2}))^k\eqsp . 
    \end{equation}
    Combining this with \Cref{hyp:lipschitz_hard}($h,\Kmax$) and $|\nabla^2 U(q)|\leq \ltt_1 $, for any $p_0\in \Rset^d$,
    \begin{align}
        &h\frac{1}{2}\nabla^2 U(q_{T_1}(p_0)) \dd q_{T_1}(p_0)+\dd_{p_0}\Theta_{h,T_1}^{(2)}(q_0,p_0) \\
     &\leq h \ltt_1^{\frac{1}{2}} \sum_{i=1}^{T_1} (1+h \ltt_1^{\frac{1}{2}}\mathcal{V}(h \ltt_1^\frac{1}{2}))^{i}\leq \frac{[(1+h \ltt_1^{\frac{1}{2}}\mathcal{V}(h \ltt_1^\frac{1}{2}))^{T_1}-1]}{\mathcal{V}(h \ltt_1^\frac{1}{2})} \\
    &\leq[(1+h \ltt_1^{\frac{1}{2}}\mathcal{V}(h \ltt_1^\frac{1}{2}))^{\Kmax}-1]<\frac{1}{4}\eqsp .
    \end{align}
    The last inequality holds because $\mathcal{V}(h \ltt_1^{1/2})\geq 1 $.
    
    The inequalities above show that $|A|<1/4$, and the demonstration for $|B|<1/4$ is essentially the same by noting that  $T_2-T_1-i \leq T_2-T_1$ in the expression of $\Theta_{h,T_2-T_1}^{(1)}$.
    Thus Lemma \ref{lemma:perp} may be applied with \eqref{eq:nabla_F_expression} to finish the proof that $\nabla F_{q_0}^{T_1,T_2}(p_0)\neq 0$ whenever either $p_{T_1} \neq 0$ or $q_{T_2}-q_{T_1} \neq 0$.
    
    The symmetric situation where $p_{T_1}$ is replaced by $p_{T_2}$ does not affect the general sketch of proof and the proof of $|A|\leq 1/4$ and $|B|\leq 1/4$.
    
    \textbf{If $p_{T_1}=0$ and $q_{T_1}=q_{T_2}$}.
    For any $q_0 \in \Rset^d$, let $\Delta_{q_0} \subset \Rset^d$ denote the set of initial momenta $p_0$ for which this happens.
    We show that $\Delta_{q_0}$ is a finite set for any $q_0 \in \Rset^d$ to conclude. Indeed,
    by setting
    \begin{equation} 
    \mathrm{Y}_{q_0}=\Rset^d \setminus \Delta_{q_0} \subset \nabla \msf_{q_0,-0} \eqsp ,
    \end{equation}
    we have that $\mathrm{Y}_{q_0}$ is dense if $\Delta_{q_0}$ is a finite, then $\nabla \msf_{q_0,-0}$ is dense and thus $\msf_{q_0,-0}$ is dense by \Cref{lemma:densitygradient_imply}.

    For any $q_0\in \Rset^d$ and $T_1\in \Zset^*$, $f^{T_1}_{q_0}: p\in \Rset^d\mapsto \operatorname{proj}_1 \Phiverlet[h][T_1](q_0,p)$ is $\mathrm{C}^1$ on $\Rset^d$ and from \eqref{eq:df_p_0} we have $\dd f^{T_1}_{q_0}(p_0)= \Idd -A$ with $|A|<1$.
    It follows that for any $q_0\in \Rset^d$, $f^{T_1}_{q_0}$ is a $\mathrm{C}^1$-diffeomorphism and $(f^{T_1}_{q_0})^{-1}(\{0\})$ is a singleton (see Exercice 3.26 in \cite{c1duistermaat2004multidimensional}). 
    Thus, for any $q_0\in \Rset^d$,
    \begin{equation} \label{eq:set-Delta}
    \Delta_{q_0}=\bigcup_{i=-2^\Kmax+1}^{2^\Kmax-1} (f^{i}_{q_0})^{-1}(\{0\})   
    \end{equation}
    is a finite set,
    finishing the proof.

    \begin{remark}
        In the last demonstration we can replace $U\in \mathrm{C}^2(\Rset^d)$ by $U\in \mathrm{C}^1(\Rset^d)$ and $\nabla U$ Lipschitz if we find $\tilde{p}$ with a Banach fixed point argument.
        The proof being technical already, we have preferred to give a simpler argument.
    \end{remark}
    
\subsubsection{Proof of \Cref{prop:general_condition_check}-\ref{prop:condition_check} under \Cref{hyp:regularity} and \Cref{hyp:analytic_potential}}

    We start by proving several lemmas.
    \begin{lemma}
      \label{lemma:bounding_gradient}
      Assume $\lim_{|q| \to \infty} \| \nabla^2 U(q) \| = 0$.
      Then,
         $\lim_{|q|\to \infty} |\nabla U(q)|/|q|=0 $
    \end{lemma}
    \begin{proof}
      Let $\epsilon>0$.
      There exists $R>0$ such that for any $q\in \Rset^d$ with $|q|\geq R$, $\| \nabla^2 U(q) \|\leq \epsilon $ .
       Let $\bar{q}\in \Rset^d$ such that $|\bar{q}|\geq R$.
      By bounding the variation, for any $q\in \Rset^d$ such that $|q|\geq R$,
      \begin{equation}
         |\nabla U(q)-\nabla U(\bar{q})|\leq\epsilon  |q-\bar{q}| \eqsp .
      \end{equation}
      Then, by the triangule inequality, for any $q\in \Rset^d$ such that $|q|\geq R$,
      \begin{equation}
         \frac{|\nabla U(q)|}{|q|}\leq \epsilon \frac{|q-\bar{q}|}{|q|}+\frac{|\nabla U(\bar{q})|}{|q|} \eqsp ,
      \end{equation}
      letting $|q|\to \infty$ yields that there exists $R_1>0$ such that for any $q\in \Rset^d$ with $|q|\geq R_1$, we have,
      \begin{equation}
         \frac{|\nabla U(q)|}{|q|}\leq \epsilon \eqsp,
      \end{equation}
      this completes the proof.
    \end{proof}

    \begin{lemma}
      \label{lemma:identity_far}
      Assume $h > 0$ and that $\lim_{|q| \to \infty} \| \nabla^2 U(q) \| = 0$. Let $q_0\in \Rset^d$ and denote by $q_T:\, p\in \Rset^d \mapsto\operatorname{proj}_1 \Phiverlet[h][T](q_0, p)$ and $p_T:\, p\in \Rset^d \mapsto\operatorname{proj}_2 \Phiverlet[h][T](q_0, p)$.
      Then, 
      \begin{enumerate}[wide, labelwidth=!, labelindent=0pt, label=(\alph*),leftmargin=*] 
         \item \label{lemma:item1_affine_qT} For any $T\in \Zset^*$, there exists a function $\epsilon_{h,T,q_0}^{(1)}:\Rset_{\geq0} \mapsto \Rset_{\geq0} $ depending on $T,q_0$ such that
      \begin{equation}
         \label{eq:q_T_affine}
      \frac{| q_T(p_0) - T h p_0|}{|p_0|} \leq \epsilon_{h,T,q_0}^{(1)}(|p_0|) 
         \qquad
         \,
        \lim_{|p_0|\to \infty} \epsilon_{h,T,q_0}^{(1)}(|p_0|)=0\eqsp .
      \end{equation}
      \item \label{lemma:item2_affine_pT} For any $T\in \Zset^*$, there exists a function $\epsilon_{h,T,q_0}^{(2)}:\Rset_{\geq0} \mapsto \Rset_{\geq0} $ depending on $T,q_0$ such that
      \begin{equation}
         \label{eq:p_T_affine}
      \frac{| p_T(p_0) - p_0|}{|p_0|} \leq \epsilon_{h,T,q_0}^{(2)}(|p_0|) 
         \qquad
         \,
        \lim_{|p_0|\to \infty} \epsilon_{h,T,q_0}^{(2)}(|p_0|)=0\eqsp .
      \end{equation}
   
      \item \label{lemma:item3_small_df} For any $T\in \Zset^*$, 
         \begin{equation}
            |\dd  q_T(p_0)-hT \Idd| \to 0 \, \text{as}\, |p_0|\to \infty \eqsp .
         \end{equation}
   \end{enumerate}
      
   \end{lemma}
   \begin{proof}
      Let $q_0\in \Rset^d$ and $h>0$. Denote by $q_T:\, p\in \Rset^d \mapsto\operatorname{proj}_1 \Phiverlet[h][T](q_0, p)$.
      By a strong reccurence on $T\in \Nset^*$, we prove \eqref{eq:q_T_affine}, and then extends the results by symmetric properties of leapfrogs transition.

      \textbf{Case $T=1$.} For any $p_0\in \Rset^d\setminus \{0_d\}$,
      \begin{equation}
         \frac{|q_1(p_0)-hp_0|}{|p_0|}\leq \frac{|q_0|+h^2 |\nabla U(q_0)|}{|p_0|}\eqsp ,
      \end{equation}
      setting $\epsilon_{h,1,q_0}^1(|p_0|)=(|q_0|+h^2 |\nabla U(q_0)|)/|p_0|$ conclude the proof for $T=1$.
      Let $T\in \Nset^*\setminus \{1\}$ and assume \eqref{eq:q_T_affine} for $1$ to $T-1$.
      By \eqref{eq:phiver_q_develop}, for any $p_0 \in \Rset^d\setminus \{0_d\}$,
      \begin{equation}
         \frac{|q_T(p_0)-hTp_0|}{|p_0|}\leq \frac{|q_0|+h^2 T \sum_{i=1}^{T-1} |\nabla U(q_T(p_0))|}{|p_0|} \eqsp .
      \end{equation}
      Using the inductions hypothesis, we have for any $i\in [T-1]$, $ \frac{|q_i(p_0)|}{ih |p_0|}\to 1 $ as $|p_0| \to \infty $.
      Then, \Cref{lemma:bounding_gradient} completes the proof.

      \ref{lemma:item1_affine_qT} holds for $T\in \Zset$, indeed,
      denoting by $\Phiverlet[h][-T](q_0,p_0)=(q_{-T},p_{-T})$ and $ \Phiverlet[h][T](q_0,-p_0)=q_{T}',p_{T}' $ for any $T\in \Nset^* $ and $q_0,p_0\in (\Rset^d)^2 $,
    applying the expression of the leapfrog scheme \eqref{eq:iteration_verlet} yields:
    \begin{equation}
        q_{-T}=q_{T}'\,,\, p_{-T}=-p_{-T}'\eqsp .
    \end{equation}
    This completes the proof of \ref{lemma:item1_affine_qT}.

    We prove \ref{lemma:item2_affine_pT} as an application of \ref{lemma:item1_affine_qT} and \Cref{lemma:bounding_gradient}.
    By \eqref{eq:phiver_p_develop}, for any $p_0\in \Rset^d\setminus \{0\}$, 
    \begin{equation}
      \frac{|p_T(p_0)-p_0|}{|p_0|}\leq h\frac{\sum_{i=0}^{T} |\nabla U(q_T(p_0))|}{|p_0|} \eqsp .
    \end{equation}
    By \ref{lemma:item1_affine_qT}, for any $i\in [T-1]$, $ \frac{|q_i(p_0)|}{ih |p_0|}\to 1 $ as $|p_0| \to \infty $ and then applying \Cref{lemma:bounding_gradient} inside the sum 
    yields \ref{lemma:item2_affine_pT}.

    We prove now \ref{lemma:item3_small_df}. Using \eqref{eq:df_qq} with $T_2=T\in \Zset^*$ and $T_1=0$, we have for any $p_0\in \Rset^d$, $\dd q_T(p_0)=hT (\Idd-B(p_0))$ and $B$ is defined in \eqref{eq:B_def}.
    By \eqref{eq:B_def} and \eqref{eq:theta_1} and using that $q_i$ is Lipschitz with Lipschitz constant $L_i$ for any $i\in [T-1]$ by \cite[Lemma 17]{Durmus2017-tf}, we have for any $p_0\in \Rset^d$,
    \begin{equation}
      \label{eq:bounding_B}
      |B(p_0)|\leq h \sum_{i=1}^{T-1} \|\nabla^2 U(q_i(p_0))\| L_i \eqsp .
    \end{equation}
    Applying \eqref{eq:q_T_affine} for any $i\in [T-1]$, we have $|q_i(p_0)|\to \infty$ as $|p_0|\to \infty$.
     Therefore, \eqref{eq:bounding_B} and $\lim_{|q| \to \infty} \| \nabla^2 U(q) \| = 0$ complete the proof.

   \end{proof}

    \begin{lemma}
    \label{lemma:analytic_dichotomy}
    Suppose the potential $U: \Rset^d \to \Rset$ is real analytic.
    Let $h > 0$ and $T_1 \neq T_2 \in \mathbb{Z}$ be fixed and denote by $(q_T, p_T) = \Phiverlet[h][T](q_0, p_0)$ for any $(q_0, p_0) \in (\Rset^d)^2$ and $T \in \mathbb{Z}$,
     $F^{T_1,T_1}_{q_0}(p_0)= p_{T_1}^\top (q_{T_2} - q_{T_1})$ for any $q_0,p_0\in (\Rset^d)^2$.
    Depending on the potential $U$, one of the following statements holds.
    \begin{enumerate}[wide, labelwidth=!, labelindent=0pt, label=(\alph*),leftmargin=*]
       \item 
       \label{lemma:dicho_item1}
           $
               \Leb_{2d}\left( \left\{ (q_0, p_0) \in (\Rset^d)^2: F^{T_1,T_1}_{q_0}(p_0) = 0 \right\} \right) = 0.
           $
    
       \item $\operatorname{proj}_1 \Phiverlet[h][T_2-T_1](q, 0) = q$ for all $q \in \Rset^d$.
    \end{enumerate}
    Moreover if $\|\nabla^2 U(q)\| \to 0$ as $|q| \to \infty$ , \ref{lemma:dicho_item1} holds and in fact we obtain the stronger statement that for any $q_0 \in \Rset^d$ the set $\msf_{q_0,-0}^{T_1,T_2} = \Rset^d \setminus (F^{T_1,T_1}_{q_0}(p_0))^{-1}(\{0_d\})$ is dense and open.
    \end{lemma}
    In words, in the real analytic case either the boundary perturbation in \Cref{hyp:3}$(h,\Kmax)$-\ref{hyp:item_i_topo} is possible for almost all $(q_0, p_0) \in (\Rset^d)^2$ or the discretized Hamiltonian dynamics are degenerate in the sense that a $T_2-T_1$-step integration with zero initial momentum leaves every $q \in \Rset^d$ fixed.
    Note that with $h^2 = 2 \left( 1 - \cos \uppi/[2(T_2-T_1)] \right)$ and standard Gaussian target, all $(q_0, p_0) \in (\Rset^d)^2$ satisfy $(q_{T_2-T_1}, p_{T_2-T_1}) = (q_0, p_0)$, so nonzero potentials may have degenerate dynamics in this sense.
    The growth condition on the Hessian rules out the degeneracy.
    
    \begin{proof}
      As a prelimenary remark, the gradient of an analytic function is analytic and the space of analytic function is a vector space, thus $\nabla U $ and $\Phiverlet[h][T]$ are real analytic for any $T\in \Zset$.
      Moreover, if an analytic function vanishes on an open set $\msv\subset \Rset^d$, then by analytic continuation, the analytic function vanishes on all the space $\Rset^d$ (see \cite[Corollary 1.2.6]{krantz2002primer}).

    Suppose there exists an open set $\mathsf{V} \subset (\Rset^d)^2$ such that $p_{T_1}^\top (q_{T_2}-q_{T_1}) = 0$ for any $(q_0, p_0) \in \mathrm{V}$.
    By the definition of the leapfrog integrator, $\Phiverlet[h][1]$ and therefore $\Phiverlet[h][T_1]$ is an homeomorphism $(\Rset^d)^2 \to (\Rset^d)^2$.
    It follows that $\Phiverlet[h][T_1](\mathsf{V}) \subset (\Rset^d)^2$ is open and therefore for any $\bar{q}_{T_1} \in \operatorname{proj}_1 \Phiverlet[h][T_1](\mathrm{V})$, there exists an open neighborhood $\mathsf{W}_{\bar{q}_{T_1}} \subset \Rset^d$ such that for any $\bar{p}_{T_1}\in \mathsf{W}_{\bar{q}_{T_1}} $, 
    \begin{equation}
      F_{\bar{q}_{T_1}}^{T_2-T_1}(\bar{p}_{T_1}) = \bar{p}_{T_1}^\top (\tilde{q}_{T_2-T_1}(\bar{p}_{T_1}) - \bar{q}_{T_1}) = 0\eqsp ,
    \end{equation}
     where we have denoted by $\tilde{q}_{T_2-T_1}(\bar{p}_{T_1}) = \operatorname{proj}_1 \Phiverlet[h][T_2-T_1](\bar{q}_{T_1}, \bar{p}_{T_1})$.
    Thus for any $\bar{q}_{T_1} \in \operatorname{proj}_1 \Phiverlet[h][T_1](\mathsf{V})$ the real analytic function $\bar{p}_{T_1}\in \Rset^d \mapsto F_{\bar{q}_{T_1}}^{T_2-T_1}(\bar{p}_{T_1})$ vanishes on an open set $\mathsf{W}_{\bar{q}_{T_1}}$ and therefore, by analytic continuation, everywhere on $\Rset^d$.
    In particular, for any $\bar{q}_{T_1} \in \operatorname{proj}_1 \Phiverlet[h][T_1](\mathsf{V})$ and $\tilde{p} \in \Rset^d$,
    $$
       \nabla F_{\bar{q}_{T_1}}^{T_2-T_1}(\tilde{p}) = \tilde{q}_{T_2-T_1}(\tilde{p}) - \bar{q}_{T_1} + \left( \dd \tilde{q}_{T_2-T_1}(\tilde{p}) \right)^\top \tilde{p} = 0 \eqsp ,
    $$
    and by setting $\tilde{p} = 0$, we obtain $\operatorname{proj}_1 \Phiverlet[h][T_2-T_1](\bar{q}_{T_1}, 0) = \bar{q}_{T_1}$ for any $\bar{q}_{T_1} \in \operatorname{proj}_1 \Phiverlet[h][T_1](\mathsf{V})$.
    Noting that $\operatorname{proj}_1 \Phiverlet[h][T_1](\mathsf{V}) \subset \Rset^d$ is an open set and that we have shown that the real analytic mapping $q\in \Rset^d \mapsto \operatorname{proj}_1 \Phiverlet[h][T_2-T_1](q, 0) - q$ vanishes on $\operatorname{proj}_1 \Phiverlet[h][T_1](\mathsf{V})$, by analytic continuation we have in fact shown that $q\in \Rset^d \mapsto \operatorname{proj}_1 \Phiverlet[h][T_2-T_1](q, 0) - q$ vanishes for any $q \in \Rset^d$.
    In summary, if an open set $\mathsf{V} \subset (\Rset^d)^2$ as above is supposed to exist, the second option of the lemma holds.
    
    Then suppose that no such open set $\mathsf{V}$ exists, i.e. for every open set $\mathsf{V}' \subset (\Rset^d)^2$ we suppose that there exists $(q_0, p_0) \in \mathsf{V}'$ such that $p_{T_2}^\top (q_{T_2} - q_{T_1}) \neq 0$.
    In this case, as a nonvanishing real analytic function $(q_0, p_0)\in (\Rset^d)^2 \mapsto p_{T_1}^\top (q_{T_2} - q_{T_1})$ has a zero set of Lebesgue measure zero since all the zeros are isolated by \cite[Corollary 1.2.7]{krantz2002primer}, showing that the first option of the lemma holds.

    To prove the final statement, assume that $\|\nabla^2 U(q)\| \to 0$ as $|q| \to \infty$.
    Fix $q_0 \in \Rset^d$ and denote by $q_T:\, p\in \Rset^d \operatorname{proj}_1 \Phiverlet[h][T](q_0, p)$ and $p_T:\, p\in \Rset^d \operatorname{proj}_2 \Phiverlet[h][T](q_0, p)$.
    Let $T_1,T_2 \in \Zset^2$ such that $T_1\neq T_2$, we show that $(F^{T_1,T_2}_{q_0})^{-1}(\{0\}) = \big\{ p_0 \in \Rset^d: p_{T_1}^\top (q_{T_2} - q_{T_1}) = 0 \big\}$ has an empty interior by contradiction, note we have already that $\Rset^d\setminus \Delta_{q_0}$ is open by continuity of $F^{T_1,T_2}_{q_0}:p_0\in \Rset^d \mapsto p_{T_1}^\top (q_{T_2} - q_{T_1})$ under \Cref{hyp:regularity}. 
    If $ F^{T_1,T_2}_{q_0}$ vanishes on an open set $\msv$, then by analytic continuation we have $ F^{T_1,T_1}_{q_0}(p)=0$ for any $p\in \Rset^d$.
    However, we have by \Cref{lemma:identity_far}-\ref{lemma:item1_affine_qT},\ref{lemma:item2_affine_pT}, for any $p_0\in \Rset^d$,
    $$
       \lim_{|p_0|\to \infty}\left|p_{T_1}(p_0)^\top (q_{T_2} - q_{T_1})(p_0)
       -
       (T_2-T_1) h \, p_0^\top  p_0 \right|/|p_0|^2=0
    $$
    which in particular implies that $p_{T_1}^\top (q_{T_2} - q_{T_1}) \neq 0$ for some $p_0$ with $|p_0|$ large enough, which contradicts that $ F^{T_1,T_2}_{q_0}$ vanishes on an open set .
    This proves that $(F^{T_1,T_2}_{q_0})^{-1}(\{0\})$ has an empty interior.
    Since $q_0 \in \Rset^d$ is arbitrary, the proof is complete.

    \end{proof}

    \begin{lemma}
    \label{lemma:analytic_local_homeo}
    Assume \Cref{hyp:analytic_potential}.
    Then, \Cref{hyp:3}$(h,\Kmax)$-\ref{hyp:item_ii_homeo} is satisfied.
    
    \end{lemma}
    
    \begin{proof}
    Let $q_0 \in \Rset^d$ and $T\in [-2^{\Kmax}+1:2^{\Kmax}-1]\setminus \{0\}$ be arbitrary and denote by $q_T:\, p\in \Rset^d \mapsto \operatorname{proj}_1 \Phiverlet[h][T](q_0, p)$.
    Using \Cref{lemma:identity_far}-\ref{lemma:item3_small_df} and the set of invertible matrices is open, there exists $R_T>0$ such that for any $p_0\in \Rset^d $ such that $|p_0|\geq R_T$, $\dd q_T(p_0) $ is an isomorphism. 
    Therefore, using that $q_T$ is $\mathrm{C}^1$ under \Cref{hyp:regularity}, by the local inversion Theorem, we deduce that $q_T$ is a local $\mathrm{C}^1$ diffeomorphism on $\Rset^d \setminus \bar{\mathrm{B}}(0_d,R_T+1)$.
    Denoting by $R= \max_{T\in [-2^{\Kmax}+1:2^{\Kmax}-1]\setminus \{0\}} (R_T+1) $, let $ p_0\in \Rset^d \setminus \bar{\mathrm{B}}(0_d,R)$, for any $T\in [-2^{\Kmax}+1:2^{\Kmax}-1]\setminus \{0\}$,
     there exists $r_T>0$ such that $q_T|_{\mathrm{B}(p_0,r_T)} $ is a local homemorphism, then setting $r_H=\min_{T\in [-2^{\Kmax}+1:2^{\Kmax}-1]\setminus \{0\}} r_T $ conclude the proof.

    \end{proof}

    To show that \Cref{hyp:analytic_potential} implies \Cref{hyp:3}$(h,\Kmax)$-\ref{hyp:item_i_topo}, let $q \in \Rset^d$ be arbitrary.
    By \Cref{lemma:analytic_dichotomy} the set $\msf_{q_0,-0}^{T_1,T_2}$ is dense and open for any $T_1,T_2 \in \Zset^2$ with $T_1\neq T_2$.
    Hence the union over $T_1,T_2 \in \Zset^2$ with $T_1\neq T_2$ is also dense and open, since a countable intersection of open and dense sets is still open and dense by the Baire's theorem ($\Rset^d$ is complete).
     This yields \Cref{hyp:3}$(h,\Kmax)$-\ref{hyp:item_i_topo}.
    
    \Cref{lemma:analytic_local_homeo} directly imply that \Cref{hyp:3}$(h,\Kmax)$-\ref{hyp:item_ii_homeo} follows from \Cref{hyp:analytic_potential}, finishing the proof.
  
\subsubsection{Proof of \Cref{prop:general_condition_check}-\ref{prop:condition_check_gaussian}}
Using \Cref{lemma:gaussian_case}, there exists a countable set $\msh_0\subset \Rset_{\geq0}$ such that for any $h\in \Rset_{>0}\setminus \msh_0$,
 $ \psi_{q_0}$ is a linear one-to-one map for any $q_0\in \Rset^d$ and thus is an homeomorphism since the dimension is finite, hence \Cref{hyp:3}$(h,\Kmax)$-\ref{hyp:item_ii_homeo}.
Recalling that for any $q\in \Rset^d$ (see \Cref{lemma:densitygradient_imply})
\begin{equation}
\nabla\msf_{q,-0}=\{p\in \Rset^d:\, \nabla F_q^{T_1,T_2}(p)\neq 0,\, T_1,T_2\in[-2^{\Kmax}+1:2^{\Kmax}-1]^2,\, T_1\neq T_2 \}\eqsp ,
\end{equation}
we have for any $q_0\in \Rset^d$,
\begin{equation}
   \nabla\msf_{q,-0} \subset \Rset^d \setminus \bigcup_{T_1,T_2\in\Zset^2,\, T_1\neq T_2} (\nabla F_{q_0}^{(T_1,T_2)})^{-1}(\{0_d \}) \eqsp.
\end{equation}
$\bigcup_{T_1,T_2\in\Zset^2,\, T_1\neq T_2} (\nabla F_{q_0}^{(T_1,T_2)})^{-1}(\{0_d \})$ is countable since $\nabla F_{q_0}^{(T_1,T_2)}$ is a one-to-one map by \Cref{lemma:gaussian_case}
 and thus \Cref{hyp:3}$(h,\Kmax)$-\ref{hyp:item_i_topo} follows by \Cref{lemma:densitygradient_imply}.

\section{Proofs of \Cref{section:ergo_geo}}

\subsection{Proof of \Cref{lemma:sketchofproof}}
     \label{sec_supp:sketchofproof}

  
  Let $h>0$. Let $R_0>0$ be
  such that \eqref{eq:norm_position_condition},\eqref{eq:lemma_hamiltonian_condition} apply. 
  The construction of $\rmq_h$ (see \Cref{lemma:definition_continuity_q}) implies that for any $(q_0,p_0)\in (\Rset^d)^2$ and $\msj \subset \mathcal{P}(\Zset)$ such that $\rmp_h(\msj|q_0,p_0)>0 $
   we have $\rmq_h(0 \mid \msj, q_0, p_0)=0 $ as soon as $H(\Phiverlet[h][j](q_0,p_0))-H(q_0,p_0)\leq 0$ for $j\in \{-1,1\}$, since it implies $1\wedge \tpi(\Phiverlet[h][j](q_0,p_0))/\tpi(q_0,p_0)=1 $ for $j\in \{-1,1\}$.
  Let $q_0\in \Rset^d $ satisfy $|q_0|\geq R_0$ and decompose $\KkerU_{h}\VFL_a(q_0)$ to isolate $\rmq_h(0 \mid \msj, q_0, p_0)$:
  $$ \KkerU_{h}\VFL_a(q_0) = \int  G(q_0,p_0)\dd p_0\,+K(q_0)\eqsp ,$$
  where for any $p_0\in \Rset^d$ we denote
  \begin{align}
  G(q_0,p_0)&=
  \begin{multlined}[t]
      \GaussStandard(p_0)\times \\
      \sum_{\msj \subset \zset} \sum_{j \in \msj,j\neq 0}  \rmp_h(\msj \mid q_0, p_0) \rmq_h(j \mid \msj, q_0, p_0) \VFL_a(\operatorname{proj}_1(\Phiverlet[h][j](q_0, p_0))\eqsp ,
  \end{multlined} \\
  K(q_0)&=\int \dd p_0 \, \GaussStandard(p_0)\sum_{\msj \subset \zset} \rmp_h(\msj \mid q_0, p_0) \rmq_h(0 \mid \msj, q_0, p_0) \VFL_a(q_0) \eqsp .
  \end{align}
  Using that $\rmq_h(0 \mid \msj, q_0, p_0)=0 $ when $p_0\in \mathrm{B}(q_0)$ we have
  \begin{equation}
      K(q_0)=\VFL_a(q_0)\int_{\mathrm{B}(q_0)^\complementary} \dd p_0\GaussStandard(p_0) \underbrace{\sum_{\msj \subset \zset} \rmp_h(\msj \mid q_0, p_0) \rmq_h(0 \mid \msj, q_0, p_0)}_{\leq 1}\eqsp .
  \end{equation}
  It follows that
  \begin{equation}
      \label{eq:limit1_sketch1}
   \frac{|K(q_0)|}{\VFL_a(q_0)}\leq \int_{\mathrm{B}(q_0)^\complementary} \dd p_0\GaussStandard(p_0) \to 0\ee \text{as}\ee |q_0| \to \infty  
  \end{equation}
  by applying the Lebesgue dominated convergence theorem.
  Next, we aim to decompose the integral of $p_0 \in \Rset^d \mapsto G(q_0,p_0)$ according to $\mathrm{B}(q_0)$ as well.
  To this end, we bound $|\operatorname{proj}_1(\Phiverlet[h][j](q_0, p_0))|=|q_j|$ for $j\in [-2^{\Kmax} +1:2^{\Kmax}-1]$. Recall that
  \begin{equation}
      q_j= q_0+jhp_0-j\frac{h^2}{2}\nabla U(q_0)-h^2 \Theta_{h,j}^{(1)}(q_0,p_0)\eqsp,
  \end{equation}
  where $\Theta_{h,j}^{(1)}$ is defined in \eqref{eq:theta_1}.
  Applying \cite[Lemma S3]{Durmus2017-tf} and using \Cref{hyp:rappel}($m$)-\ref{hyp:rappel:item_tailgrad}, there exist $c_1,c_2>0$ such that for any $q,p\in (\Rset^d)^2$,
  \begin{align}
      &|\Theta_{h,j}^{(1)}(q, p)|\leq |\Theta_{h,j}^{(1)}(q, 0)|+|\Theta_{h,j}^{(1)}(q, p)-\Theta_{h,j}^{(1)}(q, 0)| \eqsp , \\
      &|\Theta_{h,j}^{(1)}(q, p)-\Theta_{h,j}^{(1)}(q, 0)|\leq c_1|p|,\ee |\Theta_{h,j}^{(1)}(q, 0)|\leq c_2 (|q|^{m-1}+1 +|\nabla U(q)|)\eqsp,
  \end{align}
  where $c_1$ and $c_2$ depend on the constants $h > 0$ and $\Kmax \in \Nset^*$.
  Since we have $|\nabla U(q)|\leq \msm_1 (1+|q|^{m-1})  $ and $|q|^{m-1}=(|q|^{\gamma})^{(m-1)/\gamma}\leq |p|^{(m-1)/\gamma}$ for any $p_0\in \mathrm{B}(q_0)^\complementary$, there exists $c_2'>0$ (provided that $R_0>1$ which can be assumed) such that
  \begin{equation}
      |\Theta_{h,j}^{(1)}(q_0, p_0)|\leq c_1' (1+|p_0|^{\frac{m-1}{\gamma}})\eqsp
  \end{equation}
  for any $p_0\in \mathrm{B}(q_0)^\complementary$.
  This implies, with $ (m-1)/\gamma\geq 1$, the existence of $c_2'>0$
  such that, for any $p_0\in \mathrm{B}(q_0)^\complementary$,
  \begin{equation} 
      |q_j|\leq c_2'(1+|p_0|^{\frac{m-1}{\gamma}})+|q_0|\eqsp . 
  \end{equation}
  Denoting by $\phi:p_0\in \Rset^d  \mapsto \exp(a c_2'(1+|p_0|^{(m-1)/\gamma}))\GaussStandard(p_0)$,
  then for any $p_0\in \mathrm{B}(q_0)$,
  we have $\VFL_a(q_j)\GaussStandard(p_0)\leq \phi(p_0) \VFL_a(q_0)$ and $ G(q_0,p_0)\leq \phi(p_0)\VFL_a(q_0) $.
  Therefore,
  \begin{equation}
      \label{eq:limit2_sketch2}
   \int_{\mathrm{B}(q_0)^\complementary} \frac{G(q_0,p_0)}{\VFL_a(q_0)} \dd p_0 \leq\int_{\mathrm{B}(q_0)^\complementary} \phi(p_0)\dd p_0 \to 0 \quad\text{as} \ee \ee|q_0|\to \infty\eqsp ,
  \end{equation}
  applying the Lebesgue dominated convergence theorem with $\phi$ which is integrable on $\Rset^d$ because $(m-1)/\gamma<2$.
  Condition \eqref{eq:norm_position_condition} yields
   \begin{equation}
      \label{eq:major_sketch3}
  \int_{\mathrm{B}(q_0)}\frac{G(q_0,p_0)}{\VFL_a(q_0)} \dd p_0 \leq \exp (a \times -1) \int_{\mathrm{B}(q_0)} \GaussStandard(p_0) \dd p_0\leq \exp(-a)\eqsp . 
   \end{equation}
   Then, combining \eqref{eq:limit1_sketch1}, \eqref{eq:limit2_sketch2} and \eqref{eq:major_sketch3},  for any $q_0\in\Rset^d$ such that $|q_0|\geq R_0$, we get
  \begin{equation}
       \KkerU_{h}\VFL_a(q_0) \leq \exp(-a) \VFL_a(q_0)+W(q_0)\eqsp ,
  \end{equation}
  where $W(q_0)/\VFL_a(q_0)\to 0 $ as $|q_0|\to \infty$.
   Therefore, there exists $R_0'>0$ such that for any $q_0\in \Rset^d$ with  $|q_0|\geq R_0'$ we have $W(q_0)/\VFL_a(q_0) \leq \frac{1-\exp(-a)}{2} $ and so there exists $\lambda \in (0,1)$ such that for any $q_0\in \Rset^d$ with $|q_0|\geq R_0'$
  \begin{equation}
      \KkerU_{h}\VFL_a(q_0) \leq \lambda \VFL_a(q_0)\eqsp . 
  \end{equation}
  Setting $R'=R_0'+1$ completes the proof.

  \subsection{Proof of \Cref{prop:degrowth_energy}}
  \label{sec_supp:degrowth_energy}
  
    The case $j=1$ follows from \cite[Proposition 7]{Durmus2017-tf} under \Cref{hyp:rappel}($m$) and \cite[Proposition 9]{Durmus2017-tf} under \Cref{hyp:gaussian_perturbation}. For $j'=-1=-j$ we denote by
    $\Phiverlet[h][-1](q_0,p_0)=q_{-1},p_{-1}$ and $ \Phiverlet[h][1](q_0,-p_0)=q_{1}',p_{1}' $.
    Applying the expression of the leapfrog scheme \eqref{eq:iteration_verlet} yields
    \begin{equation}
        q_{-1}=q_{1}'\,,\ee p_{-1}=-p_{-1}'\eqsp ,
    \end{equation}
    and then
    \begin{equation} H(q_{-1},p_{-1})=H(q_{1}',-p_{1}')=H(q_{1}',p_{1}')\eqsp .
    \end{equation}
    If $q_0,p_0$ satisfy that $|q_0|\geq R_H$ and $|p_0|\leq |q_0|^\gamma$, then it is the same for $q_0,-p_0$ and applying \cite[Proposition 7 or 8]{Durmus2017-tf}, we have $H(\Phiverlet[h][1](q_0,-p_0))-H(q_0,p_0)\leq 0$ and therefore $H(\Phiverlet[h][-1](q_0,p_0))-H(q_0,p_0)\leq 0$.

\subsection{Proof of \Cref{lemma:normposition}}
\label{sec_supp:norm_position}

  Let $h>0$.
  This proof adapts and follows the arguments in the proof of \cite[Proposition 5]{Durmus2017-tf}.
  Let $(q_0,p_0)\in (\Rset^d)^2$ such that $|p_0|\leq |q_0|^\gamma$ and $T\in \Nset^*$.
   We denote by $q_j=\operatorname{proj}_1\Phiverlet[h][j](q_0, p_0)$ for any $j\in[-T:T]$ and
  we write 
  \begin{equation}
      |q_j|^2=|q_0|^2+A_{h,j}^{(1)}-2h^2 A_{h,j}^{(2)} \eqsp ,
  \end{equation}
  where
  \begin{align}
      &A_{h,j}^{(1)}=2jh q_0^\transpose p_0+\Big|jhp_0-(jh^2/2)\nabla U(q_0)-h^2\sum_{i=1}^{j-1} (j-i)\nabla U(q_i)\Big|^2\eqsp , \\
      &A_{h,j}^{(2)}=q_0^\transpose \Big[(j/2)\nabla U(q_0)+\sum_{i=1}^{j-1} (j-i)\nabla U(q_i)\Big]\eqsp . 
  \end{align}

  We first assume $m<2$.
  We aim to bound $|A_{h,j}^{(1)}|,\,A_{h,j}^{(2)}$ from below and from above respectively.
  By the Cauchy--Schwartz inequality,
  \begin{equation}
      |q_0^\transpose p_0|\leq |q_0|^{\gamma+1}\eqsp  
  \end{equation}
and by using \cite[Lemma S2-(i)]{Durmus2017-tf} with \Cref{hyp:rappel}($m$)-\ref{hyp:rappel:item_tailgrad}, there exists a constant $C_0>0$ (depending on $h$ and $T$) such that
\begin{equation}
  \label{nablaU}
  |\nabla U(q_i) |\leq C_0(1+|q_0|^{m-1}+|p_0|)\leq C_0(1+2|q_0|^{m-1})\eqsp .
\end{equation}
  Using $2(m-1)<\gamma+1 $ when $m<2$, this implies that
  \begin{equation}
      |A_{h,j}^{(1)}|\leq C_1(1+|q_0|^{\gamma+1} )\eqsp
  \end{equation}
  for some constant $C_1> 0$.
   On the other hand, $A_{h,j}^{(2)}=A_{h,j}^{(2,1)}+A_{h,j}^{(2,2)}$ with
  \begin{align}
      &A_{h,j}^{(2,1)}=\frac{j}{2} q_0^\transpose \nabla U(q_0)+\sum_{i=1}^{j-1} (j-i)q_i^\transpose \nabla U(q_i)\eqsp ,\\
      &A_{h,j}^{(2,2)}=-\sum_{i=1}^{j-1} (j-i)(q_0-q_i)^\transpose \nabla U(q_i)\eqsp .
  \end{align}
  Under \Cref{hyp:rappel}($m$)-\ref{hyp:rappel:item_rappel},
  \begin{equation}
      A_{h,j}^{(2,1)}\geq A_1 \frac{j}{2}|q_0|^m-\frac{j(j-1)}{2}A_2\eqsp . 
  \end{equation}
  Further by \eqref{nablaU} and \cite[Lemma S2-(i)]{Durmus2017-tf}, there exists $C_2>0$, such that
  \begin{equation}
       |A_{h,j}^{(2,2)}|\leq C_2(1+|p_0|^2+|q_0|^{2(m-1)})\leq C_2(1+2|q_0|^{2(m-1)})\eqsp .
  \end{equation}
  Combining the bounds on $A_{h,j}^{(2,1)}$ and $A_{h,j}^{(2,2)}$, there exists $C_3>0$ such that
  \begin{equation}
      A_{h,j}^{(2)}\geq \frac{j}{2} A_1 |q_0|^m-C_3(1+|q_0|^{2(m-1)})\eqsp . 
  \end{equation}
  Using the fact that $2(m-1)\leq \gamma +1 $ when $m<2$, with the bounds on $A_{h,j}^{(2)}$ and $A_{h,j}^{(1)}$, we have 
  \begin{equation}
      A_{h,j}^{(1)}-2h^2 A_{h,j}^{(2)}\leq c_0 (|q_0|^{1+\gamma}+1)-c_1|q_0|^m \eqsp ,
  \end{equation}
  where $c_0,c_1>0$ are some constants depending on $h$ and $T$.
  Since $1+\gamma<m$, there exists $R_0>0$ such that for any $q_0\in \Rset^d$ with 
  $|q_0|\geq R_0 $, we have
  \begin{equation}
       |\operatorname{proj}_1\Phiverlet[h][j](q_0, p_0)|-|q_0|\leq -1 \eqsp .
  \end{equation}

  We then consider the case $m=2$. 
  The arguments follow the same lines as the proof of \cite[Proposition 5]{Durmus2017-tf} by noticing that $(|q_0| |p_0|+|p_0|)/|q_0|^2\to 0$ as $|q_0|\to \infty$ since $|p_0|\leq |q_0|^{2/3}$.

  The same result holds for $j'=-j$ since 
  \begin{equation}
      \operatorname{proj}_1\Phiverlet[h][-j](q_0, p_0)=\operatorname{proj}_1\Phiverlet[h][j](q_0, -p_0)\eqsp . 
  \end{equation}
 Choosing $R_0$ large enough such that the desired results is given for any $j\in[-T,T]$ with $T\neq 0$ completes the proof.

\subsection{Proof of \Cref{thm:ergo_geo}}
\label{sec_supp:thm_geo_ergo}
    
  Noticing that \Cref{hyp:gaussian_perturbation} implies \Cref{hyp:regularity},
  using assumptions (\Cref{hyp:regularity}, \Cref{hyp:3}$(h,\Kmax)$) or (\Cref{hyp:gaussian_perturbation}, \Cref{hyp:3}$(h,\Kmax)$), the ergodicity of $\KkerU_h$ is given by \Cref{thm:ergodic_general} and it remains to show the Foster--Lyapunov condition \eqref{Foyster-Lyapunov-condtion}.
   We write our proof under the assumptions \Cref{hyp:regularity}, \Cref{hyp:rappel}(m) as the case \Cref{hyp:gaussian_perturbation} is similar.
   Let $m\in (1,2]$ such that \Cref{hyp:rappel}($m$) holds, let $\gamma\in(\max(2(m-1)-1,(m-1)/2),m-1)$ and let $h>0$ such that $(\ltt_1 h^2)/2<1$.
  \begin{itemize}
      \item First we address \ref{thm:last_b}, i.e. assume $m=2$. 
  Applying \Cref{prop:degrowth_energy} using \Cref{hyp:regularity}, \Cref{hyp:rappel}($2$), there exist $\bar{S}>0$ and $R_H\in \Rset_{\geq0}$ such that \eqref{eq:lemma_hamiltonian_condition} holds when $h 2^{\Kmax} \leq \bar{S}$.
  Then, using \Cref{hyp:regularity}, \Cref{hyp:rappel}($2$)-\ref{hyp:rappel:item_rappel} and applying \Cref{lemma:normposition} with $T=2^{\Kmax}$ yields $\bar{S}>0$ and $R_0>0$ so that \eqref{eq:norm_position_condition} holds when $h 2^{\Kmax} \leq \bar{S}$.
  By decreasing $\bar{S}$ if necessary we may assume that both conclusions hold.
  For $h>0$ such that $h2^{\Kmax}\leq \bar{S}$ and $q_0\in \Rset^d$ such that $|q_0|\geq \max(R_H,R_0) $ and using assumptions \Cref{hyp:regularity}, \Cref{hyp:rappel}($2$)-\ref{hyp:rappel:item_tailgrad}, we may apply \Cref{lemma:sketchofproof} so there exists $R'>0$ for which
  \begin{equation}
      \label{eq:pre_drift_last}
      \KkerU_{h}\VFL_a \leq\lambda  \VFL_a +b \mathbbm{1}_{\bar{\mathrm{B}}(0,R')} \eqsp .
  \end{equation}
  By \Cref{thm:3small}, $\bar{\mathrm{B}}(0,R')$ is a small set so the drift condition on $\KkerU_h$ holds.
  Therefore, the NUTS kernel $\KkerU_h$ is $\VFL_a$-uniformly geometrically ergodic. 
      \item Proof of \ref{thm:last_a} follows the same lines but there is no constraint on $h > 0$. 
  \end{itemize}

  \section{Proof of \Cref{section:general_HMC}}
  \subsection{Proof of \Cref{thm:trajectory-transitivity}}
  \label{sec_supp:transitivity_thm}
  Let $q,\tilde{q} \in(\rset^d)^2$.  
    We first show that there exists a solution $(p, \tilde{p}) \in (\Rset^d)^2$ to the equation
    \begin{equation}
        \label{eq:transitivity-momentum-equation}
        \Phiverlet[h][T](q, p) = (\tilde{q}, \tilde{p})\eqsp .
    \end{equation}
    Instead of looking for the solution in terms of the $(p, \tilde{p})$ it turns out to be more convenient to eliminate the momenta from the equations and instead look for a leapfrog trajectory $(q_0, q_1, \dots, q_{T-1}, q_T)$ with $q_0 = q$ and $q_T = \tilde{q}$ satisfying
for any
    $j \in[T-1]$ 
    \begin{equation}
        \label{eq:fixed-point}
        q_j = \frac{q_{j-1} + q_{j+1}}{2} + \frac{h^2}{2} \nabla U(q_j) 
        \eqsp .
    \end{equation}
    Indeed, if $(q_0, q_1, \dots, q_T)$ satisfies \eqref{eq:fixed-point}, we can define for $i \in [T]$     \begin{equation}
      \label{eq:revover_p_i}
      p_i = p_{i-1} -(h/2)[\nabla U(q_{i-1}) +\nabla U(q_i)]  \eqsp, \quad p_0 = h^{-1}(q_1-q_0) +(h/2) \nabla U(q_0) \eqsp,
    \end{equation}
that satisfy
    $\Phiverlet[h][k](q_0, p_0) = (q_k, p_k)$ for $k\in [T]$ and as result $p=p_0$ and $\tilde{p} = p_T$ are solutions of \eqref{eq:transitivity-momentum-equation}.\footnote{
        In Lagrangian mechanics, the dynamics of a system are found by finding a trajectory (i.e. positions and velocities) for which the action (i.e. integral of the Lagrangian function of the system over the trajectory) is stationary.
        The leapfrog integrator may be derived as approximate Lagrangian mechanics by considering only piecewise linear trajectories; the stationarity equations for the action are exactly \eqref{eq:fixed-point}.
        The Lagrangian perspective to mechanics seems particularly appropriate here, even though the equations \eqref{eq:fixed-point} can be easily derived from the usual leapfrog equations in the ad hoc manner presented here.
    }

    In matrix form  \eqref{eq:fixed-point} may be written as
    \begin{equation}
        \label{eq:fixed-point-matrix}
        Q = Q_0/2 + Q A + h^2G(Q)/2 \eqsp ,
    \end{equation}
    where
    $Q = ( q_1 \quad q_2 \quad \dots \quad q_{T-1} ) \in \mathbb{R}^{d \times
      (T-1)}$, $Q_0 = ( q_0 \quad 0 \quad \dots \quad 0 \quad q_T )  \in \mathbb{R}^{d
      \times (T-1)} $,   $G(Q) = (
            \nabla U(q_0) \quad \dots \quad \nabla U(q_{T-1})
     ) \in \mathbb{R}^{d \times (T-1)}$ and $\bfA$ is the tridiagonal matrix given by
    \begin{equation}
      \bfA =
\begin{pmatrix}
0 & 1/2 \\
1/2 &0 & 1/2 \\
&1/2 & \ddots & \ddots \\
& & \ddots & \ddots & 1/2 \\
& & & 1/2 &0
\end{pmatrix} \in \mathbb{R}^{(T-1) \times (T-1)}
\eqsp.
\end{equation}
    we show that the mapping $F_{Q_0}: \mathbb{R}^{d \times (T-1)} \to \mathbb{R}^{d \times (T-1)}$ defined by $F_{Q_0}(Q) = Q_0/2 + Q A + h^2G(Q)/2  $ is contractive in the Frobenius norm $\|\bfB\|_{\Fr}^2 = \operatorname{Tr} \bfB \bfB^\top$ when the condition \eqref{eq:transitivity-condition} holds, which will complete the proof by \eqref{eq:fixed-point-matrix} and  the Banach fixed point theorem.
    First, \cite[Proposition 2.1]{kulkarni1999eigenvalues}  shows that the eigenvalues of the tridiagonal matrix $\bfA$ are given by $-\cos (k \uppi/T)$, for any $k\in [T-1]$ and therefore its operator norm $\|\bfA\|_{\mathrm{op}} = \cos (\uppi/T)$.
    Next, it is clear from the assumption that $q \mapsto U(q)$ is $\ltt_1$-Lipschitz that $\|G(Q) - G(Q')\|_{\Fr} \leq \ltt_1 \|Q - Q'\|_{\Fr}$ for any $Q, Q' \in \mathbb{R}^{d \times (T-1)}$.
    Thus
    \begin{align}
        \|F_{Q_0}(Q) - F_{Q_0}(Q')\|_{\Fr}
        & =
        \| (Q-Q')\bfA + (h^2/2) (G(Q)-G(Q')) \|_{\Fr}
        \\
        & \leq
        \{ \cos(\uppi/T) + \ltt_1 h^2/2\} \|Q-Q'\|_{\Fr} \eqsp ,
    \end{align}
    for any $Q, Q' \in \mathbb{R}^{d \times (T-1)}$.
    This estimate gives contractivity of $F_{Q_0}$ exactly when \eqref{eq:transitivity-condition} holds, and thus implies the existence and unicity of solutions to \eqref{eq:fixed-point-matrix} and as result existence of solution $p,\tilde{p}$ of \eqref{eq:transitivity-momentum-equation}.

    It remains to show that the map $\tilde{q} \mapsto \psi_{q}^{\inv}(\tilde{q})$ is Lipschitz for any $q \in \mathbb{R}^d$.
    For any $Q_0 = \begin{pmatrix} q_0 & 0 & \dots & 0 & q_T \end{pmatrix}$ and $Q_0' = \begin{pmatrix} q_0 & 0 & \dots & 0 & q_T' \end{pmatrix}$, denoting the fixed point of $F_{Q_0}$ by $Q(Q_0)$, we have 
    \begin{align}
        \|Q(Q_0) & - Q(Q_0')\|_{\Fr}
        \\
        & =
        \big\| Q_0 - Q_0' + (Q(Q_0) - Q(Q_0')) \bfA + \fracD{h^2}{2} (G(Q(Q_0)) - G(Q(Q_0'))) \big\|_{\Fr}
        \\
        & \leq
        \| Q_0 - Q_0' \|_{\Fr} + \{ \cos \fracD{\uppi}{T} + \fracD{\ltt_1 h^2}{2} \} \|Q(Q_0)-Q(Q_0')\|_{\Fr} \eqsp ,
    \end{align}
    which implies
    \begin{equation}
        \|Q(Q_0)-Q(Q_0')\|_{\Fr} \leq \frac{\|Q_0-Q_0'\|_{\Fr}}{1 - \{ \cos \fracD{\uppi}{T} + \fracD{\ltt_1 h^2}{2}}= \frac{\|q_T-q_T'\|}{1 - \{ \cos \fracD{\uppi}{T} + \fracD{\ltt_1 h^2}{2}} \eqsp .
      \end{equation}
As a result, we get the desired property of $\psi_q$ since   $\psi_{q_0}^{\inv}(q_T)$ is the last column of $ Q(Q_0)$.

\subsection{Proof of \Cref{lemma:gaussian_case}}
\label{sec_supp:gaussian_case}

Equation~\eqref{eq:iteration_verlet} gives that for any $q_0,p_0\in \Rset^d$:
  \begin{equation}
    \label{eq:gaussian_1_expression}
      \Phiverlet[h][1](q_0,p_0)=((\Idd-\fracD{h^2}{2}\Sigma)q_0+hp_0,(\Idd-\fracD{h}{2}\Sigma)p_0+h(\fracD{h^2}{4}\Sigma^2-\Sigma)q_0)\eqsp . 
  \end{equation}
  Thus by composition and an easy induction, for any $T\in \Zset$ and $h\in \Rset$, there exist $A_\Sigma^{(T)}(h),B_\Sigma^{(T)}(h)\in \symset_d(\Rset)$ such that for any $q_0,p_0\in (\Rset^d)^2$,
  \begin{equation}
      \operatorname{proj}_1\Phiverlet[h][T](q_0,p_0)=A_\Sigma^{(T)}(h)q_0+B_\Sigma^{(T)}(h)p_0\eqsp ,
  \end{equation}
  where the components of $A_\Sigma^{(T)}(h),B_\Sigma^{(T)}(h)$ are polynomial in $h$ and  $\frac{\dd}{\dd h}B_\Sigma^{(T)}(h)|_{h=0}=T\Idd$. 
  Then,  for any $T\in \Zset$, the function $P^{(T)}:\, h\in \Rset \mapsto \det(B_\Sigma^{(T)}(h)) $ is a polynomial in $h$ and is not identically zero.
  Therefore,  for any $T\in \Zset$, the equation $P^{(T)}(h)\, = \,0 $ has a finite number of solutions in $\Rset_{\geq0}$ denoted by  $\msh_0^{(T)}$ and for any $h\in \Rset_{\geq0}\setminus \bigcup_{T\in\Zset }\msh_0^{(T)}$, $q\in \Rset^d$ the function $\psi_q^{(T)} =        \operatorname{proj}_1\Phiverlet[h][T](q,\cdot)$ is a linear one-to-one map.
  
  Applying the same reasoning, by composition of \eqref{eq:gaussian_1_expression},
for any $T\in \Zset$ and $h\in \Rset$, there exist $\tilde{A}_\Sigma^{(T)}(h),\tilde{B}_\Sigma^{(T)}(h)\in \symset_d(\Rset)$ such that for any $q_0,p_0\in (\Rset^d)^2$,
  \begin{equation}
      \operatorname{proj}_2\Phiverlet[h][T](q_0,p_0)=\tilde{A}_\Sigma^{(T)}(h)q_0+\tilde{B}_\Sigma^{(T)}(h)p_0\eqsp ,
  \end{equation}
  where the components of $\tilde{A}_\Sigma^{(T)}(h),\tilde{B}_\Sigma^{(T)}(h)$ are polynomial in $h$.
As a result,  for any $T_1, T_2, q_0\in \Zset^2\times \Rset^d$ with $T_1\neq T_2$ and $h\in \Rset$, defining $p_0\in \Rset^d \mapsto F_{q_0}^{(T_1,T_2)}(p_0)=p_{T_1}^\top (q_{T_2}-q_{T_1})$ is at most quadratic in $q_0,p_0$ and thus there exist $C_\Sigma^{(T_1,T_2)}(h),D_\Sigma^{(T_1,T_2)}(h)\in \mcs(\Rset)$ such that for any $q_0,p_0\in (\Rset^d)^2$,
  \begin{equation}
      \nabla F_{q_0}^{(T_1,T_2)}(p_0)=C_\Sigma^{(T_1,T_2)}(h)q_0+D_\Sigma^{(T_1,T_2)}(h)p_0\eqsp ,
  \end{equation}
  where the coordinates of $C_\Sigma^{(T_1,T_2)}(h),D_\Sigma^{(T_1,T_2)}(h)$ are polynomial in $h$.
  Then, the function $Q^{(T_1,T_2)}:\, h\in \Rset \mapsto \det(D_\Sigma^{(T_1,T_2)}(h)) $ is polynomial in $h$ and is not zero since by \eqref{eq:phiver_q_develop}, \eqref{eq:phiver_p_develop}, \eqref{eq:df_p_0}, \eqref{eq:df_qq}  $\frac{\dd}{\dd h}D_\Sigma^{(T_1,T_2)}(h)|_{h=0}=(T_2-T_1)\Idd$.
  Therefore, the equation $Q^{(T_1,T_2)}(h)\, = \,0 $ has a finite number of solutions in $\Rset_{\geq0}$ denoted by the set $\msh_0^{(T_1,T_2)}$ and for any $h\in \Rset_{\geq0}\setminus \bigcup_{T_1,T_2\in\Zset^2,\, T_1\neq T_2 }\msh_0^{(T_1,T_2)}$, $q\in \Rset^d$ and for any $T_1,T_2\in \Zset^2$ with $ T_1\neq T_2$, the functions $\nabla F_{q}^{(T_1,T_2)}$ are linear one-to-one map. 
  
  Setting 
  \begin{equation}
      \msh_0= \bigcup_{T_1,T_2\in\Zset^2,\, T_1\neq T_2 }\msh_0^{(T_1,T_2)} \cup \bigcup_{T\in\Zset }\msh_0^{(T)}
  \end{equation}
  completes the proof since a countable union of countable sets is countable.
  

  \subsection{Proof of \Cref{thm:ergodicity_compile_HMC}}
  \label{sec_supp:ergo_HMC}

We prove the result under \Cref{hyp:gaussian_perturbation_nice} and then under the bound on the stepsize \eqref{eq:transitivity-condition}.

  We first establish \Cref{lemma:perturbation_point} to bound the distance between the trajectory generated by a Gaussian potential and a Gaussian perturbated potential.
  \Cref{lemma:perturbation_point} is then combined with the homotopy argument proposed in \cite[Proposition 14]{Durmus2017-tf} to derive \Cref{thm:ergodicity_compile_HMC}.

   To state \Cref{lemma:perturbation_point}, we introduce the following notations.
 Let $h\in \Rset_{>0}\setminus \msh_0$ be fixed where $\msh_0$ is defined in \Cref{lemma:gaussian_case}.
  We denote by $c(x_1,x_2,\cdots,x_n)$ a constant depending on the constants $x_1,x_1,\cdots,x_n$.
   For any $T\in \Zset$ and $q_0,p_0\in (\Rset^d)^2$, we denote by
    \begin{equation}
        q_T(q_0,p_0)=\operatorname{proj}_1\Phiverlet[h][T](q_0,p_0),\quad p_{T}(q_0,p_0)=\operatorname{proj}_2\Phiverlet[h][T](q_0,p_0)\eqsp,
    \end{equation}
    where the potential related to $\Phiverlet[h][\cdot]$ is $U_\Sigma$ defined in \Cref{lemma:gaussian_case}, 
    \begin{equation}
        \tilde{q}_T(q_0,p_0)=\operatorname{proj}_1 \tilde{\Phi}_h^T(q_0,p_0),\quad \tilde{p}_{T}(q_0,p_0)=\operatorname{proj}_2\tilde{\Phi}_h^T(q_0,p_0)\eqsp,
    \end{equation}
    where the potential related to $\tilde{\Phi}_h^T$ is $U$ defined in \Cref{hyp:gaussian_perturbation_nice}.
\begin{lemma}
    \label{lemma:perturbation_point}
    Assume \Cref{hyp:gaussian_perturbation_nice} and let $q_0\in \Rset^d$.
    For any $T\in \Zset$
    there exists a constant $C(T,h,\Sigma,q_0,A_5)>0$ such that for any $p_0\in\Rset^d$,
    \begin{equation}
        \label{eq:torecur_qtilde}
        |\tilde{q}_T(q_0,p_0)-q_T(q_0,p_0)|\leq C(T,h,\Sigma,q_0,A_5) (|p_0|^{\rho-1}+1)\eqsp ,
    \end{equation}

\end{lemma}
\begin{proof}[Proof of \Cref{lemma:perturbation_point}]
    Let $q_0\in \Rset^d$.
    The result is proved by a strong recurrence on $\Nset$ and extended to $\Zset$ by symmetry properties of $\Phiverlet[h][\cdot] $.
    As a prelimenary remark, observe that for any $i\in\Zset^* $ there exists $c(T,q_0,\Sigma,h)>0$ such that for any $p_0\in \Rset^d$,
    \begin{equation}
        \label{eq:pre_remark_point}
    |\nabla \tilde{U}(q_i)|\leq A_5\left(1+|q_i|^{\varrho-1}\right) \leq c (|p_0|^{\rho-1}+1)
    \end{equation}
     by using \cite[Lemma 17]{Durmus2017-tf} and \Cref{hyp:gaussian_perturbation_nice} .
    For any $T\in \Nset$, \eqref{eq:torecur_qtilde} is the hypothesis of reccurence $H(T)$.
    $H(0)$ is straightforward since $\tilde{q}_0=q_0$ and $H(1)$ as well since for any $q_0,p_0\in (\Rset^d)^2$, by \eqref{eq:iteration_verlet},
    \begin{align}
        &q_1=q_0+hp_0-\frac{h^2}{2} \nabla U_{\Sigma}(q_0)\eqsp ,\\
        &\tilde{q}_1=q_0+hp_0-\frac{h^2}{2} \nabla U(q_0)\eqsp , \\
        &|\tilde{q}_1-q_1|\leq \frac{h^2}{2} |\nabla \tilde{U}(q_0) | \eqsp .
    \end{align}
    Let $T\in \Nset^*\setminus\{1\}$ such that $H(1),\ldots,H(T-1)$ and prove $H(T)$. We recall that for any $q_0,p_0 \in (\Rset^d)^2$,
    \begin{equation}
        \label{eq:q_develop_nice}
        \tilde{q}_{T}=q_0+ Thp_0-T\frac{h^2}{2}\nabla U(q_0)-h^2 \tilde{\Theta}_{h,T}^{(1)}(q_0,p_0)\eqsp ,
    \end{equation}
    \begin{equation}
        q_{T}=q_0+ Thp_0-T\frac{h^2}{2}\nabla U_\Sigma(q_0)-h^2 \Theta_{h,T}^{(1)}(q_0,p_0)\eqsp ,
    \end{equation}
    where for any $q,p \in (\Rset^d)^2$,
    \begin{equation}
        \tilde{\Theta}_{h,T}^{(1)}(q,p)=\sum_{i=1}^{T-1}(T-i)\nabla U(\tilde{q_i}(q,p))\eqsp ,
    \end{equation} 
    \begin{equation}
        \Theta_{h,T}^{(1)}(q,p)=\sum_{i=1}^{T-1}(T-i)\nabla U_\Sigma(q_i(q,p)) \eqsp .
    \end{equation}
     Then for any $q_0,p_0 \in (\Rset^d)^2$, 
    \begin{equation}
        |\tilde{q}_{T}-q_T |\leq T\frac{h^2}{2}|\nabla \tilde{U}(q_0)|+h^2 |\Theta_{h,T}^{(1)}(q_0,p_0)-\tilde{\Theta}_{h,T}^{(1)}(q_0,p_0) |\eqsp ,
    \end{equation}
    \begin{equation}
        \begin{multlined}|\Theta_{h,T}^{(1)}(q_0,p_0)-\tilde{\Theta}_{h,T}^{(1)}(q_0,p_0) |\leq \sum_{i=1}^{T-1}(T-i)
             [|\Sigma(q_i(q_0,p_0)-\tilde{q}_i(q_0,p_0)) |+|\nabla \tilde{U}(\tilde{q}_i(q_0,p_0)) |]\eqsp .
        \end{multlined}
    \end{equation}
    Remarking that $T\frac{h^2}{2}|\nabla \tilde{U}(q_0)|$ is a constant regarding $p_0$ and combining the two previous inequalities with \eqref{eq:pre_remark_point} and applying $H(i)$ for any $ i\in [T-1]$ gives $H(T)$ since $|\Sigma|<\infty$.

    For any $T\in \Nset$, $H(T)$ holds. The result extends to $T\in \Zset$, indeed
    denoting by $\Phiverlet[h][-T](q_0,p_0)=q_{-T},p_{-T}$ and $ \Phiverlet[h][T](q_0,-p_0)=q_{T}',p_{T}' $ for any $T\in \Nset^* $ and $q_0,p_0\in (\Rset^d)^2 $,
    applying the expression of the leapfrog scheme \eqref{eq:iteration_verlet} yields:
    \begin{equation}
        q_{-T}=q_{T}'\,,\ee p_{-T}=-p_{-T}'\eqsp .
    \end{equation}
    This completes the proof.
\end{proof}

    Assume \Cref{hyp:gaussian_perturbation_nice}, there exists $\msh_0$ such that \Cref{lemma:gaussian_case} applies.
    We prove that every compact is accessible and 1-small for any $h\in \Rset_{>0} \setminus \msh_0$ for the HMC kernel $\KkerH$, which implies that the Lebesgue measure is an irreducibility measure.

    Let $\Khmc\in \Nset^*$, $h\in \msh_0$ and $q_0\in \Rset^d$.
    By the continuity of the map $(q, p)\in (\Rset^d)^2 \mapsto \alpha_{h,\Khmc}(q,p) \in \Rset_{>0}$ under \Cref{hyp:gaussian_perturbation_nice}, for any $r > 0$ and $\tilde{r} > 0$, we have $\inf_{(q, p) \in \mathrm{B}(0_d, r) \times \mathrm{B}(0_d, \tilde{r})} \alpha_{h,\Khmc}(q,p) = \varepsilon_1(r,\tilde{r}) > 0$.
    Let $r>0$ be fixed.
    Thus for any $q_0 \in \bar{\mathrm{B}}(0_d, r),\, \msa \in \mathcal{B}$ and $\tilde{r}>0$, we can bound from below the transition kernel \eqref{eq:transition-kernel} by 
    \begin{align}
        \label{eq:minor}
        \KkerH(q_0, \msa)
        & \geq
        \int \dd p \, \mathbbm{1}_{\mathrm{B}(0_d, \tilde{r})}(p) \alpha_{h,T}(q_0,p) \GaussStandard(p)\mathbbm{1}_\msa(\operatorname{proj}_1(\Phiverlet[h][\Khmc](q_0, p)))\\
        & \geq  \min_{p\in \bar{\mathrm{B}}(0_d, \tilde{r})}\{\GaussStandard(p)\}  \varepsilon_1(r,\tilde{r}) \int \dd p \, \mathbbm{1}_{\mathrm{B}(0_d, \tilde{r})}(p) \mathbbm{1}_\msa(\operatorname{proj}_1(\Phiverlet[h][\Khmc](q_0, p)))\eqsp .
    \end{align}
    To bound the last term, we use the following Lemma.
   
    \begin{lemma}
        \label{lemma:ball_image_central_hmc}
    Let $q_0\in \Rset^d$ and $\Khmc\in \Nset_{>0}$. If for any $q\in \Rset^d$ the function $\psi_q: p\in \Rset^d\to \operatorname{proj}_1(\Phiverlet[h][\Khmc](q, p))$ is Lipschitz with a Lipschitz constant denoted by $L_\psi$,
     and that there exist $M, \tilde{r},r >0$  such that for any $q_0\in \bar{\mathrm{B}}(0_d,r)$, we have
     \begin{equation}\label{eq:ball-image-ball_hmc}
         \mathrm{B}(0_d, M) \subset \psi_{q_0}(\mathrm{B}(0_d, \tilde{r}))\eqsp .
     \end{equation}
     Then, for any $q_0\in \bar{\mathrm{B}}(0_d,r)$,
     \begin{equation}
        \int \dd p \, \mathbbm{1}_{\mathrm{B}(0_d, \tilde{r})}(p) \mathbbm{1}_\msa(\psi_{q_0}(p)) \geq L_{\psi}^{-d}\Leb (\mathrm{B}(0_d,M)\cap \msa) \eqsp .
     \end{equation}
    \end{lemma}
    \begin{proof}
        This Lemma is a direct application of \cite[Proposition 11]{Durmus2017-tf} for different choice of $\Theta$, with their notations $\Theta=\psi_{q_0}$ for any $q_0\in \bar{\mathrm{B}}(0_d,r)$ and $\mathrm{B}(\tilde{y}_0, \tilde{r})=\mathrm{B}(0_d, \tilde{M})$ .
        They define the measure $\lambda_\Theta$ on $(\Rset^d, \msb(\Rset^d))$ by setting for any $\msa \in \msb(\Rset^d)$,
        $$\lambda_\Theta(\msa)= \int \dd p \, \mathbbm{1}_{\mathrm{B}(0_d, \tilde{r})}(p) \mathbbm{1}_\msa(\psi_{q_0}(p))=\Leb(\psi_{q_0}^{-1}(\msa)\cap \mathrm{B}(0_d, \tilde{M})) \eqsp .$$
    \end{proof}
     If we show the assumptions of \Cref{lemma:ball_image_central_hmc}, we have for any $q\in \bar{\mathrm{B}}(0_d,r)$
    $$\KkerH(q, \msa)\geq \varepsilon_2  \varepsilon_1 L_{\psi}^{-d}\Leb (\mathrm{B}(0_d,M)\cap \msa) \eqsp ,$$
    which gives that $ \bar{\mathrm{B}}(0_d,r)$ is a 1-small set. $r$ being arbitrary, every compact set $K\subset \Rset^d$ are 1-small.

    The stated uniform Lipschitz continuity of $\psi_q$ follows from \cite[Lemma 17]{Durmus2017-tf} and \Cref{hyp:gaussian_perturbation_nice}.
    To show the missing condition \eqref{eq:ball-image-ball_hmc},
    we use the following lemma.
    \begin{lemma}
        \label{lemma:antecedents}
        Assume \Cref{hyp:gaussian_perturbation_nice}.
         Then, for any $M'>0 $ and $q' \in \Rset^d$, there exists $\tilde{r}>0$ such that 
        $\mathrm{B}(0,M')\subset \psi_{q'}(\mathrm{B}(0,\tilde{r}))$.
    \end{lemma} 
   
    \begin{proof}
        In this proof, we use the degree theory.
         The definition of the degree denoted by $\operatorname{deg}$ in the following, can be found in \cite[Proposition and Definition 1.1, 2.1]{degreeouterelo2009mapping} or \cite[p.20 and p.21]{Durmus2017-tf} .

        Let $M'>0$ and $q'\in \Rset^d$ be fixed.
        For any $t\in [0,1]$,
        we denote by
         \begin{equation}
            \psi_{q'}^t: p\in \Rset^d \mapsto \operatorname{proj}_1 [\Phiverlet[h][\Khmc]]_t(q',p)\eqsp ,
        \end{equation}
        where $[\Phiverlet[h][\Khmc]]_t$ is the map related to the potential
        \begin{equation}
            U_t: q\in \Rset^d \mapsto q^\top \Sigma q/2+t \tilde{U}\eqsp .
        \end{equation} 
        Remark that $U_1=U$ and $\psi_{q'}^1=\psi_{q'}$.
         Using $h\in \Rset_{>0} \setminus \msh_0$ and \Cref{lemma:gaussian_case}, $\psi_{q'}^0$ is a linear one-to-one map and thus 
         $(\psi_{q'}^0)^{-1}(\mathrm{B}(0,M'))$ is well defined, open and bounded.
          Therefore, there exists $r_1>0$ such that $\mathrm{B}(0,M')\subset \psi_{q'}^0(\mathrm{B}(0,r_1))$.

          We aim to transfer this property to $\psi_{q'}^1$ using \cite[Proposition 2.4]{degreeouterelo2009mapping}. 
          If there exists $r_2>r_1$ such that for any $t\in [0,1]$ and $p\in \Rset^d$ with $|p|\geq r_2$,
           we have 
           \begin{equation}
            \label{eq:r_2_to_prove}
            |\psi_{q'}^t(p) |\geq M'+1 \eqsp ,
           \end{equation}
              then, defining the homotopy
          \begin{equation}
            H:\,t,p\in [0,1]\times \bar{\mathrm{B}}(0,r_2)\mapsto \psi_{q'}^t(p) \eqsp ,
          \end{equation}
          we have,
          \begin{equation}
            \label{eq:border_control}
            \mathrm{B}(0,M') \subset \Rset^d \setminus H([0,1],\partial\bar{\mathrm{B}}(0,r_2))\eqsp.
          \end{equation}

          If \eqref{eq:r_2_to_prove} holds, using \cite[Proposition 2.4]{degreeouterelo2009mapping} with $H$, we have for any $z\in\mathrm{B}(0,M')$, 
          \begin{equation}
            \operatorname{deg}(z,\bar{\mathrm{B}}(0,r_2),\psi_{q'}^1)=\operatorname{deg}(z,\bar{\mathrm{B}}(0,r_2),\psi_{q'}^0)\eqsp .
          \end{equation}
          Moreover, we have for any $z\in\mathrm{B}(0,M')$, $\operatorname{deg}(z,\bar{\mathrm{B}}(0,r_2),\psi_{q'}^0)\neq 0$ since
           \begin{equation}
            (\psi_{q'}^0)^{-1}(\mathrm{B}(0,M'))\subset \bar{\mathrm{B}}(0,r_1 )\subset \bar{\mathrm{B}}(0,r_2) \eqsp .
           \end{equation}
           This implies that for any $z\in\mathrm{B}(0,M')$, $\operatorname{deg}(z,\bar{\mathrm{B}}(0,r_2),\psi_{q'}^1)\neq 0$ and thus applying \cite[Corollary 2.5, Chapter IV]{degreeouterelo2009mapping}, for any $z\in\mathrm{B}(0,M')$ there exists $y\in \bar{\mathrm{B}}(0,r_2)$ such that $z=\psi_{q'}^1(y) $.
           Therefore, by setting $\tilde{r}=r_2+1$, under \eqref{eq:r_2_to_prove}, we have,
          \begin{equation}
            \mathrm{B}(0,M')\subset \psi_{q'}^1(\mathrm{B}(0,\tilde{r}))=\psi_{q'}^1(\mathrm{B}(0,\tilde{r}))\eqsp .
          \end{equation}

          We prove now that we can choose $r_2>r_1$ such that \eqref{eq:r_2_to_prove}.
           First, we apply \Cref{lemma:perturbation_point} to the trajectory generated by the potential $U_t$ for any $t\in [0,1]$ and we notice that $U_t$ verifies \Cref{hyp:gaussian_perturbation_nice}
          with the same constant $A_5$ taken for $U_1$. There exists a constant $C_1(\Khmc,h,\Sigma,q',A_5)>0$ such that for any $p\in \Rset^d$ and $t\in[0,1]$,
          \begin{equation}
            |\psi_{q'}^t(p) -\psi_{q'}^0(p) | \leq C_1 (|p|^{\varrho-1}+1)\eqsp.
          \end{equation}
           Second, using \Cref{lemma:gaussian_case}, $p\in \Rset^d \mapsto \psi_{q'}^0(p)$ is a linear one-to-one map, thus there exists a constant $C_2(\Khmc,h,\Sigma,q',A_5)>0$ such that for any $p\in \Rset^d$,
           \begin{equation}
            |\psi_{q'}^0(p) | \geq C_2 |p| \eqsp .
          \end{equation}
        
           Finally, for any $p\in \Rset^d$ and $t\in[0,1]$,
          \begin{align}
            |\psi_{q'}^t(p) |&\geq |\psi_{q'}^0(p) |-|\psi_{q'}^t(p) -\psi_{q'}^0(p) |\\
            &\geq C_2 |p|-C_1 (|p|^{\varrho-1}+1)\eqsp .
          \end{align}
           Therefore, using $\varrho-1<1$ yields \eqref{eq:r_2_to_prove}, which completes the proof.
         
    \end{proof}
    
    By \cite[Lemma 17]{Durmus2017-tf}, the map $q\in \Rset^d \mapsto \psi_q(p)$ is Lipschitz for any $p\in \Rset^d$ with a Lipschitz constant $L_{\psi,2}$ independent of $p$.
    Applying \Cref{lemma:antecedents} with $M' = 2r L_{\psi,2}=2M>0$ and $q'=0_d$, there exists $\tilde{r}>0$ such that $\mathrm{B}(0,2M)\subset \psi_{0_d}(\mathrm{B}(0,\tilde{r}))$.
    Therefore, for any $q\in \bar{\mathrm{B}}(0,r)$, we have $\mathrm{B}(0,M)\subset \psi_q(\mathrm{B}(0,\tilde{r}))$.
      Therefore, all compact sets are 1-small for $\KkerH$.

     Moreover, \Cref{lemma:antecedents} and \eqref{eq:minor} imply that for any compact set $K$ with $\Leb(K)>0$ and any $q'\in \Rset^d$, we have $\KkerH(q',K)>0 $.
     Using the regularity of the Lesbegue measure, for any $\msa\in \mcb(\Rset^d) $ with $\pi(\msa)>0$ or equivalently $\Leb(\msa)>0$, we have $\KkerH(q,\msa)>0$ for any $q\in \Rset^d$.
      This implies the accessibility property in one step and thus the aperiodicity.
     Hence, the Lebesgue measure is an irreducibility measure for $\KkerH$ and $\KkerH$ is aperiodic. 
    Therefore, we have \Cref{thm:ergodicity_compile_HMC}-\ref{thm:item_i_ergo1_HMC}.
     The last statement then follows from \cite[Theorem 13.3.4]{markovchainmeyn2012markov} since $\KkerU_h$ is positive recurrent with invariant probability $\pi$ (by construction).

     Now, we prove the Theorem under \Cref{hyp:regularity} and \eqref{eq:transitivity-condition}. In the previous proof, \Cref{lemma:antecedents} is the only part where the assumption \Cref{hyp:gaussian_perturbation_nice} is involved.
     The arguments of the previous proof apply by changing \Cref{lemma:antecedents} with the following Lemma:
     \begin{lemma}
        Assume \Cref{hyp:regularity} and \eqref{eq:transitivity-condition}.
         Then, for any $M'>0 $ and $q' \in \Rset^d$, there exists $\tilde{r}>0$ such that 
        $\mathrm{B}(0,M')\subset \psi_{q'}(\mathrm{B}(0,\tilde{r}))$.
     \end{lemma}
     \begin{proof}
        Let $M'>0$ and $q'\in \Rset^d$. By \Cref{thm:trajectory-transitivity}, $\psi_{q'}$ is an homeomorphism.
        Thus, $\psi_{q'}^{-1}(\mathrm{B}(0,M'))$ is well defined open and bounded.
         Therefore, there exists $\tilde{r}>0$ such that 
        $\mathrm{B}(0,M')\subset \psi_{q'}(\mathrm{B}(0,\tilde{r}))$.
     \end{proof}
     The proof is complete.

  \section{The practical implementation.}
  \label{equivalence_section}

  \begin{algorithm}[H]
    \caption{
        Practical implementation for the no U-turn HMC proposal
    }
    \label{alg:practical}
    \begin{algorithmic}[1]
        \Statex \textbf{Input}
            Initial position and momentum $(q_0, p_0) \in (\Rset^d)^2$,
            maximum tree height $\Kmax$, the inverse mass matrix $\mathbf{M}^{-1}$.
        \State $q^-_0,p^-_0 \gets (q_0, p_0)$, $q^+_0,p^+_0 \gets (q_0,p_0)$, $k\gets 0$, $s_0\gets 1$, $\Sigma(\pi)_0\gets \tpi(q_0,p_0)$
      
        \While {$s_k=1$ and $k < \Kmax$}
            \State $V_k \sim \operatorname{Ber}(\tfrac{1}{2})$,\ee $q_{k+1}\gets q_k$
            \If {$V_k=0$}
              \State $q^-_{k+1},p^-_{k+1},\_,\_,q'_k,\Sigma(\pi)'_k,s'_k\gets $ BuildTree($q^-_{k},p^-_{k},-1,k$)
            \Else
              \State $\_,\_, q_{k+1}^+ ,q_{k+1}^{+} ,q'_k,\Sigma(\pi)_k',s'_k\gets $BuildTree($q^+_{k},p^+_{k},1,k$)
            \EndIf
            \If {$s'_k=1$}
              \State with probability $\min\left(1,\Sigma(\pi)'_k/\Sigma(\pi)_k\right)$, set $ q_{k+1}\gets q'_k$
    \EndIf
    
  \State $\Sigma(\pi)_{k+1}\gets \Sigma(\pi)_k+\Sigma(\pi)_k'$
  \State $s_{k+1}\gets s'_k  \mathbbm{1}[(q^+_{k+1} - q^-_{k+1}).p^-_{k+1}\geq 0]  \mathbbm{1}[(q^+_{k+1} -q^-_{k+1}).p^+_{k+1}\geq 0]$,
    \State $k\gets k+1$
        \EndWhile
    \State \textbf{return} $q_k$

        \State \textbf{function} BuildTree($q,p,v,k$)
        \If {$k=0$}
        \State $q',p' \gets$ Leapfrog($q,p,vh,\mathbf{M}^{-1}$)
  \State $\Sigma(\pi)'\gets \tpi(q',p'), \ee s'\gets 1$
  \State \textbf{return} $q',p',q',p',q',\Sigma(\pi)',s'$
  
        \Else
          \State $q^-,p^-,q^+,p^+,q',\Sigma(\pi)',s'\gets$ BuildTree($q,p,v,k-1$)
         
        \If {$s'=1$}
        \If {$v=0$}
              \State $q^-,p^-,\_,\_,q'',\Sigma(\pi)'',s''\gets $ Buildtree($q^-,p^-,v,k-1$)
            \Else
              \State $\_,\_,q^+,p^+,q'',\Sigma(\pi)'',s''\gets $Buildtree($q^+,p^+,v,k-1$)
            \EndIf
            
           \State with probability $\Sigma(\pi)'/\left(\Sigma(\pi)'+\Sigma(\pi)''\right)$, set $ q'\gets q''$
           \State $\Sigma(\pi)'\gets \Sigma(\pi)'+\Sigma(\pi)''$
  \State $s'\gets s''  \mathbbm{1}[(q^+ - q^-).p^-\geq 0]  \mathbbm{1}[(q^+ -q^-).p^+\geq 0]$
  \EndIf
  \State \textbf{return} $q^-,p^-,q^+,p^+,q',\Sigma(\pi),s'$
  \EndIf
    \end{algorithmic}
\end{algorithm}

\begin{algorithm}
  \caption{
      One Leapfrog step.
  }
  \label{alg:leapfrog}
  \begin{algorithmic}[1]
      \Statex \textbf{Input}
          Initial position and momentum $(q_0, p_0) \in (\Rset^d)^2$,
          the stepsize $h\in \Rset^*$, the inverse mass matrix $\mathbf{M}^{-1}$.
      \State $p'\gets p_0-h\nabla U(q_0) /2 $
      \State $q'\gets q_0+h\mathbf{M}^{-1}p' $
      \State $p'\gets p'-h\nabla U(q') /2 $

      \State \textbf{return} $q',p'$
  
  \end{algorithmic}
\end{algorithm}

The practical implementation described here (\Cref{alg:practical}) is equivalent to Algorithm 1 regarding the position output $q_{j_f}$, which is equal to the output $q_k$ in \Cref{alg:practical} when $M=\Idd$.

The inverse mass matrix $\mathbf{M}^{-1}$ is absent in this paper to simplify the presentation since a linear change of variables enables us to recover it.
However, we use it in the pseudocode presented here to allow for a precise comparison between the different variants of the algorithm.

Experts may notice that \Cref{alg:practical} is not precisely the one implemented in Stan, as discussed in the beginning of Section 3; 
these operational details are beyond the theoretical aims of this paper.

\end{document}